\documentclass[11pt,a4paper,english]{article}
\usepackage{geometry}
\usepackage[T1]{fontenc}
\usepackage[utf8]{inputenc}
\usepackage[english]{babel}

\usepackage{lipsum}
\usepackage{url}
\usepackage[hidelinks]{hyperref}

\usepackage[inline]{enumitem}

\usepackage{graphicx}
\usepackage{booktabs}
\usepackage[labelfont=bf]{caption}

\usepackage{amsmath}
\usepackage{mathrsfs}
\usepackage{amssymb}
\usepackage[nointegrals]{wasysym}
\usepackage{amsthm}
\usepackage{mathtools}
\usepackage{arydshln}

\newtheorem*{theorem*}{Theorem}
\newtheorem{proposition}{Proposition}
\newtheorem{lemma}{Lemma}
\theoremstyle{definition}
\newtheorem{remark}{Remark}
\newtheorem*{comment*}{Comment}

\newcommand{\norm}[1]{\left\lVert#1\right\rVert}

\usepackage[autostyle,italian=guillemets]{csquotes}
\usepackage[backend=biber,style=numeric,sorting=nyt,citestyle=numeric-comp]{biblatex}
\title{Relegation-free closed-form perturbation theory and the domain of secular 
motions in the Restricted 3-Body Problem}
\author{Mattia Rossi and Christos Efthymiopoulos\\
	\small{Università degli Studi di Padova}\\
	\small{Dipartimento di Matematica ``Tullio Levi-Civita''}\\
	\small{Via Trieste, 63 - 35121 Padova, Italy}\\
	\small{mrossi@math.unipd.it, cefthym@math.unipd.it}}
\date{\today}

\begin{document}

\maketitle

\begin{abstract}
We propose a closed-form (i.e. without expansion in the orbital eccentricities) scheme for computations in perturbation theory in the restricted three-body problem (R3BP) when the massless particle is in an orbit exterior to the one of the primary perturber.  Starting with a multipole expansion of the barycentric (Jacobi-reduced) Hamiltonian, we carry out a sequence of normalizations in Delaunay variables by Lie series, leading to a secular Hamiltonian model without use of relegation. To this end, we introduce a book-keeping analogous to the one proposed in \cite{cavallari2022closed} for test particle orbits interior to the one of the primary perturber, but here adapted, instead, to the case of exterior orbits.  We give numerical examples of the performance of the method in both the planar circular and the spatial elliptic restricted three-body problem, for parameters pertinent to the Sun-Jupiter system. In particular, we demonstrate the method's accuracy in terms of reproducibility of the orbital elements' variations far from mean-motion resonances. As a basic outcome of the method, we show how, using as criterion the size of the series' remainder, we reach to obtain an accurate semi-analytical estimate of the boundary (in the space of orbital elements) where the secular Hamiltonian model arrived at after eliminating the particle's fast degree of freedom provides a valid approximation of the true dynamics.\\

\end{abstract}

\section{Introduction}
\label{sec:intro}
As opposed to the usual (Laplace-Lagrange) theory, closed-form perturbation theory \cite{palacian2002normal} provides a framework for series calculations in perturbed Keplerian problems without expansions in powers of the bodies' orbital eccentricities. This is mainly motivated by the necessity to construct secular models for sufficiently eccentric orbits, like those of many asteroids, in our solar system, or the planets in extrasolar planetary systems. 

The efficiency of the usual series methods of expansion in the orbital eccentricities is limited by the fact that the inversion of Kepler's equation in powers of the eccentricity converges only up to the so-called Laplace limit $e_L\approx0.66274$ \cite{finch2003mathematical}. Generally, such convergence slows down way before this value (around $e\sim 0.3-0.4$ in many applications). In order to address this issue, closed form perturbation theory aims at solving in `closed-form' the homological equation by which the Lie generating function is computed at every perturbative step (see for example \cite{deprit1969canonical,efthymiopoulos2011canonical}). The process is far from being priceless: a major obstruction appears when the kernel of the homological equations contains addenda beyond the Keplerian terms. The most common such addendum (\cite{palacian2002normal}) is the centrifugal term $-\nu H$, where $\nu$ is the angular frequency in a frame co-rotating with the primary perturber, and $H$ is the Delaunay action equal to the particle's angular momentum in the direction of the axis of rotation. In the case of a planet's orbiter, $\nu$ is equal to the planet's rotation frequency, and the problem appears for all non-axisymmetric terms (tesseral harmonics) of the planet's multipole potential. In the R3BP, instead, $\nu$ represents the mean motion of the primary perturber (e.g. Jupiter in the Sun-Jupiter system), while the problem appears in a similar way after introducing a multipole expansion of the disturbing function in the particle's Hamiltonian. 

An algorithm to overcome the above issue, called the \textit{relegation algorithm}, has been proposed in works by Deprit, Palacia\'n and collaborators \cite{subiela1992teoria, deprit2001relegation,lara2013averaging,ceccaroni2014analytical,sansottera2017rigorous}. Briefly, given a quasi-integrable Hamiltonian $H=H_0+\varepsilon H_1$, where $\varepsilon$ is a small parameter, suppose that $H_0=H_0^{\prime}+H_0^{\prime\prime}$, where, in a domain in phase space we have that $H_0^{\prime}$ yields the dominant contribution to the Hamiltonian flow of $H_0$ versus the $H_0^{\prime\prime}$ term. In usual perturbation theory, we seek to partly normalize the perturbation $H_1$ via a sequence of canonical transformations defined by generating functions $\chi^{(r)}$, $r=1,2,\ldots$ satisfying a homological equation of the form $\{H_0,\chi^{(r)}\}+h_1^{(r)}=0$, where $\{\cdot,\cdot\}$ denotes the Poisson bracket between two functions of the canonical variables and $h_1^{(r)}$ is a term in the Hamiltonian to be normalized. In the relegation technique, we use instead the equation $\{H_0',\chi^{(r)}\}+h_1^{(r)}=0$, i.e., letting only the dominant function $H_0'$ in the kernel of the homological equation. Such a choice stems mostly from motives of algorithmic convenience. For example, identifying $H_0'$ with the Keplerian term (when $\nu$ is small) leads to a homological equation that can be solved in closed form (we set, instead, $H_0'=-\nu H$ when $\nu$ is large).  However, all Poisson brackets of $\chi^{(r)}$ with the part $H_0^{\prime\prime}$ left out of the kernel lead to terms which need to be `relegated', i.e., pushed to normalization in subsequent steps. For reasons explained in detail in \cite{segerman2000analytical}, only a 
finite number or relegation steps can be performed before reaching a point beyond which the scheme generates divergent sequences of terms (see also \cite{sansottera2017rigorous}). This implies that the process necessarily stops after some steps, leading to a finite, albeit possibly quite small remainder. 

Relegation is a technique particularly suitable to the limiting situation of a strongly hierarchical problem, when the integrable part $H_0$ depends on a frequency vector involving $n$ frequencies $\omega=(\omega_1,\ldots,\omega_n)$ out of which one, say $\omega_i$ for some $i$ with $1\leq i\leq n$ is significantly larger in absolute value than the rest. In particular, the harmonics $\cos(k\cdot\varphi)$ in the Hamiltonian whose normalization can be `relegated' should satisfy $|k_i\omega_i|\gg|k_j\omega_j|$, $j=1,\ldots,n$, $j\neq i$, for every integer $k_i,k_j\in\mathbb{Z}\setminus\{0\}$ (assuming also the non-resonant condition $k\cdot\omega\neq0$, $k=(k_1,\ldots,k_n)$). For example, as explained in \cite{segerman2000analytical} in the simple case with $n=2$ and $\omega_2\gg\omega_1$, the generating function $\chi^{(N)}$ produced after $N$ relegation steps contains terms with coefficients growing as a geometric sequence with ratio $k_1\omega_1/k_2\omega_2$. Thus, relagation is limited to those terms for which the above ratio is smaller than unity. This includes most harmonics of low Fourier order in the Hamiltonian perturbation when $\omega_2\gg\omega_1$, but only few when the two frequencies become comparable in size. Hence, by construction, relegation has limited applicability in this latter, non-hierarchical, case. 

Variants of the relegation technique have been discussed in literature to address perturbed Keplerian problems in which the gravitational potential is due to an extended body expanded in spherical harmonics (e.g. \cite{lara2013averaging, mahajan2018exact}). To address the non-hierarchical case, a techique similar to the one of the present paper is discussed in \cite{lara2013averaging}, referring to the averaging of the tesseral harmonics in the case of the Earth's artificial satellites. In the case of the R3BP, instead, Cavallari and Efthymiopoulos \cite{cavallari2022closed} discuss a relegation-free algorithm for the elimination of short-period terms in the particle's Hamiltonian, when the orbit of the particle (e.g. an asteroid) is totally interior to the orbit of the primary perturber (e.g. Jupiter). We are aware of no relegation-free algorithm proposed in literature which addresses, instead, the case when the particle's orbit is exterior to the orbit of the primary perturber. Providing such an algorithm, discussing some of its important differences with past-proposed algorithms, as well as checking its limits of applicability, constitutes the primary goal of our present paper.\\  

The R3BP is defined by the motion of a body $\mathcal{P}$ of negligible mass in the gravitational field of two massive bodies $\mathcal{P}_0$ (the primary or central body) and $\mathcal{P}_1$ (the secondary or primary perturber), which perform a motion $r_1(t)$ either elliptic in the more general version (ER3BP) or circular (CR3BP). The starting point for our analysis in the sequel is the Hamiltonian of the model, obtained after reduction via Jacobi coordinates $(R,P)$.\footnote{In the R3BP problem the Jacobi transformation is implemented when $\norm{R}>\norm{r_1}$. }. Expressing time through the secondary's mean anomaly $M_1=n_1t$, where $n_1$ is the mean motion of the secondary, and canonically conjugating $M_1$ with a dummy action variable $J_1$ allows to express the Hamiltonian as
\begin{equation}
\label{eqn:ham}
\mathcal{H}(R,M_1,P,J_1)=\frac{\norm{P}^2}{2}-\frac{\mathcal{G}m_0}{\norm{R+\mu r_1(M_1)}}-\frac{\mathcal{G}m_1}{\norm{R-(1-\mu)r_1(M_1)}}+n_1J_1\;,
\end{equation}
where $\mathcal{G}$ is the gravitational constant and
\begin{equation*}
\mu=\dfrac{m_1}{m_0+m_1}\in(0,1/2]\;
\end{equation*} 
is the mass parameter;
\begin{equation}
\label{eqn:r1}
r_1(M_1)=a_1\left(\cos E_1(M_1)-e_1,\sqrt{1-e_1^2}\sin E_1(M_1),0\right)
\end{equation}
is the elliptic revolution of $\mathcal{P}_0-\mathcal{P}_1$ around their barycenter with eccentricity $e_1$ and semi-major axis $a_1$, in which the dependence of the system's eccentric anomaly $E_1\in\mathbb{T}=\mathbb{R}/(2\pi\mathbb{Z})$ on the mean anomaly $M_1\in\mathbb{T}$ is given through Kepler's equation according to standard two-body problem setting; $(R=(X,Y,Z), P=(P_X,P_Y,P_Z))\in T^*(\mathbb{R}^3\setminus\{-\mu r_1,(1-\mu)r_1\})$ is the position-momentum couple of $\mathcal{P}$ and the phase space is endowed with standard symplectic form $dP_X\wedge dX+dP_Y\wedge dY+dP_Z\wedge dZ+dJ_1\wedge dM_1$.\\
We make use then of Delaunay elements $(\ell,g,h,L,G,H)$, defined by 
\begin{align}
\label{eqn:Del}
L&=\sqrt{\mathcal{G}m_0a}\;, & \ell&=M\;, \nonumber \\
G&=L\sqrt{1-e^2}\;, & g&=\omega\;,\\
H&=G\cos i\;, & h&=\Omega\;, \nonumber
\end{align}
where $a,e,i,M,\Omega,\omega$ stand for the semi-major axis, the eccentricity, the inclination, the mean anomaly, the longitude of the ascending node, the argument of pericenter of the particle. 

A key ingredient of the method proposed below is the following: similarly as in \cite{cavallari2022closed}, we introduce a book-keeping symbol $\sigma$ with numerical value equal to $1$, whose role is to organize the perturbative scheme so as to successively normalize terms of similar order of smallness, treating together all small quantities of the problem, i.e.,   
\begin{itemize}
	\item[--] the eccentricities $e$, $e_1$ (when $e_1\neq0$),
	\item[--] the mass ratio $\mu$,
	\item[--] the semi-major axis fluctuation $\delta L$ around the mean $L_*$ for a particular particle trajectory.
\end{itemize}
The book-keeping symbol acts by assigning powers $\sigma^1$ and $\sigma^{\nu_1}$, $\sigma^{\nu}$, $\sigma^\nu$ respectively, for non-zero natural numbers $\nu$, $\nu_1$ defined below, to all the terms in the original Hamiltonian as well as in the Hamiltonian produced after every normalization step. Given this baseline, we arrive (in Section \ref{sec:theory}) to the following result: we demonstrate that, for $k_{\mu},k_{\text{mp}}\in\mathbb{N}\setminus\{0\}$ with $k_{\mu}>1$, the combination of expansions of \eqref{eqn:ham} up to $\mu^{k_{\mu}}$ and $(\norm{r_1}/\norm{R})^{k_{\text{mp}}}$ is canonically conjugate by $\nu(k_{\mu}-1)$ near-identity transformations to a secular model, obtained as a normal form with respect to the fast angles $\ell,M_1$
\begin{equation}
\label{eqn:hamnf}
\mathscr{H}(\ell,g,h,M_1,\delta L,G,H,J_1)=\mathscr{H}_0(g,h,\delta L,G,H,J_1)+\mathscr{R}(\ell,g,h,M_1,\delta L,G,H)\;,
\end{equation}
with
\begin{equation}
\label{eqn:H0}
\mathscr{H}_0=n_*\delta L+n_1J_1+\sum_{l=\nu}^{\nu k_{\mu}-1}\sum_{p\in\mathbb{Z}^2}c_{l,p}(\delta L,e,i;\mu,L_*,a_1,e_1)\cos(p_1g+p_2h)\sigma^{l}\;,
\end{equation}
\begin{multline}
\label{eqn:R}
\mathscr{R}=\sum_{s\in\mathbb{Z}^4}d_{\nu k_{\mu},s}(E_1,\delta L,e,i;\mu,L_*,a_1,e_1)\cos(s_1f+s_2g+s_3h+s_4E_1)\sigma^{\nu k_{\mu}}\\
+\mathcal{O}\left(\sigma^{\nu k_{\mu}+1};\left(\frac{\norm{r_1}}{\norm{R}}\right)^{k_{\text{mp}}+1}\right)\;.
\end{multline}
The dependencies $f=f(\ell,\delta L,G)$ for the true anomaly, $e=e(\delta L,G)$ and $i=i(G,H)$ are implied in all the above expressions; $c_{l,p},d_{\nu k_{\mu},s}$ are real coefficients. A crucial point is the way by which the positive integers $\nu=\nu(e_*,\mu)\ge1$, $\nu_1=\nu_1(e_*,e_1)\ge1$ are chosen. As detailed below, these integers, which regulate the book-keeping scheme, are suitably tuned on the basis of a selected reference value $e_*\in(0,1)$: 
\begin{equation}
\label{eqn:nun1}
\nu=\left\lceil\frac{\log_{10}\mu}{\log_{10}e_*}\right\rceil\;,\quad\quad\nu_1=\left\lceil\frac{\log_{10}e_1}{\log_{10}e_*}\right\rceil\;,
\end{equation}
where $\lceil\cdot\rceil$ is the ceiling function. The normalizing scheme leading to \eqref{eqn:hamnf} is local: knowing that the semi-major axis is preserved under the flow of the (secular) normal form, we introduce the splitting $L=L_*+\delta L$, where $L_*=\sqrt{\mathcal{G}m_0a_*}\gg\delta L$, $n_*=\sqrt{\mathcal{G}m_0}a_*^{-3/2}$ is a targeted reference value for the semi-major axis $a_*$, and expand the Hamiltonian in powers of $\delta L$, rendering $\delta L$ the new action variable canonically conjugated to the particle's mean anomaly. 

Given the above, the normalization algorithm provides a sequence of Lie generating functions $\chi^{(j)}_{\nu+j-1}=\mathcal{O}(\sigma^{\nu+j-1})$, $j=1,\ldots,\nu(k_{\mu}-1)$, which yields the Lie canonical transformation allowing to recursively normalize all terms depending on the angles $f$ and $E_1$ in the Hamiltonian. The normalizing trasformations are possible to define for values of the frequencies $n_*$ (mean motion of the particle at the semi-major axis $a_*$) and $n_1$ far from mean-motion resonances (see Remark \ref{rem:Dioph}). Furthermore, the generating functions are computed as solutions of a homological equation of the form 
\begin{equation}
\label{eqn:homeq}
\{\mathscr{Z}_0,\chi_{\nu+j-1}^{(j)}\}+\mathscr{R}_{\nu+j-1,\nu+j-1}^{(j-1)}=\mathcal{O}(\sigma^{\nu+j-1})\;,
\end{equation}
where $\mathscr{Z}_0=n_*\delta L+n_1J_1$ and $\mathscr{R}_{\nu+j-1,\nu+j-1}^{(j-1)}\sim\sigma^{\nu+j-1}$ collects the trigonometric monomials of $\mathcal{O}(\sigma^{\nu+j-1})$ depending on at least one of the two anomalies. The key to obtaining a closed-form solution for \eqref{eqn:homeq} is, precisely, the appropriate choice of a $\mathcal{O}(\sigma^{\nu+j-1})$ remainder left in the second hand of the equation. In words, we do not seek for an exact cancellation of the terms $\mathscr{R}_{\nu+j-1,\nu+j-1}^{(j-1)}$, but only for an approximate cancellation, leading to a remainder, which, however, is of higher order in book-keeping, and, hence, possible to reduce at subsequent steps. 

As discussed in Section \ref{sec:app}, a relevant outcome of the analysis of the behavior of the remainder obtained by the above method stems from an estimation of the optimal number of normalization steps $j_{opt}$, where the remainder becomes of order $\nu+j_{opt}-1$ in the book-keeping parameter, with $j_{opt}\le\nu(k_{\mu}-1)$. 
The value of $j_{opt}$ is defined as the one where the error bound $\mathscr{E}^{(j)}(a_*,e_*)=\sum_{\nu+j\le l\le\nu k_{\mu},s}|d_{l,s}^{(j)}|\ge\|\mathscr{R}^{(j)}_{\nu+j}\|_{\infty}=\sup|\mathscr{R}^{(j)}_{\nu+j}|$ becomes minimum, with $\mathscr{R}^{(j)}_{\nu+j}=\mathcal{O}(\sigma^{\nu+j})$ and $d_{l,s}^{(j)}$ as in $\eqref{eqn:R}$ after $j$ normalization steps. As typical in perturbation theory, the value of $j_{opt}$ depends on the chosen reference values $(a_*,e_*)$. With the present method one can then obtain a map of the size of the optimal remainder as a function of $(a_*,e_*)$ in the semi-plane $a>a_1$. Using this information, we compute the limiting locus uniting all points in $(a_*,e_*)$ such that the normal form computation yields no improvement with increasing number of normalization steps, i.e., where $j_{opt}=1$. Comparing with numerical stability maps obtained with the Fast Lyapunov Indicator (FLI) \cite{lega2016theory}, one sees that, the limiting locus found semi-analytically essentially coincides with the numerical (FLI map) limit where no harmonic in the Hamiltonian associated with one of the exterior mean-motion resonances affects the dynamics. As a consequence, all motions in the sub-domain of the plane $(a_*,e_*)$ below the limiting locus are stable in the \textit{secular} sense, i.e., protected against instabilities caused by short-period resonant effects. For this reason, we identify this locus as the border of the \textit{domain of secular motions}, and substantiate the fact that its semi-analytical computation (through the normal forms) yields results in precise agreement with those found by the heuristic definition of the same border via the fully numerical (FLI) computation of stability maps.\\   

The paper is structured as follows. Section \ref{sec:theory} presents step-by-step the algorithm that gives rise to \eqref{eqn:H0} and \eqref{eqn:R}, supplemented with the formulas for the Poisson algebra in Keplerian elements used in all closed-form computations. Section \ref{sec:app} is devoted to a numerical investigation of the method's accuracy for an asteroid in the Sun-Jupiter system, first in the spatial ER3BP, and then in the planar CR3BP; in the latter case, the computations are short enough to allow for a specification of the optimal normalization order in a grid of values in the $(a_*,e_*)$ plane, leading to the semi-analytical determination of the border of the domain of secular motions. Section \ref{sec:conc} summarizes the basic conclusions of the present study and gives some relevant comments for future work.\\

\section{The closed-form method for the outermost R3BP}
\label{sec:theory}

\subsection{Multipole expansion of the perturbation}
\label{subsec:mpdist}

Referring to section \ref{sec:intro}, let $\mathcal{H}$ be given in barycentric Cartesian coordinates as in \eqref{eqn:ham}:
\begin{equation}
\label{eqn:hamdistf}
\mathcal{H}=\frac{\norm{P}^2}{2}+n_1J_1-\mathcal{G}m_0\mathcal{R}\;,
\end{equation}
Assuming $\norm{r_1}/\norm{R}<1$, we carry out a multipole expansion of the function $\mathcal{R}(R,M_1)$ in powers of the ratio $\norm{r_1}/\norm{R}<1$:
\begin{align}
\label{eqn:distf}
\begin{split}
\mathcal{R}&=\frac{1}{\norm{R+\mu r_1}}+\frac{\mu}{1-\mu}\frac{1}{\norm{R+(1-\mu)r_1}}\\
&=\frac{1}{\norm{R}}\large\Bigg(\sum_{l=0}^{\infty}\binom{-1/2}{l}\left(\frac{2\mu r_1\cdot R}{\norm{R}^2}+\mu^2\left(\frac{\norm{r_1}}{\norm{R}}\right)^2\right)^l\\
&\quad+\frac{\mu}{1-\mu}\sum_{l=0}^{\infty}\binom{-1/2}{l}\left(-\frac{2(1-\mu)r_1\cdot R}{\norm{R}^2}+(1-\mu)^2\left(\frac{\norm{r_1}}{\norm{R}}\right)^2\right)^l\large\Bigg)\\
&=\frac{1}{1-\mu}\frac{1}{\norm{R}}+\mathcal{O}\left(\left(\frac{\norm{r_1}}{\norm{R}}\right)^2\right)\;.
\end{split}
\end{align}
where, for $\beta\in\mathbb{R}$
\begin{equation*}
\binom{\beta}{l}=\frac{\beta(\beta-1)\cdots(\beta-l+1)}{l!}
\end{equation*}
indicates the generalized binomial coefficient (equal to $1$ for $l=0$).\\
\begin{figure}
	\centering
	\includegraphics[scale=1.05]{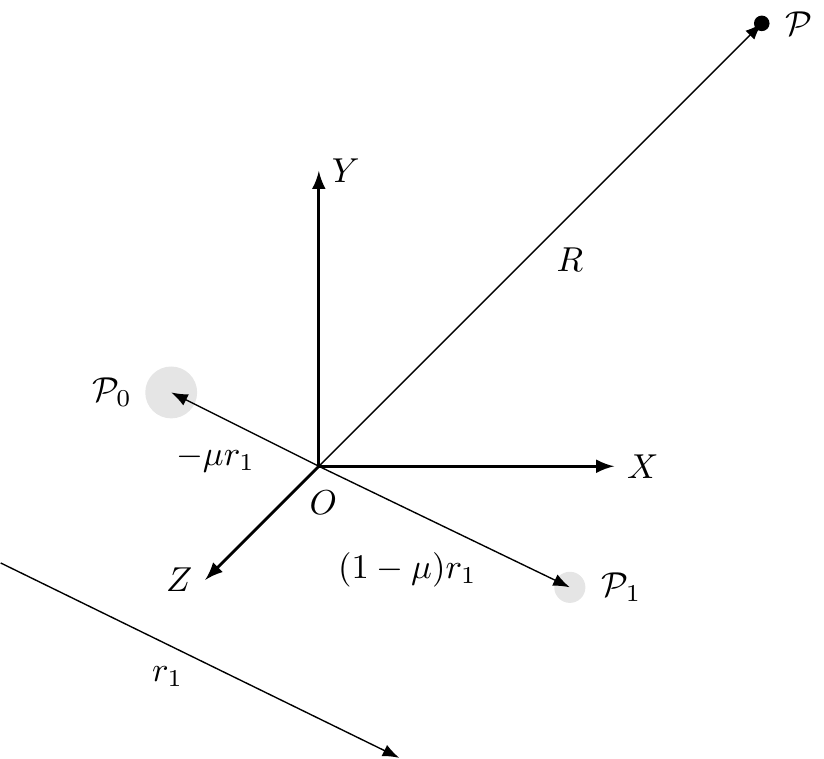}
	\caption{Representation of the R3BP in the barycentric frame (or equivalently in Jacobi variables) with $\norm{R}>\norm{r_1}$.}
	\label{fig:masses}
\end{figure}

\begin{remark}
	\label{rem:noindirect}
	For $l=1$ in Eq.(\ref{eqn:distf}) the coefficients of the dipole term $(r_1\cdot R)/\norm{R}^3$ in the two sums in the r.h.s. of the equation cancel each other exactly. Thus, no dipole term appears in the disturbing function. This is a consequence of the choice of Jacobi coordinates.\\
\end{remark}

\subsection{Canonical form of the Hamiltonian}
\label{subsec:Del}
Performing an extra series expansion in powers of $\mu<1$ yields the standard nearly-integrable form
\begin{equation}
\label{eqn:hammump}
\mathcal{H}=\mathcal{H}_0+\mu\mathcal{H}_1\;,
\end{equation}
where the Keplerian part reads
\begin{equation}
\label{eqn:H0Cart}
\mathcal{H}_0=\frac{\norm{P}^2}{2}-\frac{\mathcal{G}m_0}{\norm{R}}+n_1J_1\;
\end{equation}
and the disturbing function becomes
\begin{multline}
\label{eqn:H1Cart}
\mathcal{H}_1=-\frac{\mathcal{G}m_0}{\norm{R}}\large\Bigg(\sum_{l=0}^{\infty}\mu^l+\sum_{l=1}^{\infty}\mu^{l-1}\binom{-1/2}{l}\left(\frac{2r_1\cdot R}{\norm{R}^2}+\mu\left(\frac{\norm{r_1}}{\norm{R}}\right)^2\right)^l\\
+\sum_{l=1}^{\infty}(1-\mu)^{l-1}\binom{-1/2}{l}\left(-\frac{2 r_1\cdot R}{\norm{R}^2}+(1-\mu)\left(\frac{\norm{r_1}}{\norm{R}}\right)^2\right)^l\large\Bigg)\;.
\end{multline}\\
We now move to Delaunay action-angle variables \eqref{eqn:Del} by replacing into \eqref{eqn:hammump} the relationships
\begin{equation}
\label{eqn:canH0}
\mathcal{H}_0=-\frac{\mathcal{G}m_0}{2a}+n_1J_1\;,
\end{equation}
\begin{equation}
\label{eqn:normR}
\norm{R}=\frac{a(1-e^2)}{1+e\cos f}\;,
\end{equation}
\begin{multline}
\label{eqn:Rdotr1}
r_1\cdot R=a_1\norm{R}\Big((\cos E_1-e_1)\left(\cos h\cos(g+f)-\sin h\sin(g+f)\cos i\right)\\
+\sqrt{1-e_1^2}\sin E_1\left(\sin h\cos(g+f)+\cos h\sin(g+f)\cos i\right)\Big)\;
\end{multline}
as well as \eqref{eqn:r1} for the vector $r_1$. We get
\begin{equation}
\label{eqn:Helem}
\mathcal{H}=-\frac{\mathcal{G}m_0}{2a}+n_1J_1+\mu\mathcal{H}_1(f,g,h,E_1,a,e,i;\mu,a_1,e_1)\;.
\end{equation}

\begin{remark}
	\label{rem:cosfnum}
	Only the square of the norm $\norm{r_1}^2=r_1\cdot r_1$ is required in Eq.(\ref{eqn:H1Cart}), while the norm $\norm{R}$ appears only in the denominator of the above equation, in powers equal to or higher than quadratic. Then equations \eqref{eqn:normR} and \eqref{eqn:r1}, respectively dependent on $f$ and $E_1$, lead to a representation of the disturbing function as a sum of trigonometric polynomials depending on harmonics of the form $\cos(s_1f + s_2g+s_3h+s_4 E_1)$. This is a key ingredient of the closed-form method, i.e., working with the angles $f$ and $E_1$, instead of the mean anomalies $M, M_1$, no series reversion of Kepler's equation is used throughout the whole perturbative scheme.\\
\end{remark}

In order to avoid relegation, our method discussed below works locally, by constructing a model for the secular Hamiltonian valid for a particle's semi-major axis varying as $a=a_*+\delta a(t)$, i.e., by a small quantity $\delta L$ around some reference value $a_*$. By standard secular theory, we have the estimate $\delta a=\mathcal{O}(\mu)$ far from mean-motion resonances. Formally, introducing the new canonical variable $\delta L$ as
\begin{equation}
\label{eqn:splitL}
L=L_*+\delta L=\sqrt{\mathcal{G}m_0a_*}+\frac12\sqrt{\frac{\mathcal{G}m_0}{a_*}}\delta a+\mathcal{O}(\delta a^2)\;.
\end{equation}
and expanding the Hamiltonian in powers of the quantity $\delta L$ around $L_*$, we obtain
\begin{align}
\begin{split}
\label{eqn:hamDelexp}
\mathcal{H}&=-\frac{\mathcal{G}^2m_0^2}{2L_*^2}\sum_{l=0}^{\infty}\binom{-2}{l}\left(\frac{\delta L}{L_*}\right)^{l}+n_1J_1+\mu\sum_{l=0}^{\infty}\frac{1}{l!}\left.\frac{\partial^l\mathcal{H}_1}{\partial L^l}\right\vert_{L=L_*}\delta L^l\\
&=n_*\delta L+n_1J_1+\mu\left(\left.\mathcal{H}_1\right\vert_{\delta L=0,\text{ }\mu=0}+\left.\frac{\partial\mathcal{H}_1}{\partial \delta L}\right\vert_{\delta L=0,\text{ }\mu=0}\delta L\right)+\mathcal{O}(\mu^2,\delta L^2)\;,
\end{split}
\end{align}
where a constant term $-\mathcal{G}^2m_0^2/(2L_*^2)$ was dropped from the expansion. The constant $n_*=\mathcal{G}^2m_0^2/L_*^3$ is equal to the particle's mean motion under Keplerian orbit at the semi-major axis $a_*$.

\begin{remark}
	\label{rem:Dioph}
	The choice of the reference value $a_*$ determines the kind of divisors appearing in the normalization procedure. In the present paper, we deal only with the `non-resonant' case, in which the frequencies $n_*$ and $n_1$ satisfy no-commensurability condition. For example, to be far from any resonance we may require that $n_*$ and $n_1$ satisfy a diophantine condition 
	\begin{equation}
	\label{eqn:Dioph}
	|k_*n_*+k_1n_1|>\frac{\gamma}{|k|^{\tau}}\;,\quad\forall k=(k_*,k_1)\in\mathbb{Z}^2\setminus\{0\}
	\end{equation}
	with $|k|=|k_*|+|k_1|$ and some suitable $\gamma>0$, $\tau>1$. \\
	However, the algorithm presented below can be readily extended to cases of mean-motion resonance. We leave the details for a future work, noting only that in resonant cases we have the estimate $\delta L=\mathcal{O}(\mu^{1/2})$, instead of $\mathcal{O}(\mu)$. The effect of approaching close to a mean-motion resonance with the present series is seen, instead, as a rise in the value of the series' remainder, caused by (non-zero) small divisors in the series (as visible, for example, in Fig. \ref{fig:err2d} discussed in section \ref{sec:app} below).\\
\end{remark}

\subsection{Poisson structure and book-keeping}
\label{subsec:Poiss}

\subsubsection{Poisson bracket formulas}
\label{subsubsec:poissform}
All steps of closed-form perturbation theory involve Poisson brackets between differentiable functions of the form $F(\ell,g,h,M_1,\delta L,G,H,J_1)\in\mathcal{C}^{\infty}(\mathbb{T}^4\times D)$, $D\subset\mathbb{R}^4$ being an open set, whose dependence on the variables $\ell$, $M_1$, $G$ and $H$ is given in implicit form through the functions $f(\ell,\delta L,G)$, $E_1(M_1,e(\delta L,G))$, $e(\delta L,G)$, $\iota_c(G,H)=\cos i(G,H)$, $\iota_s(G,H)=\sin i(G,H)$, $\eta(\delta L,G)=\sqrt{1-e(\delta L,G)^2}$, $\norm{r_1}(M_1)=a_1(1-e_1\cos E_1(M_1))$, and $\phi_1(M_1)=E_1(M_1,e(\delta L,G))-M_1$. The Poisson bracket between two functions $F_1,F_2$ of the above form is computed by the formulas
\begin{align}
	\begin{split}
	\label{eqn:poisson}
\left\{F_1,F_2\right\}&=
{\text{d}F_1\over\text{d}\ell}{\text{d}F_2\over\text{d}\delta L}  
+
{\text{d}F_1\over\text{d}g}{\text{d}F_2\over\text{d}G}  
+
{\text{d}F_1\over\text{d}h}{\text{d}F_2\over\text{d}H}  
+
{\text{d}F_1\over\text{d}M_1}{\text{d}F_2\over\text{d}J_1}  \\
&-
{\text{d}F_1\over\text{d}\delta L}{\text{d}F_2\over\text{d}\ell} 
-
{\text{d}F_1\over\text{d}G}{\text{d}F_2\over\text{d}g}  
-
{\text{d}F_1\over\text{d}h}{\text{d}F_2\over\text{d}H}  
-
{\text{d}F_1\over\text{d}J_1}{\text{d}F_2\over\text{d}M_1} 
\end{split}
\end{align}
implemented to the closed-form version of the functions $F_1,F_2$. The closed-form version of a function $F$ is defined as:
\begin{equation}\label{eqn:funcclf}
F=F(f,g,h,E_1,\delta L,e,\eta,\iota_c,\iota_s,J_1)\;. 
\end{equation}
The derivatives in the canonical variables of a function $F$ as in Eq.(\ref{eqn:poisson}) are computed by the chain rule formulas
\begin{gather}
\label{eqn:dFdl}
\frac{\text{d}F}{\text{d}\ell}=\frac{\partial F}{\partial f}\frac{\partial f}{\partial\ell}\;,\\
\label{eqn:dFdg}
\frac{\text{d}F}{\text{d}g}=\frac{\partial F}{\partial g}\;,\\
\label{eqn:dFdh}
\frac{\text{d}F}{\text{d}h}=\frac{\partial F}{\partial h}\;,\\
\label{eqn:dFdM1}
\frac{\text{d}F}{\text{d}M_1}=\left(\frac{\partial F}{\partial E_1}+\frac{\partial F}{\partial\norm{r_1}}\frac{\text{d}\norm{r_1}}{\text{d}E_1}+\frac{\partial F}{\partial\phi_1}\right)\frac{\text{d}E_1}{\text{d}M_1}-\frac{\partial F}{\partial\phi_1}\;,
\end{gather}
\begin{gather}
\label{eqn:dFddeltaL}
\frac{\text{d}F}{\text{d}\delta L}=\frac{\partial F}{\partial f}\frac{\partial f}{\partial\delta L}+\frac{\partial F}{\partial\delta L}+\frac{\partial F}{\partial e}\frac{\partial e}{\partial\delta L}+\frac{\partial F}{\partial\eta}\frac{\partial\eta}{\partial\delta L}\;,\\
\label{eqn:dFdG}
\frac{\text{d}F}{\text{d}G}=\frac{\partial F}{\partial f}\frac{\partial f}{\partial G}+\frac{\partial F}{\partial e}\frac{\partial e}{\partial G}+\frac{\partial F}{\partial\eta}\frac{\partial\eta}{\partial G}+\frac{\partial F}{\partial\iota_c}\frac{\partial\iota_c}{\partial G}+\frac{\partial F}{\partial\iota_s}\frac{\partial\iota_s}{\partial G}\;,\\
\label{eqn:dFdH}
\frac{\text{d}F}{\text{d}H}=\frac{\partial F}{\partial \iota_c}\frac{\partial \iota_c}{\partial H}+\frac{\partial F}{\partial\iota_s}\frac{\partial\iota_s}{\partial H}\;,\\
\label{eqn:dFdJ1}
\frac{\text{d}F}{\text{d}J_1}=\frac{\partial F}{\partial J_1}\;,
\end{gather}
where 
\begin{equation}
\label{eqn:dfdl}
\frac{\partial f}{\partial\ell}=\frac{(1+e\cos f)^2}{\eta^3}\;,
\end{equation}
\begin{gather}
\label{eqn:dr1dE1}
\frac{\text{d}\norm{r_1}}{\text{d}E_1}=a_1e_1\sin E_1\;,\\
\label{eqn:dE1dM1}
\frac{\text{d}E_1}{\text{d}M_1}=\frac{a_1}{\norm{r_1}}\;,
\end{gather}
\begin{gather}
\label{eqn:dfddeltaL}
\frac{\partial f}{\partial\delta L}=\frac1L\left(\frac{2\sin f}{e}+\frac{\sin(2f)}{2}\right) =\frac{1}{L_*}\left(\frac{2\sin f}{e}+\frac{\sin(2f)}{2}\right)\left(1-\frac{\delta L}{L_*}\right)+\mathcal{O}(\delta L^2)\;,\\
\label{eqn:deddeltaL}
\frac{\partial e}{\partial\delta L}=\frac{\eta^2}{eL}=\frac{\eta^2}{eL_*}\left(1-\frac{\delta L}{L_*}\right)+\mathcal{O}(\delta L^2)\;,\\
\label{eqn:detaddeltaL}
\frac{\partial \eta}{\partial\delta L}=-\frac{\eta}{L}=-\frac{\eta}{L_*}\left(1-\frac{\delta L}{L_*}\right)+\mathcal{O}(\delta L^2)\;,
\end{gather}
\begin{gather}
\label{eqn:dfdG}
\begin{align}
\frac{\partial f}{\partial G}&=-\frac{1}{\eta L}\left(\frac{2\sin f }{e}+\frac{\sin(2f)}{2}\right)\nonumber\\
&=-\frac{1}{\eta L_*}\left(\frac{2\sin f}{e}+\frac{\sin(2f)}{2}\right)\left(1-\frac{\delta L}{L_*}\right)+\mathcal{O}(\delta L^2)\;,
\end{align}\\
\label{eqn:dedG}
\frac{\partial e}{\partial G}=-\frac{\eta}{eL}=-\frac{\eta}{eL_*}\left(1-\frac{\delta L}{L_*}\right)+\mathcal{O}(\delta L^2)\;,\\
\label{eqn:detadG}
\frac{\partial\eta}{\partial G}=\frac1L=\frac{1}{L_*}\left(1-\frac{\delta L}{L_*}\right)+\mathcal{O}(\delta L^2)\;,\\
\label{eqn:diotacdG}
\frac{\partial\iota_c}{\partial G}=-\frac{\iota_c}{\eta L}=-\frac{\iota_c}{\eta L_*}\left(1-\frac{\delta L}{L_*}\right)+\mathcal{O}(\delta L^2)\;,\\
\label{eqn:diotasdG}
\frac{\partial\iota_s}{\partial G}=-\frac{1-\iota_s^2}{\eta L\iota_s}=-\frac{1-\iota_s^2}{\eta L_*\iota_s}\left(1-\frac{\delta L}{L_*}\right)+\mathcal{O}(\delta L^2)\;,
\end{gather}
\begin{gather}
\label{eqn:diotacdH}
\frac{\partial\iota_c}{\partial H}=\frac{1}{\eta L}=\frac{1}{\eta L_*}\left(1-\frac{\delta L}{L_*}\right)+\mathcal{O}(\delta L^2)\;,\\
\label{eqn:diotasdH}
\frac{\partial\iota_s}{\partial H}=-\frac{\iota_c}{\eta L\iota_s}=-\frac{\iota_c}{\eta L_*\iota_s}\left(1-\frac{\delta L}{L_*}\right)+\mathcal{O}(\delta L^2)\;.
\end{gather}
A sketch of the derivation of the above formulas can be found in Appendix \ref{sec:appendixder}. They are strictly valid with $e\in(0,1)$, $i\in(0,\pi)$. However, several cancellations lead to no singular behavior of the Poisson bracket formulas arising throughout the various perturbative steps also when $e=0$ or $i=0$.\\

\subsubsection{Book-keeping: Hamiltonian} 
\label{subsubsec:bkrules}
We introduce in the series a book-keeping symbol $\sigma$ (see \cite{efthymiopoulos2011canonical} for an introduction to the book-keeping technique), with numerical value $\sigma=1$, whose role is to provide a grouping of all the various terms in the series according to their `order of smallness'. Hence, a group of terms with common factor $\sigma^l$, $l\in\mathbb{Z}$, indicates a term considered as of the `$l$-th order of smallness'.

Since in our series there are several small quantities, we introduce a book-keeping scheme allowing to simultaneously deal with all small quantities while maintaining the closed-form character of the series. To this end, we make the following substitutions, called `book-keeping rules', within the initial Hamiltonian:
\begin{itemize}
	\item BK-Rule 1: $e\leadsto \sigma^1 e=\sigma e$ (not applicable to the quantity $e^2$ within $\eta=\sqrt{1-e^2}$),
    \item BK-Rule 2: $\eta\leadsto\sigma^0\eta=\eta$,
	\item BK-Rule 3: $\mu\leadsto \sigma^{\nu}\mu$, with $\nu$ as in Eq.(\ref{eqn:nun1}),
	\item BK-Rule 4: $e_1\leadsto \sigma^{\nu_1}e_1$, with $\nu_1$ as in Eq.(\ref{eqn:nun1}) (not applicable to the quantity $e_1^2$ within $\eta_1\coloneqq\sqrt{1-e_1^2}$),
    \item BK-Rule 5: $\frac{1}{\eta^2}\leadsto\left(\frac{1}{\eta^2}-1\right)\sigma^2+1$,
    \item BK-Rule 6: $\eta_1\leadsto (\eta_1-1)\sigma^{2\nu_1}+1$,
	\item BK-Rule 7: $\delta L^{\lambda}\leadsto \sigma^{l\nu}\delta L^{\lambda}$ with $l=\begin{cases}
	\lambda\;,&\text{if }\delta L^{\lambda}\text{ comes from }\mathcal{H}_1\;,\\
	\lambda-1\;,&\text{if }\delta L^{\lambda}\text{ comes from }\mathcal{H}_0\;,
	\end{cases}\hspace{2mm}\lambda\in\mathbb{N}\setminus\{0\}$.
\end{itemize}
Since $\sigma=1$, the above substitutions affect the structure of the series only at the formal level, and can be substituted directly into the original Hamiltomian, whereby they propagate at subsequent normalization steps once these steps are organized in successive powers $\sigma,\sigma^2$, etc., of the book-keeping symbol. The BK-Rules 1 to 7 above are justified on physical ground as well as on motives of algorithmic convenience. In particular: \\
\\
\noindent
- BK-Rule 1 implies that, despite the use of closed-form formulas, the basic small quantity in powers of which the series are organized is the eccentricity of the test particle. \\
\\
\noindent
- BK-Rule 3 implies that a factor $\mu$ in front of a series term should be treated as of comparable order of smallness as a term of order $e^\nu$, with $\nu$ given by Eq.(\ref{eqn:nun1}). Similarly, BK-Rule 4 implies that a term containing a factor $e_1$ raised to some power should be treated as of comparable order of smallness with a term $e^\nu_1$ raised to the same power. Note that the eccentricity $e$ is a quantity variable in time, so that to compute the exponents $\nu,\nu_1$ we need to use, for any examined trajectory, a reference value $e_*$ yielding an estimate of the overall level of eccentricity all along the orbital evolution for that trajectory. Note that, by standard secular theory we have $e(t)=e_*+\mathcal{O}(\mu)$ if $e_*$ is close to the mean eccentricity (see also discussion at the introduction). Note finally that we obtain exponents $\nu,\nu_1\ge 1$ in the typical case in which $e>\mu$ and $e\ge e_1$. These inequalities arise naturally in the case of small bodies in highly eccentric orbits perturbed by some planet of, say, our solar system, which are the cases of main interest in applying the present method (see, nevertheless, Remark \ref{rem:smalle} on the treatment of cases where the above conditions are not met).\\
\\
\noindent
- BK-Rule 7 stems from the estimate $\delta L=\mathcal{O}(\delta a)=\mathcal{O}(\mu)$ holding for the oscillations in semi-major axis of trajectories far from mean-motion resonances (as already pointed outin the latter case, instead, we have in general $\delta L=\mathcal{O}(\delta a)=\mathcal{O}(\mu^{1/2})$ and the corresponding rule has to be adapted accordingly). The lowering of the book-keeping power by one for within $H_0$ is introduced for reasons of algorithmic convenience, i.e., in order to maintain $n_*\delta L$ in the kernel of the homological equation. \\
\\
\noindent
- BK-Rules 5 and 6 imply just a partition of the unity aiming at keeping the perturbative scheme in closed-form while splitting the corresponding expressions (involving $\eta$ and $\eta_1$ respectively) in two parts, of orders $\mathcal{O}(1)$ and $\mathcal{O}(e^2)$, or $\mathcal{O}(e_1^2)$.\\

\subsubsection{Book-keeping: Poisson structure} 
\label{subsubsec:bkPoissstr}
Some of the formulas in Subsection \ref{subsubsec:poissform} imply differentiation with respect to $e$ through the corresponding partial derivatives in \eqref{eqn:dFddeltaL}, \eqref{eqn:dFdG}, thus yielding a lowering of the power of the eccentricity in some terms arising through Poisson brackets at consecutive steps of perturbation theory. To account for this fact, similarly as in \cite{cavallari2022closed} we introduce the use of the book-keeping symbol $\sigma$ in the formulas of the Poisson algebra as follows: first, we re-write the derivatives with respect to the angles $\ell,g,h,M_1$ as
\begin{gather}
\label{eqn:sigmanu1dFdl}
\frac{\text{d}F}{\text{d}\ell}=\frac{\partial F}{\partial f}\frac{\partial f}{\partial\ell}\frac{a_1(1-e_1\sigma^{\nu_1}\cos E_1)}{\norm{r_1}}\;,\\
\label{eqn:sigmanu1dFdg}
\frac{\text{d}F}{\text{d}g}=\frac{\partial F}{\partial g}\frac{a_1(1-e_1\sigma^{\nu_1}\cos E_1)}{\norm{r_1}}\;,\\
\label{eqn:sigmanu1dFdh}
\frac{\text{d}F}{\text{d}h}=\frac{\partial F}{\partial h}\frac{a_1(1-e_1\sigma^{\nu_1}\cos E_1)}{\norm{r_1}}\;,\\
\label{eqn:sigmanu1dFdM1}
\frac{\text{d}F}{\text{d}M_1}=\left(\frac{\partial F}{\partial E_1}+\frac{\partial F}{\partial\norm{r_1}}\frac{\text{d}\norm{r_1}}{\text{d}E_1}+\frac{\partial F}{\partial\phi_1}\sigma^{-\nu_1}\right)\frac{\text{d}E_1}{\text{d}M_1}-\frac{\partial F}{\partial\phi_1}\sigma^{-\nu_1}\;,
\end{gather}
and with respect to the actions $\delta L,G$ as
\begin{gather}
\label{eqn:sigmadFddeltaL}
\frac{\text{d}F}{\text{d}\delta L}=\frac{\partial F}{\partial f}\frac{\partial f}{\partial\delta L}+\frac{\partial F}{\partial\delta L}+\frac{\partial F}{\partial e}\frac{\partial e}{\partial\delta L}\sigma^{-1}+\frac{\partial F}{\partial\eta}\frac{\partial\eta}{\partial\delta L}\;,\\
\label{eqn:sigmadFdG}
\frac{\text{d}F}{\text{d}G}=\frac{\partial F}{\partial f}\frac{\partial f}{\partial G}+\frac{\partial F}{\partial e}\frac{\partial e}{\partial G}\sigma^{-1}+\frac{\partial F}{\partial\eta}\frac{\partial\eta}{\partial G}+\frac{\partial F}{\partial\iota_c}\frac{\partial\iota_c}{\partial G}+\frac{\partial F}{\partial\iota_s}\frac{\partial\iota_s}{\partial G}\;.
\end{gather}
Note that in \eqref{eqn:sigmanu1dFdM1} use was made of the identity $\phi_1=e_1\sin E_1$ (Kepler's equation). 
Finally, we revise formulas \eqref{eqn:dfdl}, \eqref{eqn:dr1dE1}, \eqref{eqn:dfddeltaL}--\eqref{eqn:diotasdH}, attributing a book-keeping to all factors involving the eccentricity function $\eta$ as 
\begin{equation}
\label{eqn:sigmadfdl}
\frac{\partial f}{\partial\ell}=1+\frac{2e\cos f}{\eta^3}\sigma+\left(\frac{1}{\eta^3}-1+\frac{e^2\cos^2 f}{\eta^3}\right)\sigma^2\;,
\end{equation}
\begin{equation}
\label{eqn:sigmadr1dE1}
\frac{\text{d}\norm{r_1}}{\text{d}E_1}=a_1e_1\sigma^{\nu_1}\sin E_1
\end{equation}
\begin{equation}
\label{eqn:sigmadfddeltaL}
\frac{\partial f}{\partial\delta L}=\frac{1}{L_*}\left(\frac{2\sin f}{e}\sigma^{-1}+\frac{\sin(2f)}{2}\right)+\mathcal{O}(\delta L\sigma^{\nu})\;,
\end{equation}
\begin{equation}
\label{eqn:sigmadeddeltaL}
\frac{\partial e}{\partial\delta L}=\frac{1}{L_*}\left(\frac1e\sigma^{-1}+\frac{\eta^2-1}{e}\sigma\right)+\mathcal{O}(\delta L\sigma^{\nu})\;,
\end{equation}
\begin{equation}
\label{eqn:sigmadetaddeltaL}
\frac{\partial\eta}{\partial\delta L}=-\frac{1}{L_*}\left(1+(\eta-1)\sigma^2\right)+\mathcal{O}(\delta L\sigma^{\nu})\;,
\end{equation}
\begin{multline}
\label{eqn:sigmadfdG}
\hspace{2cm}\frac{\partial f}{\partial G}=-\frac{1}{L_*}\Bigg(\frac{2\sin f}{e}\sigma^{-1}+\frac{\sin(2f)}{2}\\
+\frac{2\sin f}{e}\left(\frac{1}{\eta}-1\right)\sigma+\frac{\sin 2f}{2}\left(\frac{1}{\eta}-1\right)\sigma^2\Bigg)+\mathcal{O}(\delta L\sigma^{\nu})\;,\hspace{2cm}
\end{multline}
\begin{equation}
\label{eqn:sigmadedG}
\frac{\partial e}{\partial G}=-\frac{1}{L_*}\left(\frac1e\sigma^{-1}+\frac{\eta-1}{e}\sigma\right)+\mathcal{O}(\delta L\sigma^{\nu})\;,
\end{equation}
\begin{equation}
\label{eqn:sigmadetadG}
\frac{\partial\eta}{\partial G}=\frac{1}{L_*}+\mathcal{O}(\delta L\sigma^{\nu})\;,
\end{equation}
\begin{equation}
\label{eqn:sigmadiotacdG}
\frac{\partial\iota_c}{\partial G}=-\frac{\iota_c}{L_*}\left(1+\left(\frac{1}{\eta}-1\right)\sigma^2\right)+\mathcal{O}(\delta L\sigma^{\nu})\;,
\end{equation}
\begin{equation}
\label{eqn:sigmaiotasdG}
\frac{\partial\iota_s}{\partial G}=-\frac{1-\iota_s^2}{L_*\iota_s}\left(1+\left(\frac{1}{\eta}-1\right)\sigma^2\right)+\mathcal{O}(\delta L\sigma^{\nu})\;,
\end{equation}
\begin{equation}
\label{eqn:sigmadiotacdH}
\frac{\partial\iota_c}{\partial H}=\frac{1}{L_*}\left(1+\left(\frac{1}{\eta}-1\right)\sigma^2\right)+\mathcal{O}(\delta L\sigma^{\nu})\;,
\end{equation}
\begin{equation}
\label{eqn:sigmadiotasdH}
\frac{\partial\iota_s}{\partial H}=-\frac{\iota_c}{L_*\iota_s}\left(1+\left(\frac{1}{\eta}-1\right)\sigma^2\right)+\mathcal{O}(\delta L\sigma^{\nu})\;.
\end{equation}\\

\begin{remark}
\label{rem:smalle}
The \textit{small eccentricity problem} consists of the fact that the above-proposed book-keeping rules are not applicable in the case $0<e_*\lesssim\mu<e_1$, since, by \eqref{eqn:nun1},  the exponents $\nu$, $\nu_1$ would be smaller than unity. The simple solution of rounding these exponents to $1$, while maintaining the same book-keeping rules as above, fails, since, at any given normalization order $r$, the presence of $\sigma^{-1}$, $\sigma^{-\nu_1}$ terms in the formulas of the Poisson algebra leads to the generation of terms of order \textit{lower than $r$} in the normal form's remainder.  Notwithstanding our focus on a method dealing with large eccentricity orbits (for which the problem does not appear), we discuss below a variant of the main algorithm that deals with trajectories in the case $\nu=1$, i.e., when $e_*\lesssim\mu$. \\
\end{remark}

\subsection{Iterative normalization algorithm}
\label{subsec:normaliz}
\subsubsection{Preliminary step: Hamiltonian preparation}
\label{subsubsec:prelim}
After implementing BK-Rules 1 to 7 the Hamiltonian \eqref{eqn:hamDelexp} resumes the form:
\begin{equation}
\label{eqn:hamfourier}
\mathcal{H}=n_*\delta L+n_1J_1+\sum_{s\in\mathbb{Z}^4}q_s(\delta L,e,\eta,\iota_c,\iota_s;\mu,L_*,a_1,e_1,\eta_1)\cos(s_1f+s_2g+s_3h+s_4E_1)\sigma_s
\end{equation}
where $\sigma_s\in\{\sigma^{\nu},\sigma^{\nu+1},\ldots\}$ and, by D'Alembert rules, only cosines and real coefficients $q_s$ appear (invariance under simultaneous change of sign of all angles). Setting $\mathscr{Z}_0=n_*\delta L+n_1J_1$, for obtaining a closed-form normalization algorithm it turns convenient to re-express the Hamiltonian according to 
\begin{equation}
\label{eqn:trick}
\mathcal{H}=\mathscr{Z}_0+(\mathcal{H}-\mathscr{Z}_0)\frac{a_1(1-e_1\sigma^{\nu_1}\cos E_1)}{\norm{r_1}}\;.
\end{equation}
The Hamiltonian \eqref{eqn:trick} resumes the form:
\begin{equation}
\label{eqn:ham0}
\mathcal{H}=\mathscr{H}^{(0)}=\mathscr{Z}_0+\mathscr{R}^{(0)}_{\nu}\;,
\end{equation}
where
\begin{multline}
\label{eqn:fE1terms}
\mathscr{R}_{\nu}^{(0)}=\sum_{l\ge\nu}\mathscr{R}_{\nu,l}^{(0)}\\
=\sum_{l\ge\nu}\frac{a_1}{\norm{r_1}}\left(\sum_{p\in\mathbb{Z}^2}q^{\prime}_{l,p}\cos(p_1g+p_2h)+\sum_{\substack{s\in\mathbb{Z}^4\\(s_1,s_4)\neq(0,0)}}q^{\prime\prime}_{l,s}\cos(s_1f+s_2g+s_3h+s_4E_1)\right)\sigma^l\;;
\end{multline}
We call $\mathscr{R}_{\nu}^{(0)}$ the \textit{remainder} at the zero-th normalization step (i.e. in the original Hamiltonian). The terms $\mathscr{R}_{\nu,l}^{(0)}$ contain terms of book-keeping order $\sigma^{l}$, with $l\geq\nu$.\\

\subsubsection{Step 1: normalization of the $\boldsymbol{\sigma^{\nu}}$-terms}
\label{subsubsec:sigmanu}
For a suitable generating function $\chi_{\nu}^{(1)}$ to be determined in a while, we introduce the Lie series operator as
\begin{equation}
	\label{eqn:Lieseries}
	\begin{array}{c}
\displaystyle\vspace{1mm}\exp\left(\mathcal{L}_{\chi_{\nu}^{(1)}}\right)\colon\;\mathcal{C}^{\omega}(\mathbb{T}^4\times D)\longrightarrow\mathcal{C}^{\omega}(\mathbb{T}^4\times D)\\
\displaystyle\exp\left(\mathcal{L}_{\chi_{\nu}^{(1)}}\right)=\sum_{n\ge0}\frac{1}{n!}\mathcal{L}_{\chi_{\nu}^{(1)}}^n=\mathbb{I}+\mathcal{L}_{\chi_{\nu}^{(1)}}+\frac12\mathcal{L}_{\chi_{\nu}^{(1)}}\circ\mathcal{L}_{\chi_{\nu}^{(1)}}+\ldots
\end{array}\;,
\end{equation}
where $\mathcal{C}^{\omega}(\mathbb{T}^4\times D)$ denotes the set of real analytic functions in the phase space and
\begin{equation}
\label{eqn:Liederiv}
\mathcal{L}_{\chi_{\nu}^{(1)}}\cdot=\{\cdot,\chi_{\nu}^{(1)}\}
\end{equation}
is the time derivative along the Hamiltonian vector field generated by $\chi_{\nu}^{(1)}$ (Lie derivative).\\
Applying \eqref{eqn:Lieseries} to \eqref{eqn:trick} we get the transformed Hamiltonian
\begin{equation}
\label{eqn:H1}
\mathscr{H}^{(1)}=\mathscr{Z}_0+\mathscr{R}^{(0)}_{\nu}+\{\mathscr{Z}_0,\chi_{\nu}^{(1)}\}+\{\mathscr{R}_{\nu}^{(0)},\chi_{\nu}^{(1)}\}+\frac12\{\{\mathcal{H},\chi_{\nu}^{(1)}\},\chi_{\nu}^{(1)}\}+\ldots\;,
\end{equation}
in which, with the usual abuse of notation, we still indicate with $\ell,g,h,M_1,\delta L,G,H,J_1$ the new canonical variables given by the inverse transformation 
\begin{equation}
\label{eqn:Lieseriesinv}
\exp\left(\mathcal{L}_{\chi_{\nu}^{(1)}}\right)^{-1}=\exp\left(\mathcal{L}_{-\chi_{\nu}^{(1)}}\right)\;.
\end{equation}
\indent Our scope will be to define the Lie generating function $\chi_{\nu}^{(1)}$ in such a way that, after implementing the transformation \eqref{eqn:H1}, $\mathscr{H}^{(1)}$ contains no terms depending on the angles $f$ and $E_1$ at order $\sigma^{\nu}$. The required generating function $\chi_{\nu}^{(1)}$ is computed as an outcome of the following:
\begin{proposition}
	\label{prop:hotstep1}
	Define $\chi_{\nu}^{(1)}$ as
	\begin{multline}
	\label{eqn:chi1}
	\chi_{\nu}^{(1)}=\frac{\phi_1}{n_1}\sigma^{\nu+\nu_1}\sum_{p\in\mathbb{Z}^2}q^{\prime}_{\nu,p}\cos(p_1g+p_2h)\\
	+\sigma^{\nu}\sum_{\substack{s\in\mathbb{Z}^4 \\ (s_1,s_4)\neq(0,0)}}\frac{q^{\prime\prime}_{\nu,s}}{s_1n_*+s_4n_1}\sin(s_1f+s_2g+s_3h+s_4E_1)\;.
	\end{multline}
	Then, it holds that
	\begin{equation}
	\label{eqn:firsthom}
	\{\mathscr{Z}_0,\chi_{\nu}^{(1)}\}+\mathscr{R}^{(0)}_{\nu,\nu}=
	\mathscr{Z}_{\nu}^{(1)}+\mathcal{O}\left(\sigma^{\nu+1}\right)\;,
	\end{equation}
    where
	\begin{equation}
		\label{eqn:Z1nu}
	\mathscr{Z}_{\nu}^{(1)}=\sigma^{\nu}\sum_pq^{\prime}_{\nu,p}\cos(p_1g+p_2h)\;.
	\end{equation}
    Furthermore, the function $\mathscr{H}^{(1)}$ as computed by Eq.(\ref{eqn:H1}) takes the form
    \begin{equation}
    \label{eqn:nf1}
    \mathscr{H}^{(1)}=\mathscr{Z}_0+\mathscr{Z}_{\nu}^{(1)}+\mathscr{R}^{(1)}\;,
    \end{equation}
    where the remainder $\mathscr{R}^{(1)}$ is $\mathcal{O}(\sigma^{\nu+1})$ $\forall\nu\ge 1$ independently of the value of $\nu_1$. 	
\end{proposition}

\begin{proof}
	Setting
	\begin{multline*}
	\label{eqn:chi1form}
	\chi_{\nu}^{(1)}(f,g,h,E_1,\phi_1,\delta L,e,\eta,\iota_c,\iota_s)=\sigma^{\nu}\left(\phi_1\sigma^{\nu_1}\sum_{p\in\mathbb{Z}^2}\hat{q}^{\prime}_{\nu,p}(\delta L,e,\eta,\iota_c,\iota_s)\cos(p_1g+p_2h)\vphantom{\sum_{\substack{s\in\mathbb{Z}^4 \\ (s_1,s_4)\neq(0,0)}}}\right.\\
	\left.+\sum_{\substack{s\in\mathbb{Z}^4 \\ (s_1,s_4)\neq(0,0)}}\hat{q}^{\prime\prime}_{\nu,s}(\delta L,e,\eta,\iota_c,\iota_s)\sin(s_1f+s_2g+s_3h+s_4E_1)\right)\;,
	\end{multline*}
	and recalling the chain rules \eqref{eqn:sigmanu1dFdl}, \eqref{eqn:sigmanu1dFdM1} and \eqref{eqn:sigmadfdl}, \eqref{eqn:sigmadr1dE1}, \eqref{eqn:dE1dM1}, we find
	\begin{multline*}
	\{\mathscr{Z}_0,\chi_{\nu}^{(1)}\}+\mathscr{R}^{(0)}_{\nu,\nu}=-n_*\Bigg(1+\frac{2e\cos f}{\eta^3}\sigma+\left(\frac{1}{\eta^3}-1+\frac{e^2\cos^2 f}{\eta^3}\right)\sigma^2\Bigg)\\
	\frac{a_1(1-e_1\sigma^{\nu_1}\cos E_1)}{\norm{r_1}}\sigma^{\nu}\sum_{(s_1,s_4)\neq(0,0)}s_1\hat{q}^{\prime\prime}_{\nu,s}\cos(s_1f+s_2g+s_3h+s_4E_1)\\
	-n_1\frac{a_1}{\norm{r_1}}\sigma^{\nu}\left(\sum_{(s_1,s_4)\neq(0,0)}s_4\hat{q}^{\prime\prime}_{\nu,s}\cos(s_1f+s_2g+s_3h+s_4E_1)+\sum_p\hat{q}^{\prime}_{\nu,p}\cos(p_1g+p_2h)\right)\\
	+n_1\sigma^{\nu}\sum_p\hat{q}^{\prime}_{\nu,p}\cos(p_1g+p_2h)+\sigma^{\nu}\frac{a_1}{\norm{r_1}}\left(\sum_pq^{\prime}_{\nu,p}\cos(p_1g+p_2h)\vphantom{sum_{(s_1,s_4)\neq(0,0)}}\right.\\
	+\sum_{(s_1,s_4)\neq(0,0)}q^{\prime\prime}_{\nu,s}\cos(s_1f+s_2g+s_3h+s_4E_1)\bigg)\;.
	\end{multline*}
	Requiring that no trigonometric terms depending on $f,E_1$ be present at order $\sigma^{\nu}$ then leads to
	\begin{align*}
	\hat{q}^{\prime\prime}_{\nu,s}&=
	\frac{q^{\prime\prime}_{\nu,s}}{s_1n_*+s_4n_1}\;,\quad s\in\mathbb{Z}^4\colon(s_1,s_4)\neq(0,0)\;,\\
	\hat{q}^{\prime}_{\nu,p}&=\frac{q^{\prime}_{\nu,p}}{n_1}\;,\quad p\in\mathbb{Z}^2\;,  
	\end{align*}
    which implies Eq.(\ref{eqn:chi1}). At order $\sigma^\nu$ we then obtain immediately the formula
	\begin{equation*}
	\mathscr{Z}_{\nu}^{(1)}=\sigma^{\nu}\sum_pq^{\prime}_{\nu,p}\cos(p_1g+p_2h)\;.
	\end{equation*}

	\noindent 
	We now consider the function $\mathscr{H}^{(1)}$ computed by replacing \eqref{eqn:chi1} into \eqref{eqn:H1}. The function $\mathscr{H}^{(1)}$ can be decomposed as in Eq.\eqref{eqn:nf1}. We shall demonstrate that the remainder $\mathscr{R}^{(1)}$ contains no terms of order lower than $\sigma^{\nu+1}$. To this end, it suffices to show that
    \begin{equation}
    \label{eqn:ordgreatnu}
    \{\mathscr{R}_{\nu}^{(0)},\chi_{\nu}^{(1)}\}=\mathcal{O}(\sigma^{2\nu})\;,\quad\quad\frac{1}{n!}\{\ldots\{\{\mathcal{H},\underbrace{\chi_{\nu}^{(1)}\},\chi_{\nu}^{(1)}\},\ldots,\chi_{\nu}^{(1)}}_{n\ge 2}\}=\mathcal{O}(\sigma^{n(\nu-1)+2})\;,
    \end{equation}
    since $n(\nu-1)+2>\nu$, for all $n\ge 2$, $\nu\ge 1$. 

The term $\mathscr{R}_{\nu}^{(0)}$ contains terms of order equal to or larger than $\sigma^\nu$, while $\chi_\nu^{(1)}$ contains only terms of order $\sigma^\nu$. Thus, except for the Poisson bracket $\{\mathscr{Z}_0,\chi_{\nu}^{(1)}\}$, which only contributes to the secular terms $\mathscr{Z}_{\nu}^{(1)}$ due to Eq.(\ref{eqn:firsthom}), the first Poisson bracket in \eqref{eqn:ordgreatnu} contains prefactors of order $\sigma^{2\nu}$ or higher, while the second contains prefactors $\sigma^{n\nu}$ or higher. However, the exponent of $\sigma$ in these brackets can be \textit{lowered} due to the negative powers introduced in the book-keeping formulas in the following three classes of factors: 
\begin{enumerate}[label=(\roman*)]
	\item\label{item:ederiv} partial derivatives with respect to the eccentricity in \eqref{eqn:sigmadFddeltaL}, \eqref{eqn:sigmadFdG} (carrying $\sigma^{-1}$) multiplied by corresponding formulae \eqref{eqn:sigmadeddeltaL}, \eqref{eqn:sigmadedG} (another $\sigma^{-1}$), hence a total of $\sigma^{-2}$;
	\item\label{item:fderiv} differentiations \eqref{eqn:sigmadfddeltaL}, \eqref{eqn:sigmadfdG} involving $f$ (weighting $\sigma^{-1}$) again in \eqref{eqn:sigmadFddeltaL}, \eqref{eqn:sigmadFdG}, thus a pre-factor $\sigma^{-1}$;
	\item\label{item:phi1deriv} partial derivatives with respect to $\phi_1$ in \eqref{eqn:sigmanu1dFdM1} ($\sigma^{-\nu_1}$, $\nu_1\ge1$), thus a prefactor at least $\sigma^{-1}$.
\end{enumerate}

As regards (iii) $\phi_1$ shows up in the numerator of $\chi_{\nu}^{(1)}$ accompanied by a prefactor $\sigma^{\nu+\nu_1}$ (Eq.(\ref{eqn:chi1})), thus the negative powers $\sigma^{-\nu_1}$ are cancelled by the positive powers $\sigma^{\nu_1}$, implying no dependence of the minimum order of the remainder on $\nu_1$. 

As regards \ref{item:ederiv}, we first note that $\chi_{\nu}^{(1)}$ has no explicit dependence on $e$, but only an implicit dependence through $\eta$, which in the closed-form context is treated as an independent symbol. This follows from the fact that $\chi^{(1)}$ stems from balancing 
the coefficients of $\mathscr{R}_{\nu,\nu}^{(0)}$. The latter term contains a pre-factor $\mu$, which is already $\mathcal{O}(\sigma^{\nu})$, thus it cannot contain any further factors produced by any explicit power of $e$. In view of the above, setting $\partial\chi^{(1)}/\partial e=0$, we find that for any $F\in\mathcal{C}^{\infty}(\mathbb{T}^4\times D)$ the expression in $\{F,\chi_{\nu}^{(1)}\}$ pertaining \ref{item:ederiv} can be factored out as
\begin{equation}
\label{eqn:ederivPoiss}
\{F,\chi_{\nu}^{(1)}\}_{\text{\ref{item:ederiv}}}=-\frac{\partial F}{\partial e}\sigma^{-1}\left(\frac{\partial f}{\partial\ell}\frac{\partial e}{\partial\delta L}\frac{\partial\chi_{\nu}^{(1)}}{\partial f}+\frac{\partial e}{\partial G}\frac{\partial\chi_{\nu}^{(1)}}{\partial g}\right)\;.
\end{equation}
We now have the following lemma:
\begin{lemma}
	\label{lemma:dalem}
	For every term in the Hamiltonian \eqref{eqn:trick} of the form
	\begin{equation}
	\label{eqn:indipe}
	q_s(\norm{r_1},\delta L,\eta,\iota_c,\iota_s;\mu,L_*,a_1,e_1,\eta_1)\cos(s_1f+s_2g+s_3h+s_4E_1)\sigma_s\;,
	\end{equation}
	i.e., explicitly independent on $e$, we have $s_1=s_2$.
\end{lemma}
	\begin{proof}
	This is a consequence of D'Alembert rules. Using modified Delaunay angular elements
	\begin{align}
	\label{eqn:modDelang}
	\tilde{\lambda}&=\ell+g+h\;, \nonumber \\
	\tilde{p}&=-g-h\;,\\
	\tilde{q}&=-h\;, \nonumber
	\end{align}
	as well as the formulas $f=\ell+2e\sin\ell+\mathcal{O}(e^2)$, $e\eta(e)^{-2\lambda}=e+\lambda e^3+\mathcal{O}(e^5)$, $\lambda\in\mathbb{N}$, we find that, after expanding in the eccentricity $e$, \eqref{eqn:indipe} should give the terms
	\begin{equation}
	\label{eqn:indipemodDel}
	q_s\cos(s_1(\tilde{\lambda}+\tilde{p})+s_2(\tilde{q}-\tilde{p})-s_3\tilde{q}+s_4E_1)\sigma_s+\mathcal{O}(e)\;.
	\end{equation}
    However, according to the D'Alembert rules, in a generic trigonometric monomial of the form  
	\begin{equation}
	\label{eqn:termmodDel}
	b_w(\norm{r_1},\delta L,\eta,\iota_c,\iota_s;\mu,L_*,a_1,e_1,\eta_1)e^l\sigma^l\cos(w_1\tilde{\lambda}+w_2\tilde{p}+w_3\tilde{q}+w_4E_1)\sigma_w\;,\quad l\in\mathbb{N}\;,
	\end{equation}
	appearing after expanding $\mathcal{H}$ in the eccentricities $e,e_1$, we necessarily have that $l-|w_2|$ must be non-negative and even. Since for any closed-form term in the Hamiltonian, explicitly independent of $e$, the lowermost term in $e$ produced after the expansion satisfies $l=0$, we necessarily have $w_2=0$, that is $s_1=s_2$.
\end{proof}
In view, now, of \eqref{eqn:chi1}, the relation $s_1=s_2$ implies $\partial\chi_{\nu}^{(1)}/\partial f=\partial\chi_{\nu}^{(1)}/\partial g$. Therefore, making use of \eqref{eqn:sigmadfdl}, \eqref{eqn:sigmadeddeltaL} and \eqref{eqn:sigmadedG}, Eq.\eqref{eqn:ederivPoiss} translates into
\begin{equation*}
\{F,\chi_{\nu}^{(1)}\}_{\text{\ref{item:ederiv}}}=-\frac{\partial F}{\partial e}\sigma^{-1}\frac{\partial\chi_{\nu}^{(1)}}{\partial f}\left(\frac{\sigma^{-1}}{L_*e}-\frac{\sigma^{-1}}{L_*e}+\mathcal{O}(\sigma^0)\right)
=-\frac{\partial F}{\partial e}\sigma^{-1}\frac{\partial\chi_{\nu}^{(1)}}{\partial f}\mathcal{O}(\sigma^0)\;.
\end{equation*}
It follows that for any of the functions $F=\mathscr{R}_{\nu}^{(0)},\{\mathcal{H},\chi_{\nu}^{(1)}\},\{\{\mathcal{H},\chi_{\nu}^{(1)}\},\chi^{(1)}_{\nu}\},\ldots$, terms produced by derivatives of the type \ref{item:ederiv} in \eqref{eqn:H1} are subject to a lowering of the exponent of $\sigma$ per Poisson bracket only by a factor $\sigma^{-1}$, instead of $\sigma^{-2}$. In particular, in the case $F=\mathscr{R}_{\nu,\nu}^{(0)}$ (as well as for any other closed-form function explicitly independent on the eccentricity) we have that \eqref{eqn:ederivPoiss} is identically vanishing.

As regards \ref{item:fderiv}, we find that for any $F_1,F_2\in\mathcal{C}^{\infty}(\mathbb{T}^4\times D)$, the derivative $\partial f/\partial\delta L$ (Eq.\eqref{eqn:sigmadfddeltaL}) participates in the Poisson bracket $\{F_1,F_2\}$ only through the combination 
\begin{equation}
\label{eqn:fdeltaLPoiss}
\frac{\partial f}{\partial\ell}\frac{\partial f}{\partial\delta L}\left(\frac{\partial F_1}{\partial f}\frac{\partial F_2}{\partial f}-\frac{\partial F_1}{\partial f}\frac{\partial F_2}{\partial f}\right)=0\;.
\end{equation}
On the other hand, the derivative $\partial f/\partial G$ (Eq.\eqref{eqn:sigmadfdG}) participates in the same Poisson bracket through the combination
\begin{equation}
\label{eqn:fGPoiss}
\frac{\partial f}{\partial G}\left(\frac{\partial F_1}{\partial g}\frac{\partial F_2}{\partial f}-\frac{\partial F_1}{\partial f}\frac{\partial F_2}{\partial g}\right) 
\end{equation}
which, by Lemma \ref{lemma:dalem}, is also equal to zero for $F_1=\mathscr{R}_{\nu,\nu}^{(0)}$ (or any other term $\mathcal{O}(\sigma^{\nu+1})$ in $\mathcal{H}$ not depending explicitly on $e$), and $F_2=\chi_{\nu}^{(1)}$. 

In conclusion, returning to \eqref{eqn:ordgreatnu}, and taking all the above deductions into account, we arrive at the expressions
\begin{equation*}
\{\mathscr{R}_{\nu}^{(0)},\chi_{\nu}^{(1)}\}=\{\mathscr{R}_{\nu,\nu}^{(0)},\chi_{\nu}^{(1)}\}+\left\{\sum_{l\ge\nu+1}\mathscr{R}^{(0)}_{\nu,l},\chi_{\nu}^{(1)}\right\}=\mathcal{O}(\sigma^{\nu+\nu})+\mathcal{O}(\sigma^{\nu+1+\nu-1})=\mathcal{O}(\sigma^{2\nu})
\end{equation*}
and similarly,
\begin{multline*}
\frac{1}{2}\{\{\mathcal{H},\chi_{\nu}^{(1)}\},\chi_{\nu}^{(1)}\}=\frac{1}{2}\{\{\mathscr{Z}_0,\chi_{\nu}^{(1)}\},\chi_{\nu}^{(1)}\}+\frac{1}{2}\{\{\mathscr{R}_{\nu}^{(0)},\chi_{\nu}^{(1)}\},\chi_{\nu}^{(1)}\}\\
=\mathcal{O}(\sigma^{2\nu})+\mathcal{O}(\sigma^{3\nu-1})=\mathcal{O}(\sigma^{2\nu})\;,
\end{multline*}
since $\{\mathscr{Z}_0,\chi_{\nu}^{(1)}\}$ satisfies Lemma \ref{lemma:dalem}. We then have $\{\mathscr{Z}_0,\chi_{\nu}^{(1)}\}=\mathscr{Z}_{\nu}^{(1)}-\mathscr{R}_{\nu,\nu}^{(0)}+\mathcal{O}(\sigma^{\nu+1})$, with $\mathscr{Z}_{\nu}^{(1)}$ independent on $f,g,e$. Proceeding by induction
\begin{equation*}
\frac{1}{n!}\{\ldots\{\{\mathscr{Z}_0+\mathscr{R}_{\nu}^{(0)},\underbrace{\chi_{\nu}^{(1)}\},\chi_{\nu}^{(1)}\},\ldots,\chi_{\nu}^{(1)}}_{n\ge 3}\}=\mathcal{O}(\sigma^{\min\{n\nu-(n-2),\,(n+1)\nu-(n-1)\}})=\mathcal{O}(\sigma^{n(\nu-1)+2})
\end{equation*}
which concludes the proof of the proposition. 
\end{proof}

\noindent
By Proposition \ref{prop:hotstep1}, computing all Poisson brackets in \eqref{eqn:H1}, substituting $\phi_1=e_1\sin E_1$ where appropriate, and multiplying all terms missing a factor $1/\norm{r_1}$ with the factor $a_1(1-\sigma^{\nu_1} e_1\allowbreak\cos(E_1))/\norm{r_1}$ (equal to 1), the remainder $\mathscr{R}_{\nu+1}^{(1)}$ resumes the standard form
\begin{equation}
\label{eqn:R1}
\mathscr{R}_{\nu+1}^{(1)}=\sum_{l\ge\nu+1}\mathscr{R}^{(1)}_{\nu+1,l}\\
=\sum_{l\ge\nu+1}\sum_{\lambda\ge 1}\frac{a_1}{\norm{r_1}^{\lambda}}\sum_{s\in\mathbb{Z}^4}d_{l,\lambda,s}^{(1)}\cos(s_1f+s_2g+s_3h+s_4E_1)\sigma^{l}\;,
\end{equation}
where the coefficients $d_{l,\lambda,s}^{(1)}$ satisfy the relations
\begin{equation*}
d_{l,\lambda,s}^{(1)}=d_{l,\lambda,s}^{(1)}(\delta L,e,\eta,\iota_c,\iota_s,;\mu,L_*,a_1,e_1,\eta_1)=
\begin{cases}
d^{\prime(1)}_{l,\lambda,p}\;,& s_1=s_4=0,\,(s_2,s_3)=p\;,\vspace{1mm}\\
d^{\prime\prime(1)}_{l,\lambda,s}\;,& (s_1,s_4)\neq (0,0)\;,
\end{cases}
\in\mathbb{R}\;.
\end{equation*}
These last algebraic operations conclude the first normalization step. \\

\subsubsection{Loop: normalization of the $\boldsymbol{\sigma^{\nu+j-1}}$-terms}
\label{subsubsec:loop}
The procedure followed in the first step can be repeated iteratively in order to normalize consecutively terms of order $\sigma^{\nu+j-1}$, with each time an $\mathcal{O}(\sigma^{\nu+j})$ remainder, for $\nu,j>1$. As anticipated in Remark \ref{rem:smalle}, the iterative procedure described below fails in the case $\nu=1$ at step $j=2$, so an adjustment (involving one more iteration) is required, as discussed in Subsection \ref{subsubsec:nueq1} below. 

The $j$-th normalization step is carried out as follows from the next proposition. 
\begin{proposition}
\label{prop:normj}
Assume $\nu\ge 2$, $\nu_1\ge 1$. Assume that the Hamiltonian before the $j$-th normalization step has the form:
\begin{equation}
\label{hamjm1}
\mathscr{H}^{(j-1)}=\mathscr{Z}_0+\sum_{l=1}^{j-1}\mathscr{Z}^{(l)}_{\nu+l-1}+\mathscr{R}^{(j-1)}_{\nu+j-1}\;	
\end{equation}
where
\begin{equation}
	\label{eqn:Znulm1jm1}
	\mathscr{Z}^{(l)}_{\nu+l-1}=\sigma^{\nu+l-1}\sum_{\lambda\ge1}\sum_{p\in\mathbb{Z}^2}\zeta^{(l)}_{\nu+l-1,\lambda,p}\cos(p_1g+p_2h)\;.
\end{equation}
\begin{multline}
\label{eqn:Rjm1}
	\mathscr{R}^{(j-1)}_{\nu+j-1}=\sum_{l\ge\nu+j-1}\mathscr{R}^{(j-1)}_{\nu+j-1,l}=\sum_{l\ge\nu+j-1}\sum_{\lambda\ge 1}\frac{a_1}{\norm{r_1}^{\lambda}}\left(\sum_{p\in\mathbb{Z}^2}d_{l,\lambda,p}^{\prime(j-1)}\cos(p_1g+p_2h)\phantom{\sum_{\substack{s\in\mathbb{Z}^4\\(s_1,s_4)\neq(0,0)}}}\right.\\
	\left.\phantom{\sum_{\substack{s\in\mathbb{Z}^4\\(s_1,s_4)\neq(0,0)}}}+\sum_{\substack{s\in\mathbb{Z}^4\\(s_1,s_4)\neq(0,0)}}d_{l,\lambda,s}^{\prime\prime(j-1)}\cos(s_1f+s_2g+s_3h+s_4E_1\right)\sigma^l\;,
\end{multline}
for some real coefficients $\zeta^{(l)}_{\nu+l-1,\lambda,p}$, $d_{l,\lambda,p}^{\prime(j-1)}$, $d_{l,\lambda,s}^{\prime\prime(j-1)}$ specified at previous steps, where 
\begin{equation*}
\zeta_{\nu,\lambda,p}^{(1)}=
\begin{cases}
	q_{\nu,p}^{\prime}\;,&\lambda=1\\
	0\;,&\lambda>1
\end{cases}
\end{equation*}
by \eqref{eqn:Z1nu}. \\

\noindent Define the $j$-th step Lie generating function $\chi_{\nu+j-1}^{(j)}$ as
\begin{multline}
\label{eqn:chij}
	\chi_{\nu+j-1}^{(j)}=\frac{\phi_1}{n_1}\sigma^{\nu+j-1+\nu_1}\sum_{\lambda\ge 1}\sum_{\psi=1}^{\lambda}\frac{1}{a_1^{\psi-1}\norm{r_1}^{\lambda-\psi}}\sum_{p\in\mathbb{Z}^2}d_{\nu+j-1,\lambda,p}^{\prime(j-1)}\cos(p_1g+p_2h)\\
	+\sigma^{\nu+j-1}\sum_{\lambda\ge 1}\frac{1}{\norm{r_1}^{\lambda-1}}\sum_{\substack{s\in\mathbb{Z}^4\\(s_1,s_4)\neq(0,0)}}\frac{d^{\prime\prime(j-1)}_{\nu+j-1,\lambda,s}}{s_1 n_*+s_4n_1}\sin(s_1f+s_2g+s_3h+s_4E_1)\;.
\end{multline}

\noindent
Then, the Hamiltonian $\mathscr{H}^{(j)}$ produced by the Lie operation $\mathscr{H}^{(j)}=\exp\left(\mathcal{L}_{\chi_{\nu+j-1}^{(j)}}\right)\mathscr{H}^{(j-1)}$ has the form
\begin{equation}
	\label{eqn:Lieseriesj}
	\mathscr{H}^{(j)}=\exp\left(\mathcal{L}_{\chi_{\nu+j-1}^{(j)}}\right)\mathscr{H}^{(j-1)}=\mathscr{Z}_0+\sum_{l=1}^{j}\mathscr{Z}^{(l)}_{\nu+l-1}+\mathscr{R}^{(j)}_{\nu+j}\;,	
\end{equation}
where
\begin{equation}
	\label{eqn:Znulm1j}
	\mathscr{Z}^{(j)}_{\nu+j-1}=\sigma^{\nu+j-1}\sum_{\lambda\ge1}\sum_{p\in\mathbb{Z}^2}\zeta^{(j)}_{\nu+j-1,\lambda,p}\cos(p_1g+p_2h)
\end{equation}
with
\begin{equation}
\label{eqn:Znulm1jcoef}
	\zeta^{(j)}_{\nu+j-1,\lambda,p}=\frac{1}{a_1^{\lambda-1}}d^{\prime(j-1)}_{\nu+j-1,\lambda,p}\;,
\end{equation}
and
\begin{multline}
\label{eqn:Rj}
	\mathscr{R}^{(j)}_{\nu+j}=\sum_{l\ge\nu+j}\mathscr{R}^{(j)}_{\nu+j,l}=\sum_{l\ge\nu+j}\sum_{\lambda\ge 1}\frac{a_1}{\norm{r_1}^{\lambda}}\left(\sum_{p\in\mathbb{Z}^2}d_{l,\lambda,p}^{\prime(j)}\cos(p_1g+p_2h)\phantom{\sum_{\substack{s\in\mathbb{Z}^4\\(s_1,s_4)\neq(0,0)}}}\right.\\
	\left.\phantom{\sum_{\substack{s\in\mathbb{Z}^4\\(s_1,s_4)\neq(0,0)}}}+\sum_{\substack{s\in\mathbb{Z}^4\\(s_1,s_4)\neq(0,0)}}d_{l,\lambda,s}^{\prime\prime(j)}\cos(s_1f+s_2g+s_3h+s_4E_1\right)\sigma^l\;,
\end{multline}
with real coefficients $d_{l,\lambda,p}^{\prime(j)}$, $d_{l,\lambda,s}^{\prime\prime(j)}$ computed from the known coefficients $\zeta^{(l)}_{\nu+l-1,\lambda,p}$ ($l=1,\ldots,j-1$), $d_{l,\lambda,p}^{\prime(j-1)}$, $d_{l,\lambda,s}^{\prime\prime(j-1)}$.

\end{proposition}

\begin{proof}
	We repeat the strategy of Proposition \ref{prop:hotstep1} and look for a generating Hamiltonian this time dependent on $\norm{r_1}$:
	\begin{multline*}
	\chi_{\nu+j-1}^{(j)}(f,g,h,E_1,\phi_1,\norm{r_1},\delta L,e,\eta,\iota_c,\iota_s)\\
	=\sigma^{\nu+j-1}\left(\phi_1\sigma^{\nu_1}\sum_{\lambda\ge 1}\sum_{p\in\mathbb{Z}^2}\hat{d}^{\prime(j-1)}_{\nu+j-1,\lambda,p}(\norm{r_1},\delta L,e,\eta,\iota_c,\iota_s)\cos(p_1g+p_2h)\phantom{\sum_{\substack{s\in\mathbb{Z}^4\\(s_1,s_4)\neq(0,0)}}}\right.\\
	+\left.\sum_{\lambda\ge1}\sum_{\substack{s\in\mathbb{Z}^4\\(s_1,s_4)\neq(0,0)}}\hat{d}^{\prime\prime(j-1)}_{\nu+j-1,\lambda,s}\sin(s_1+s_2g+s_3h+s_4E_1)\right)\;.
	\end{multline*}
	Requiring $\{\mathscr{Z}_0,\chi_{\nu+j-1}^{(j)}\}+\mathscr{R}^{(j-1)}_{\nu+j-1,\nu+j-1}$ to be $\mathcal{O}(\sigma^{\nu+j})$ in fast angles we come up with
	\begin{gather*}
-n_*\hat{d}^{\prime\prime(j-1)}_{\nu+j-1,\lambda,s}s_1-n_1\hat{d}^{\prime\prime(j-1)}_{\nu+j-1,\lambda,s}s_4+\frac{1}{\norm{r_1}^{\lambda-1}}d^{\prime\prime(j-1)}_{\nu+j-1,\lambda,s}=0\;,\\
-n_1\frac{a_1}{\norm{r_1}}\hat{d}^{\prime(j-1)}_{\nu+j-1,\lambda,p}+n_1\hat{d}^{\prime(j-1)}_{\nu+j-1,\lambda,p}+\frac{a_1}{\norm{r_1}^{\lambda}}d^{\prime(j-1)}_{\nu+j-1,\lambda,p}=\frac{1}{a_1^{\lambda-1}}d^{\prime(j-1)}_{\nu+j-1,\lambda,p}\;,
\end{gather*}
that is, for $\lambda\ge1$,
\begin{align*}
\hat{d}_{\nu+j-1,\lambda,s}^{\prime\prime(j-1)}&=\frac{1}{\norm{r_1}^{\lambda-1}}\frac{d_{\nu+j-1,\lambda,s}^{\prime\prime(j-1)}}{s_1n_*+s_4n_1}\;,\quad s\in\mathbb{Z}^4\colon(s_1,s_4)\neq(0,0)\;,\\
\hat{d}_{\nu+j-1,\lambda,p}^{\prime(j-1)}&=\frac{1}{a_1^{\lambda-1}}\frac{d_{\nu+j-1,\lambda,p}^{\prime(j-1)}}{n_1}\sum_{\psi=0}^{\lambda-1}\left(\frac{a_1}{\norm{r_1}}\right)^{\psi}=\frac{d_{\nu+j-1,\lambda,p}^{\prime(j-1)}}{n_1}\sum_{\psi=1}^{\lambda}\frac{1}{a_1^{\psi-1}\norm{r_1}^{\lambda-\psi}}\;,\quad p\in\mathbb{Z}^2\;,
\end{align*}
which proves Eq.(\ref{eqn:chij}), and new accumulated addenda in normal form
\begin{equation*}
\mathscr{Z}^{(j)}_{\nu+j-1}=\sigma^{\nu+j-1}\sum_{\lambda\ge 1}\frac{1}{a_1^{\lambda-1}}\sum_{p}d_{\nu+j-1,\lambda,p}^{\prime(j-1)}\cos(p_1g+p_2h)\;.
\end{equation*}
which proves Eq.(\ref{eqn:Znulm1jcoef}). It remains to demonstrate that the expression \eqref{eqn:Rj} is $\mathcal{O}(\sigma^{\nu+j})$. The proof is done by induction: for $j=2$ we get
\begin{multline}
\label{eqn:H2}
\mathscr{H}^{(2)}=\mathscr{Z}_0+\mathscr{Z}_{\nu}^{(1)}+\mathscr{Z}^{(2)}_{\nu+1}+\mathcal{O}(\sigma^{\nu+2})+\sum_{l\ge\nu+2}\mathscr{R}^{(1)}_{\nu+1,l}+\{\mathscr{Z}_{\nu}^{(1)},\chi_{\nu+1}^{(2)}\}\\
+\{\mathscr{R}_{\nu+1}^{(1)},\chi_{\nu+1}^{(2)}\}+\ldots+\sum_{n\ge 2}\frac{1}{n!}\{\ldots\{\{\mathscr{H}^{(1)},\underbrace{\chi_{\nu+1}^{(2)}\},\chi_{\nu+1}^{(2)}\},\ldots,\chi_{\nu+1}^{(2)}}_{n}\}\;.
\end{multline}
Similarly as in Proposition \ref{prop:hotstep1}, a lowering of the book-keeping exponents in a Poisson bracket of the form $\{F,\chi_{\nu+1}^{(2)}\}$ can occur through derivatives of the form \ref{item:ederiv}. However, this time the latter can only appear in a Poisson bracket via the combination
\begin{equation}
\label{eqn:ederivexpr}
\sigma^{-1}\left(\frac{\partial f}{\partial\ell}\frac{\partial e}{\partial\delta L}\left(\frac{\partial F}{\partial f}\frac{\partial\chi_{\nu+1}^{(2)}}{\partial e}-\frac{\partial F}{\partial e}\frac{\partial\chi_{\nu+1}^{(2)}}{\partial f}\right)+\right.
\left.\frac{\partial e}{\partial G}\left(\frac{\partial F}{\partial g}\frac{\partial\chi_{\nu+1}^{(2)}}{\partial e}-\frac{\partial F}{\partial e}\frac{\partial\chi_{\nu+1}^{(j)}}{\partial g}\right)\right)\;;
\end{equation}
so we can infer that
\begin{equation*}
\{\mathscr{Z}_{\nu}^{(1)},\chi_{\nu+1}^{(2)}\}=\mathcal{O}(\sigma^{2\nu+1})\;,\quad\{\mathscr{R}_{\nu+1}^{(1)},\chi_{\nu+1}^{(2)}\}=\mathcal{O}(\sigma^{2\nu})\;,
\end{equation*}
\begin{multline*}
\frac{1}{n!}\{\ldots\{\{\mathscr{H}^{(1)},\underbrace{\chi_{\nu+1}^{(2)}\},\chi_{\nu+1}^{(2)}\},\ldots,\chi_{\nu+1}^{(2)}}_{n\ge 2}\}\\
=\mathcal{O}(\sigma^{\min\{n(\nu+1)-2(n-1),\,n(\nu+1)+\nu-2(n-1),\,(n+1)(\nu+1)-2n\}})=\mathcal{O}(\sigma^{n(\nu-1)+2})
\end{multline*}
because \eqref{eqn:fdeltaLPoiss}, \eqref{eqn:fGPoiss}, \eqref{eqn:ederivexpr} vanish when $F=F_1=\mathscr{Z}_{\nu}^{(1)}$. Now, for all $\nu\ge2$, $n(\nu-1)+2>\nu+1$, $n\ge2$, hence, the proposition is valid for $j=2$. For $j\ge 3$, we have
\begin{multline}
\label{eqn:Hj}
\mathscr{H}^{(j)}=\mathscr{Z}_0+\mathscr{Z}_{\nu}^{(1)}+\ldots+\mathscr{Z}^{(j-1)}_{\nu+j-2}+\mathscr{Z}^{(j)}_{\nu+j-1}+\mathcal{O}(\sigma^{\nu+j})+\sum_{l\ge\nu+j}\mathscr{R}^{(j-1)}_{\nu+j-1,l}\\
+\{\mathscr{Z}_{\nu}^{(1)}+\ldots+\mathscr{Z}^{(j-1)}_{\nu+j-2},\chi_{\nu+j-1}^{(j)}\}+\{\mathscr{R}_{\nu+j-1}^{(j-1)},\chi_{\nu+j-1}^{(j)}\}+\ldots\\
+\sum_{n\ge 2}\frac{1}{n!}\{\ldots\{\{\mathscr{H}^{(j-1)},\underbrace{\chi_{\nu+j-1}^{(j)}\},\chi_{\nu+j-1}^{(j)}\},\ldots,\chi_{\nu+j-1}^{(j)}}_{n}\}\;,
\end{multline}
and analogously
\begin{gather*}
\{\mathscr{Z}_{\nu}^{(1)},\chi_{\nu+j-1}^{(j)}\}=\mathcal{O}(\sigma^{2\nu+j-1})\;,\quad\{\mathscr{Z}_{\nu+j-2}^{(j-1)},\chi_{\nu+j-1}^{(j)}\}=\mathcal{O}(\sigma^{2\nu+2j-5})\;,\\
\{\mathscr{R}_{\nu+j-1}^{(j-1)},\chi_{\nu+j-1}^{(j)}\}=\mathcal{O}(\sigma^{2\nu+2j-4})\;,
\end{gather*}
\begin{multline*}
\frac{1}{n!}\{\ldots\{\{\mathscr{H}^{(j-1)},\underbrace{\chi_{\nu+j-1}^{(j)}\},\chi_{\nu+j-1}^{(j)}\},\ldots,\chi_{\nu+j-1}^{(j)}}_{n\ge 2}\}\\
=\mathcal{O}(\sigma^{\min\{n(\nu+j-1)-2(n-1),\,n(\nu+j-1)+\nu-2(n-1),\,n(\nu+j-1)+\nu+j-2-2n,\,(n+1)(\nu+j-1)-2n\}})\\
=\mathcal{O}(\sigma^{n(\nu+j-3)+2})\;.
\end{multline*}
However, since $\nu>1$, $n\ge 2$, we readily find $n(\nu+j-3)+2>\nu+j-1$, which concludes the proof.\\
\end{proof}

\subsubsection{The case $\nu=1$}
\label{subsubsec:nueq1}
Coming to $\nu=1$, one realizes that \eqref{eqn:H2} produces same order $\sigma^2$ non-normalized terms via $\{\mathscr{R}^{(1)}_{2},\chi_{2}^{(2)}\}$ and $\{\ldots\{\{\mathscr{Z}_0+\mathscr{R}_{2}^{(1)},\chi_{2}^{(2)}\},\chi_{2}^{(2)}\},\ldots,\chi_{2}^{(2)}\}$, namely the resulting remainder is $\mathscr{R}^{(2)}_2$, so the scheme in Proposition \ref{prop:normj} is not directly applicable beyond $j=1$. Despite this, it is worth noticing that if we manage to get rid of these spurious terms, by performing, for instance, an extra normalization $\text{II}$, such that the new outcome returns $\mathscr{R}^{(\text{II})}=\mathscr{R}^{(\text{II})}_3$, then the algorithm \eqref{eqn:Lieseriesj} will work for $j\ge3$ upon restarting the recursion from iteration $\text{II}$ in place of $2$. This is precisely the claim we are about to show to complete the treatment.\\
Let us write \eqref{eqn:H2} as $\mathscr{H}^{(2)}=\mathscr{Z}_0+\mathscr{Z}_1^{(1}+\mathscr{Z}_2^{(2)}+\mathscr{R}^{(2)}_2$. Introduce the extra second normalization $\text{II}$ based on Proposition \ref{prop:normj} targeted to $\mathscr{R}_{2,2}^{(2)}$ with generating function $\chi_{2}^{(\text{II})}$. Then we have the following.
\begin{proposition}
\label{prop:adjust}
For $\nu=1$ and any $\nu_1\ge1$, 
\begin{equation}
\label{eqn:normII}
\mathscr{H}^{(\text{II})}=\exp\left(\mathcal{L}_{\chi_{2}^{(\text{II})}}\right)\mathscr{H}^{(2)}=\mathscr{Z}_0+\mathscr{Z}_1^{(1)}+\mathscr{Z}_2^{(2)}+\mathscr{Z}_2^{(\text{II})}+\mathscr{R}^{(\text{II})}_3\;.
\end{equation}
Moreover the loop composed by \eqref{eqn:Lieseriesj}--\eqref{eqn:Rj} in Proposition \ref{prop:normj} holds true for any $j\ge 4$ under the modifications
\begin{equation}
\label{eqn:H3}
\mathscr{H}^{(3)}=\exp\left(\mathcal{L}_{\chi_3^{(3)}}\right)\mathscr{H}^{(\text{II})}=\mathscr{Z}_0+\mathscr{Z}_1^{(1)}+\mathscr{Z}_2^{(2)}+\mathscr{Z}_2^{(\text{II})}+\mathscr{Z}_3^{(3)}+\mathscr{R}^{(3)}_{4}\;,
\end{equation} 
\begin{equation}
\label{eqn:Hjmod}
\mathscr{H}^{(j)}=\exp\left(\mathcal{L}_{\chi_j^{(j)}}\right)\mathscr{H}^{(j-1)}=\mathscr{Z}_0+\sum_{l=1}^j\mathscr{Z}^{(l)}_l+\mathscr{Z}_2^{(\text{II})}+\mathscr{R}^{(j)}_{j+1}\;.
\end{equation}
\end{proposition}
\begin{proof}
	We begin with a necessary generalization of Lemma \ref{lemma:dalem}.
	\begin{lemma}
		\label{lemma:dalemgen}
		Given $F_1,F_2\in\mathcal{C}^{\omega}(\mathbb{T}\times D)$ trigonometric monomials of the form \eqref{eqn:indipe}, or equivalently in terms of the sine, fulfilling the property of Lemma \ref{lemma:dalem}, addenda of the same type in the Lie series transformation applied to $F_1$ with respect to $F_2$ preserve such property.
	\end{lemma}
	\begin{proof}
		Since $\exp\left(\mathcal{L}_{F_2}\right)F_1$ involves the computation of Poisson brackets of functions explicitly independent on $e$, we have that \eqref{eqn:ederivexpr}, with $F_1,F_2$ in place of $F,\chi_{\nu+1}^{(2)}$, is identically null, as well as \eqref{eqn:fGPoiss} because $\partial F_1/\partial f=\partial F_1/\partial g$, $\partial F_2/\partial f=\partial F_2/\partial g$ by assumption. Thus, the bracket $\{F_1,F_2\}$ in the Lie series either does not introduce any eccentricity dependence at all, or only at numerator through \eqref{eqn:sigmadfdl} multiplied by $\cos f$ or $\cos^2f$; therefore its derivatives contain products of cosines (sines) whose coefficients are independent on $e$ like
		\begin{equation*}
		\mathscr{G}_1(s_1f+s_2g+s_3h+s_4E_1)\mathscr{G}_2(u_1f+u_2g+u_3h+u_4E_1)\;,\quad\mathscr{G}_1,\mathscr{G}_2=\cos,\sin\;.
		\end{equation*}
		The arguments are now either summed or subtracted, hence they clearly satisfy the property concerned. By cascade reasoning for further nested brackets we conclude.\\
	\end{proof}
	\begin{remark}
	\label{rem:corollary} 
	A straightforward use of the lemma in conjunction with formulae \eqref{eqn:fdeltaLPoiss}, \eqref{eqn:fGPoiss}, \eqref{eqn:ederivexpr} ($\chi_{\nu+1}^{(2)}$ replaced by generic differentiable function) reveal that any transformed Hamiltonian $\mathscr{H}^{(j)}$ and corresponding generating function $\chi_{\nu+j-1}^{(j)}$ encountered are regular at $e=0$ in agreement with D'Alembert rules, i.e. they never depend on negative powers of $e$. Furthermore, every time one of the two entries of $\{\cdot,\cdot\}$ does not depend on $e$, the upshot due to item \ref{item:ederiv} in the proof of Proposition \ref{prop:hotstep1}, as soon as non-zero, is diminished by $\sigma^{-1}$ instead of $\sigma^{-2}$.\\
	\end{remark}
	\noindent We consider step $\text{II}$:
	\begin{multline}
	\label{eqn:HII}
	\mathscr{H}^{(\text{II})}=\mathscr{Z}_0+\mathscr{Z}_1^{(1)}+\mathscr{Z}_2^{(2)}+\mathscr{Z}_2^{(\text{II})}+\mathcal{O}(\sigma^3)+\sum_{l\ge 3}\mathscr{R}_{2,l}^{(2)}+\{\mathscr{Z}_1^{(1)},\chi_{2}^{(\text{II})}\}+\{\mathscr{Z}_2^{(2)},\chi_{2}^{(\text{II})}\}\\
	+\{\mathscr{R}_2^{(2)},\chi_{2}^{(\text{II})}\}+\ldots+\sum_{n\ge 2}\frac{1}{n!}\{\ldots\{\{\mathscr{H}^{(2)},\underbrace{\chi_{2}^{(\text{II})}\},\chi_{2}^{(\text{II})}\},\ldots,\chi_{2}^{(\text{II})}}_{n}\}\;.
	\end{multline}
	The analysis of the contributions reports these deductions, by which \eqref{eqn:normII} follows.
	\begin{itemize}
		\item $\{\mathscr{Z}_1^{(1)},\chi_{2}^{(\text{II})}\}=\mathcal{O}(\sigma^3)$ because $\mathscr{Z}_1^{(1)}$ is independent on $f,g,e$.
		\item $\{\mathscr{Z}_2^{(2)},\chi_{2}^{(\text{II})}\}=\mathcal{O}(\sigma^4)$
		because $\mathscr{Z}^{(2)}$ and $\chi_2^{(\text{II})}$ fulfil Lemma \ref{lemma:dalemgen}. Indeed, $\mathscr{R}_{2,2}^{(1)}$ depends on $e$ at most linearly by book-keeping rules, so it does $\chi_2^{(2)}$ by construction. At this point we show that for eccentricity dependent terms stemming from $\mathscr{R}_{2,2}^{(1)}$ (or equivalently $\chi_2^{(2)}$) $d^{\prime(1)}_{2,\lambda,p}=0$.
		\begin{lemma}
			\label{lemma:e+fastang}
			Every trigonometric monomial in $\mathscr{R}_{2,2}^{(1)}$ explicitly dependent on $e$ carries the dependence on at least one of the two fast anomalies $f,E_1$ as well, namely corresponding coefficients in \eqref{eqn:R1} are $d_{2,\lambda,s}^{(1)}=d^{\prime\prime(1)}_{2,\lambda,s}$, $(s_1,s_4)\neq(0,0)$.
		\end{lemma}
		\begin{proof}
			By Proposition \ref{prop:hotstep1}, Lemma \ref{lemma:dalem} and \ref{lemma:dalemgen}, the substitution $\phi_1=e_1\sin E_1$ and the formulas listed in Subsection \ref{subsubsec:poissform}, Subsection \ref{subsubsec:bkPoissstr}, we take out of \eqref{eqn:H1} the order $\sigma^2$ remainder and it is not restrictive to assume $\nu_1=1$ in order to include also the $e_1\cos E_1$ dependent term in \eqref{eqn:firsthom}:
			\begin{multline*}
			\mathscr{R}_{2,2}^{(1)}=\mathscr{R}_{1,2}^{(0)}+\frac{a_1}{\norm{r_1}}\Bigg( n_*\left(e_1\cos E_1-\frac{2e\cos f}{\eta^3}\right)\sigma\frac{\partial\chi_1^{(1)}}{\partial f}+\frac{\partial\mathscr{R}_{1,1}^{(0)}}{\partial f}\frac{\partial\chi_1^{(1)}}{\partial\delta L}\\
			-\frac{\partial\mathscr{R}_{1,1}^{(0)}}{\partial \delta L}\frac{\partial\chi_1^{(1)}}{\partial f}-\frac{1}{L_*}\frac{\partial\chi_1^{(1)}}{\partial\iota_c}\left(\iota_c\frac{\partial\mathscr{R}_{1,1}^{(0)}}{\partial f}-\frac{\partial\mathscr{R}_{1,1}^{(0)}}{\partial h}\right)\\
			+\frac{1}{L_*}\frac{\partial\mathscr{R}_{1,1}^{(0)}}{\partial\iota_c}\left(\iota_c\frac{\partial\chi_1^{(1)}}{\partial f}-\frac{\partial\chi_1^{(1)}}{\partial h}\right)-\frac{2\sin f}{L_*e}\sigma^{-1}\frac{\partial\chi_1^{(1)}}{\partial f}
			\left(\frac{\partial\mathscr{R}_{1,2}^{(0)}}{\partial g}-\frac{\partial\mathscr{R}_{1,2}^{(0)}}{\partial f}\right)\Bigg)\\
			-\frac{a_1}{2}\left\{\frac{1}{\norm{r_1}}\left(n_*\left(1+\frac{2e\cos f}{\eta^3}\sigma\right)\frac{\partial\chi_1^{(1)}}{\partial f}+n_1\frac{\partial\chi_1^{(1)}}{\partial E_1}\right),\chi_1^{(1)}\right\}_{2}\;,
			\end{multline*}
			where $\{\cdot,\cdot\}_2$ indicates that we retain only $\sigma^2$ quantities after the operation (in virtue of Lemma \ref{lemma:dalemgen} and Remark \ref{rem:corollary}, inductions derived to demonstrate Proposition \ref{prop:hotstep1} are a coarser bound and no other parts of order $\sigma^2$ come out). Plugging in \eqref{eqn:chi1} and \eqref{eqn:fE1terms} for $l=1,2$ and taking into account Lemma \ref{lemma:dalem}, upon simplifications the contributions involving $e$ result
			\begin{multline}
			\label{eqn:R122}
			\mathscr{R}_{1,2_e}^{(0)}-\frac{a_1en_*}{\eta^3\norm{r_1}}\sigma ^2\sum_{(s_1,s_4)\neq(0,0)}\frac{s_1q_{1,s}^{\prime\prime}}{s_1n_*+s_4n_1}(\cos((1-s_1)f-s_1g-s_3h-s_4E_1)\\
			+\cos((1+s_1)f+s_1g+s_3h+s_4E_1))\;,
			\end{multline}
			where 
			\begin{equation}
			\label{eqn:R012e}
			\mathscr{R}_{1,2_e}^{(0)}=\frac{a_1}{\norm{r_1}}\sigma^2\sum_{s\in\mathbb{Z}^4}q_{2,s}\cos(s_1f+s_2g+s_3h+s_4E_1)\;,\quad q_{2,s}=e\bar{q}_{2,s}\;.
			\end{equation}
			We employ now all D'Alembert rules to show that only the harmonics of interest can exist.\\
			Following the same argument as in Lemma \ref{lemma:dalem}, let us write the cosine input of \eqref{eqn:R012e} using modified Delaunay angles \eqref{eqn:modDelang} also for $\mathcal{P}_1$ in relation to corresponding orbital elements \eqref{eqn:Del} (subscript `$1$'):
			\begin{equation*}
			s_1\tilde{\lambda}+(s_1-s_2)\tilde{p}+(s_2-s_3)\tilde{q}+s_4\tilde{\lambda}_1+(s_4-s_5)\tilde{p}_1+(s_5-s_6)\tilde{q}_1\;,\quad s_l\in\mathbb{Z}\;,
			\end{equation*}
			in which $\tilde{p}_1=\tilde{q}_1=0$. For the elimination of the apparent singularity at $e=0$, we must have $1-|s_1+s_2|\ge0$ and even, hence $s_2=s_1\pm1$. Then, since $\mathscr{R}_{1,2_e}^{(0)}$ is independent on $e_1$ by book-keeping setting, analogously we must end up with $s_4=s_5$. Regarding instead the regularity at $i_1=0$, because of the absence of $i_1$ we must conclude that $0-|s_5-s_6|\in2\mathbb{N}$, namely $s_5=s_6$. At this stage, we invoke the invariance under rotation around the $Z$ axis, which prescribes
			\begin{equation*} s_1-s_1+s_2-s_2+s_3+s_4-s_4+s_5-s_5+s_6=s_3+s_6=0\;,
			\end{equation*}
			and summing up this implies $s_3=-s_4$. Ultimately, concerning the inclination, we must ensure that $l-|s_2-s_3|\in2\mathbb{N}$, with $l$ even as well again being $i_1$ not involved, thus $s_2=s_3\pm2n$, $n\le l/2$ natural number. Putting all together we arrive at
			\begin{equation*}
			s_1f+s_2g+s_3h+s_4E_1\implies s_1f+(s_1\pm 1)g+(s_1\mp 2n\pm 1)h+(\pm 2n\mp 1-s_1)E_1\;,
			\end{equation*}
			which always depends on at least one among $f,E_1$ since the coefficients $s_1$, $\pm 2n\mp 1-s_1$ never vanish simultaneously.\\
			By means of an identical reasoning and given the preservation of D'Alembert rules under $\exp\left(\mathcal{L}_{\chi_1^{(1)}}\right)$, we achieve the same outcome for the remaining part of \eqref{eqn:R122} after replacing $s_1\mapsto1\pm s_1$, indeed we find
			\begin{equation*}
			(1\pm s_1)+(1\pm s_1\pm 1)g+(1\pm s_1\mp 2n\pm 1)h+(\pm 2n\mp 1-1\mp s_1)E_1\;,
			\end{equation*}
			and no solutions to $1\pm s_1=0$, $\pm 2n\mp 1-1\mp s_1=0$.
		\end{proof}
		Given that the order $2$ normal form is sourced from the part of $\mathscr{R}_{2,2}^{(1)}$ explicitly independent on fast angles, it turns out that it is free of $e$. Finally, $\mathscr{R}^{(2)}_{2,2}$ is free of $e$ too, being generated by terms in $\{\mathscr{R}_{2,2}^{(1)},\chi_{2}^{(2)}\}$ and $\{\ldots\{\{\mathscr{Z}_0+\mathscr{R}_{2,2}^{(1)},\chi_2^{(2)}\},\chi_2^{(2)}\},\ldots,\chi_2^{(2)}\}$ subjected to computation \ref{item:ederiv} of Proposition \ref{prop:hotstep1} (Remark \ref{rem:corollary}). Again by construction, the same applies to $\chi_2^{(\text{II})}$.
		\item $\{\mathscr{R}_2^{(2)},\chi_{2}^{(\text{II})}\}=\mathcal{O}(\sigma^4)$ by Remark \ref{rem:corollary}.
		\item $\displaystyle\frac{1}{n!}\{\ldots\{\{\mathscr{H}^{(2)},\underbrace{\chi_{2}^{(\text{II})}\},\chi_{2}^{(\text{II})}\},\ldots,\chi_{2}^{(\text{II})}}_{n\ge 2}\}=\mathcal{O}(\sigma^4)$ consequently.
		\end{itemize}
		In order to conclude, we just need to check that the next step gives rise to an $\mathcal{O}(\sigma^4)$ perturbation and the cycle of normalizations can restart for $j\ge4$ in light of the bounds on $\sigma$ from \eqref{eqn:Hj} at the end of the proof of Proposition \ref{prop:normj}. Upon repeating the usual argument, it is easy to see that the only bracket worth investigating is $\{\mathscr{Z}_2^{(\text{II})},\chi_3^{(3)}\}$, that is, nevertheless, $\mathcal{O}(\sigma^4)$ because $\mathscr{Z}_2^{(\text{II})}$ is made out of $\mathscr{R}_{2,2}^{(2)}$ independent on $e$.\\
		\end{proof}

\begin{remark}
	By the above argument it is immediate to realize that even $p_2\equiv0$ in \eqref{eqn:chi1} and \eqref{eqn:fE1terms} for $l=\nu$, so $q^{\prime}_{\nu,p}=0$ for all $p\neq(0,0)$.\\
\end{remark}

Serving as an example, a detailed demonstration of the normalization procedure exposed in the present section for a simple model, containing just few terms of the disturbing function, is presented in Appendix \ref{sec:appendixex}.\\

\section{Numerical tests}
\label{sec:app}

\subsection{Computer-algebraic implementation of the normalization algorithm}
\label{subsec:implem}
Implementing the above normalization procedure, e.g. by use of a Computer Algebra System (CAS), requires working with a finite truncation of the initial Hamiltonian model \eqref{eqn:hammump}. To this end, the disturbing function \eqref{eqn:H1Cart} multiplied by $\mu$ can be re-arranged as 
\begin{equation}
\label{eqn:H1Cartnew}
\mu\mathcal{H}_1=-\frac{\mathcal{G}m_0\mu}{\norm{R}}\sum_{\kappa_1=0}^{\infty}
\sum_{\substack{\kappa_2=0\\
		\kappa_2\neq 1}}^{\infty}
\sum_{\kappa_3=0}^{\infty}
\tilde{h}_{\kappa_1,\kappa_2,\kappa_3}
\mu^{\kappa_1}
\left(\frac{2 r_1\cdot R}{\norm{R}^2}\right)^{\kappa_2}
\left(\frac{\norm{r_1}}{\norm{R}}\right)^{2\kappa_3}\;,
\end{equation}
where $\tilde{h}_{\kappa_1,\kappa_2,\kappa_3}$ are real coefficients derived from the coefficients of \eqref{eqn:H1Cart}. A convenient truncation of \eqref{eqn:H1Cartnew} stems from defining two separate truncation orders in powers of $\mu$ (truncation order $k_\mu$), and in powers of $\norm{r_1}/\norm{R}$ (multipole truncation order $k_\text{mp}$), through the formula
\begin{equation}
\label{eqn:H1Carttrunc}
\mathcal{H}_1^{\leq k_\mu,k_\text{mp}}=
-\frac{\mathcal{G}m_0\mu}{\norm{R}}
\sum_{\kappa_1=0}^{k_\mu-1}
\sum_{\kappa_2=0,\kappa\neq 1}^{k_\text{mp}}
\sum_{\kappa_3=0}^{\lfloor k_\text{mp}/2\rfloor}
\tilde{h}_{\kappa_1,\kappa_2,\kappa_3}
\mu^{\kappa_1}
\left(\frac{2 r_1\cdot R}{\norm{R}^2}\right)^{\kappa_2}
\left(\frac{\norm{r_1}}{\norm{R}}\right)^{2\kappa_3}\;,
\end{equation}
where $\lfloor\cdot\rfloor$ is the integer part function.
Working with the truncated Hamiltonian $\mathcal{H}^{\leq k_\mu,k_\text{mp}}=\mathcal{H}_0+\mathcal{H}_1^{\leq k_\mu,k_\text{mp}}$, we then obtain a sequence of secular models $\mathscr{Z}^{(j)}$, $j=1,2,\ldots$, where $j$ denotes the normalization step, computed via the formula
\begin{equation}\label{eqn:zj}
\mathscr{Z}^{(j)}=\mathscr{Z}_0+\sum_{l=1}^{j}\mathscr{Z}^{(l)}_{\nu+l-1}\;.
\end{equation}
In particular, we implement the following steps of the CAS algorithm:
\begin{enumerate}[label=(\roman*)]
	\item 
	for a fixed value of $\mu$, choose values for $k_{\mu},k_{\text{mp}}$, perform the corresponding expansions of the Hamiltonian as in \eqref{eqn:H1Cartnew} and compute the truncated model $\mathcal{H}^{\leq k_\mu,k_\text{mp}}$;
    \item
    choose the reference values of $a_*$ and $e_*$; 
	\item 
	pass to variables $(f,g,h,E_1,\delta L,e,\eta,\iota_c,\iota_s,J_1)$ and parameters $L_*,e_1,a_1,\eta_1$ on the basis of the selected $a_*$;
	\item 
	compute $\nu$ and $\nu_1$ (Eq.(\ref{eqn:nun1}));
	\item 
	set the appropriate book-keeping weights following the rules in Subsection \ref{subsubsec:bkrules} and expand correspondingly the Hamiltonian in $\delta L$ up to $\sigma^{\nu k_{\mu}}$;
	\item 
	drop constants, perform the identity operation \eqref{eqn:trick}, discard book-keeping powers larger than $\nu k_{\mu}$ and introduce $n_*$;
	\item 
	if $\nu>1$, compute the generating function \eqref{eqn:chi1} as well as the first-normalized Hamiltonian $\mathscr{H}^{(1)}$ by the Lie series operation \eqref{eqn:Lieseries} truncated at the maximum book-keeping order $N_\text{bk}=\nu k_\mu$; if $\nu=1$, compute $\mathscr{H}^{(1)}$ (always truncated to the book-keeping order $N_\text{bk}$) via the procedure of Subsection \ref{subsubsec:nueq1}; 
	\item 
	compute the successive normalizations $\mathscr{H}^{(j)}$, truncated at book-keeping order $N_\text{bk}$ via the procedure of Subsection \ref{subsubsec:loop}, up to a maximum normalization order $\nu+j_{max}-1<N_\text{bk}$, $j_{max}\le\nu(k_{\mu}-1)$; this allows to obtain truncated Hamiltonian models containing a finite number of normal form terms as well as a finite number of terms provided by the truncated remainder.
\end{enumerate}

In the CAS implementation of the above algorithm we work with numerical coefficients, substituting all constants with their corresponding numerical values. Several types of numerical tests of the precision and overall performance of the method can be carried out as exemplified in the sequel.  

\subsection{Numerical examples in the Sun-Jupiter ER3BP: semi-analytic orbit propagation}
\label{subsec:SunJupsystem}

For all numerical tests below we refer to the Sun-Jupiter one ($\mu=9.5364\cdot10^{-4}$). We employ Earth-orbit based units, such that $\mathcal{G}m_0=4\pi^2$AU$^3$/y$^2$, $a_1=5.2044$AU, so that Jupiter's period is $T_1=11.86$ y. Jupiter's mean motion is $n_1=2\pi/T_1$, and eccentricity $e_1=0.0489$, used throughout all  computations in the framework of the ER3BP model. 

In all tests below, a particle's orbit is defined by providing the initial conditions $a(0),e(0),i(0)$,  complemented by $f(0)=g(0)=h(0)=0$. 

Our basic probe of the efficiency of the normalization method in the framework of the ER3BP is given by comparing the short-period oscillations of the orbital elements $a(t),e(t),i(t),g(t),h(t)$, as found by two different methods. \\
\\
\noindent
\textit{Direct Cartesian propagation}: the initial conditions $z(0)\coloneqq(a(0),e(0),i(0),f(0),g(0),h(0))$ are mapped into initial conditions for the Cartesian canonical positions and conjugate momenta $(X(0),Y(0),Z(0),P_X(0),P_Y(0),P_Z(0))$. Using Hamilton's equations with the full Hamiltonian \eqref{eqn:ham} (setting also $J_1(0)=0$, $M_1(0)=0$), we obtain the numerical evolution $(X(t),Y(t),Z(t),$ $P_X(t),P_Y(t),P_Z(t))$, which  can be transformed to element evolution 
$$
z(t)=(a(t),e(t),i(t),f(t),g(t),h(t))\;.
$$
\\
\noindent
\textit{Semi-analytical propagation}: following the implementation of the normalization algorithm as described in the previous subsection, the initial osculating element state vector $z(0)$ is transformed into an initial condition for the corresponding `mean element' state vector $\xi^{(j)}(z(0))$, i.e., the element vector corresponding to the new canonical variables conjugated to the original ones after $j$ near-identity normalizing transformations. This is computed by the Lie series composition formula truncated at book-keeping order $N_\text{bk}$:
\begin{equation}\label{eqn:osctomean}
\xi^{(j)}(z)=\left(\exp\left(\mathcal{L}_{-\chi^{(1)}_{\nu}}\right)\circ\exp\left(\mathcal{L}_{-\chi^{(2)}_{\nu+1}}\right)\circ\ldots
\circ\exp\left(\mathcal{L}_{-\chi^{(j)}_{\nu+j-1}}\right) z \right)^{\leq N_\text{bk}}\;,
\end{equation}
using Eq.\eqref{eqn:Lieseriesinv} for the inverse series.
We then obtain the evolution of the mean element vector $\xi^{(j)}(t)$ through numerical integration of the \textit{secular} equations of motion
\begin{equation}
\label{eqn:eqmosec}
\dot{\xi}^{(j)}=\mathbb{J}\nabla\mathscr{Z}^{(j)}(\xi^{(j)})
\end{equation}
($\mathbb{J}$ standard symplectic unit).
This can be back-transformed to yield the evolution of the osculating element vector $z(t)$ using the truncated Lie series composition formula
\begin{equation}\label{eqn:meantoosc}
z(\xi^{(j)})=\left(\exp\left(\mathcal{L}_{\chi^{(j)}_{\nu+j-1}}\right)\circ\exp\left(\mathcal{L}_{\chi^{(j-1)}_{\nu+j-2}}\right)\circ\ldots
\circ\exp\left(\mathcal{L}_{\chi^{(1)}_{\nu}}\right) \xi^{(j)} \right)^{\leq N_\text{bk}}\;.
\end{equation}
Note that both the direct and inverse transformations (Eqs.\eqref{eqn:osctomean} and \eqref{eqn:meantoosc}), as well as Hamilton's secular equations \eqref{eqn:eqmosec}, can be computed in closed form, using the Poisson algebra rules of Subsection \ref{subsec:Poiss}. We then call semi-analytic the evolution of the element vector $z(t)$ obtained via the formula
\begin{equation}
\label{eqn:semiana}
z(t)= z(\xi^{(j)}(t))\;.
\end{equation}

\begin{figure}
	\centering
	\includegraphics[scale=0.41]{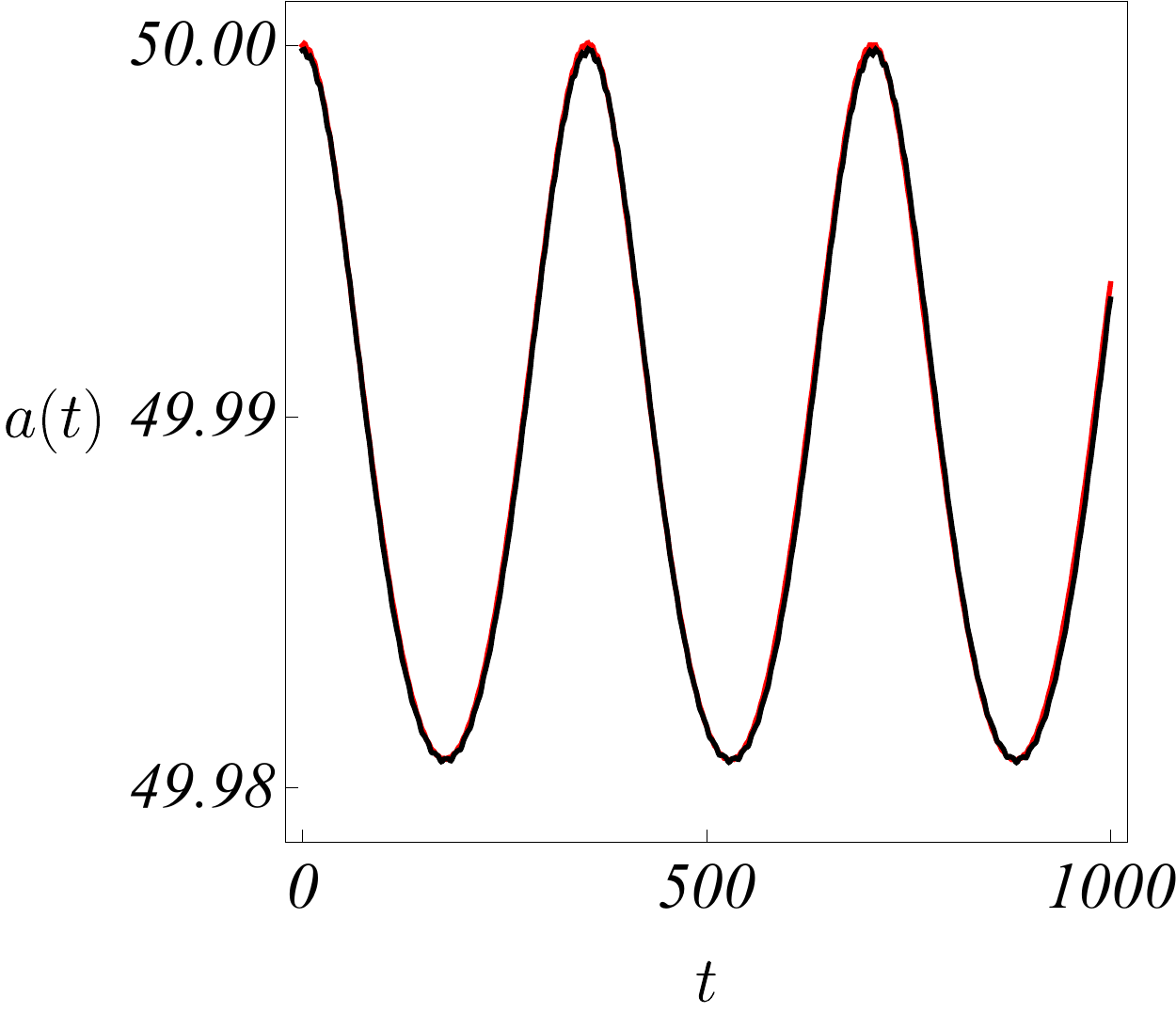}
	\includegraphics[scale=0.41]{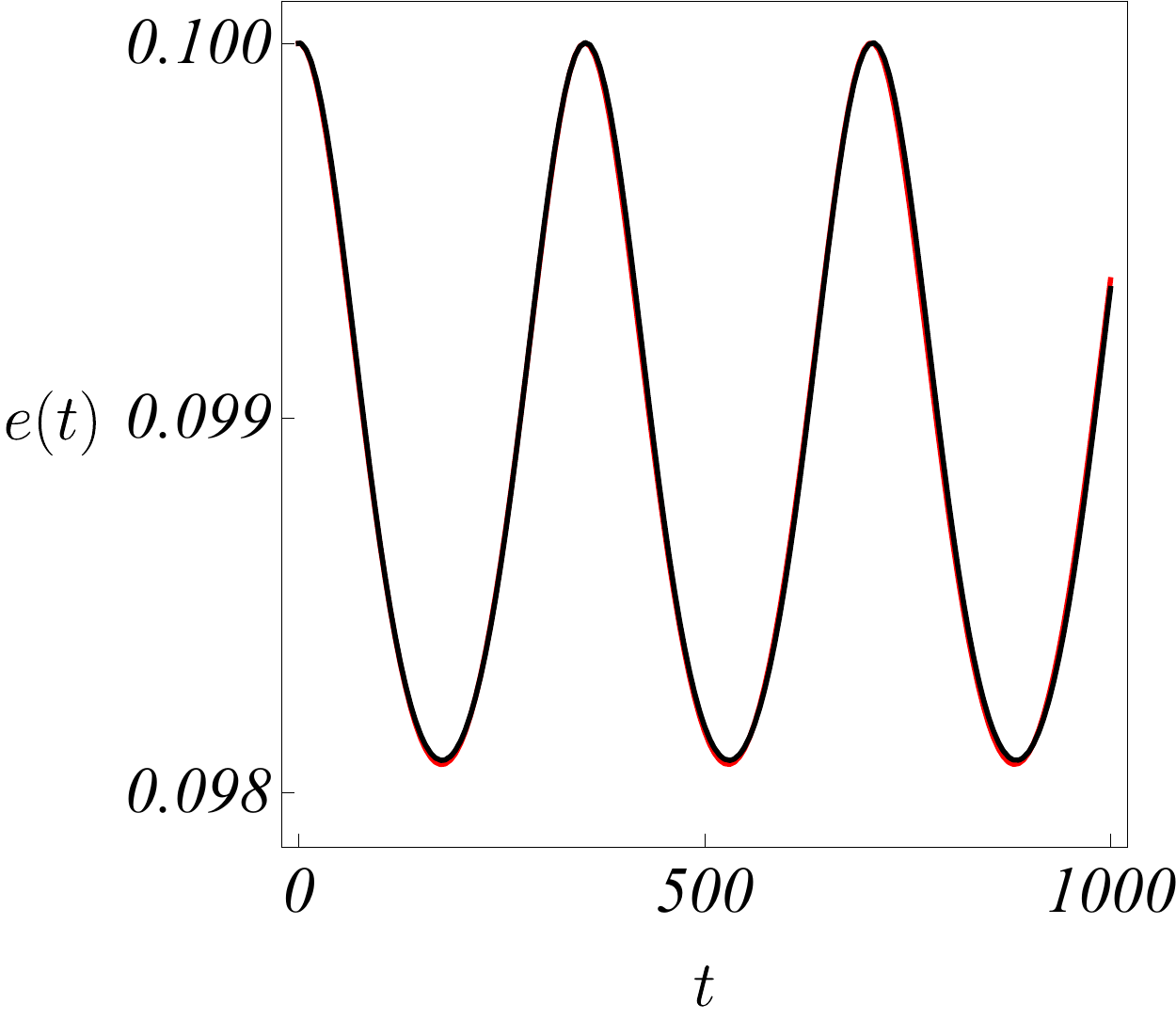}
	\includegraphics[scale=0.41]{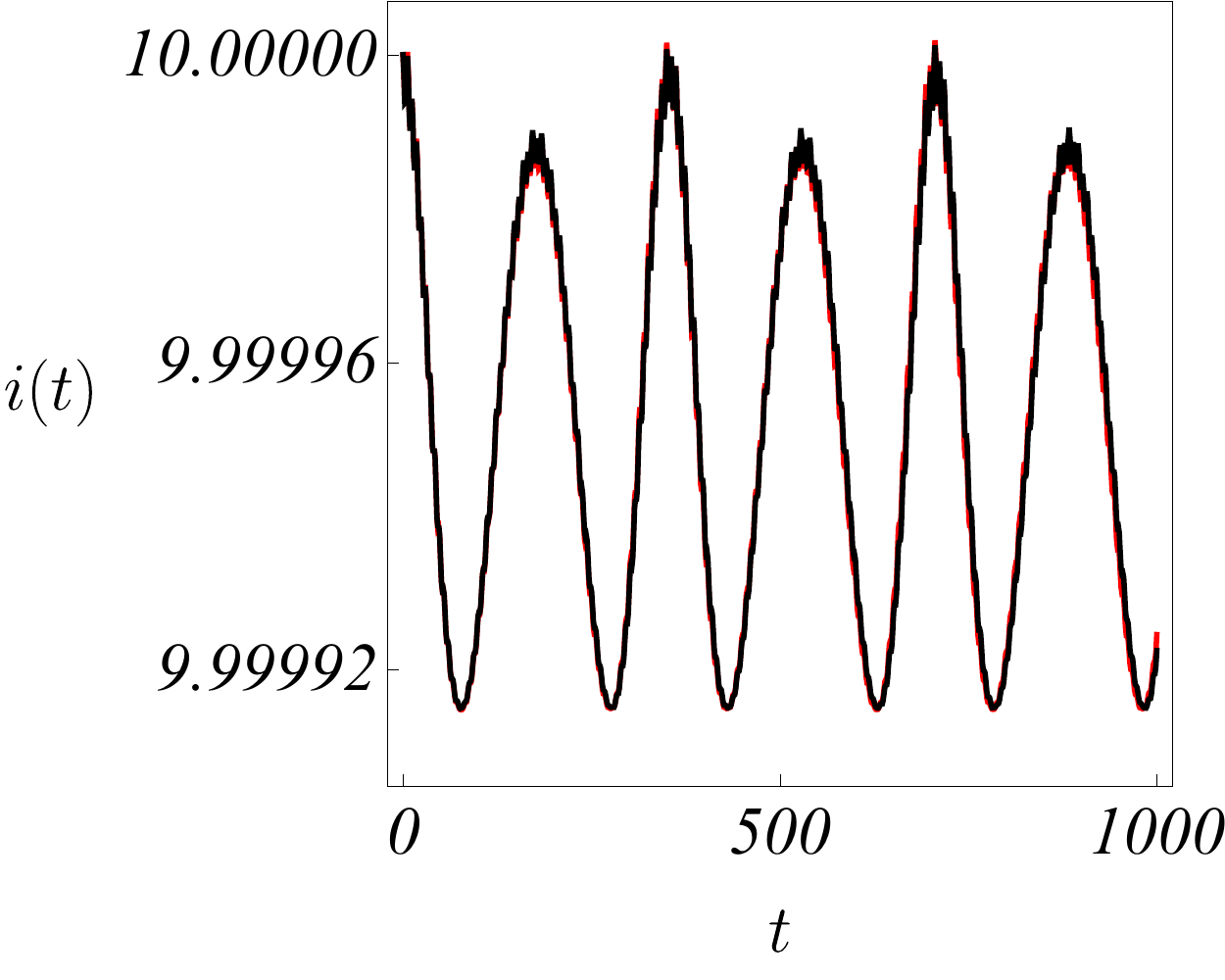}\\
	\includegraphics[scale=0.41]{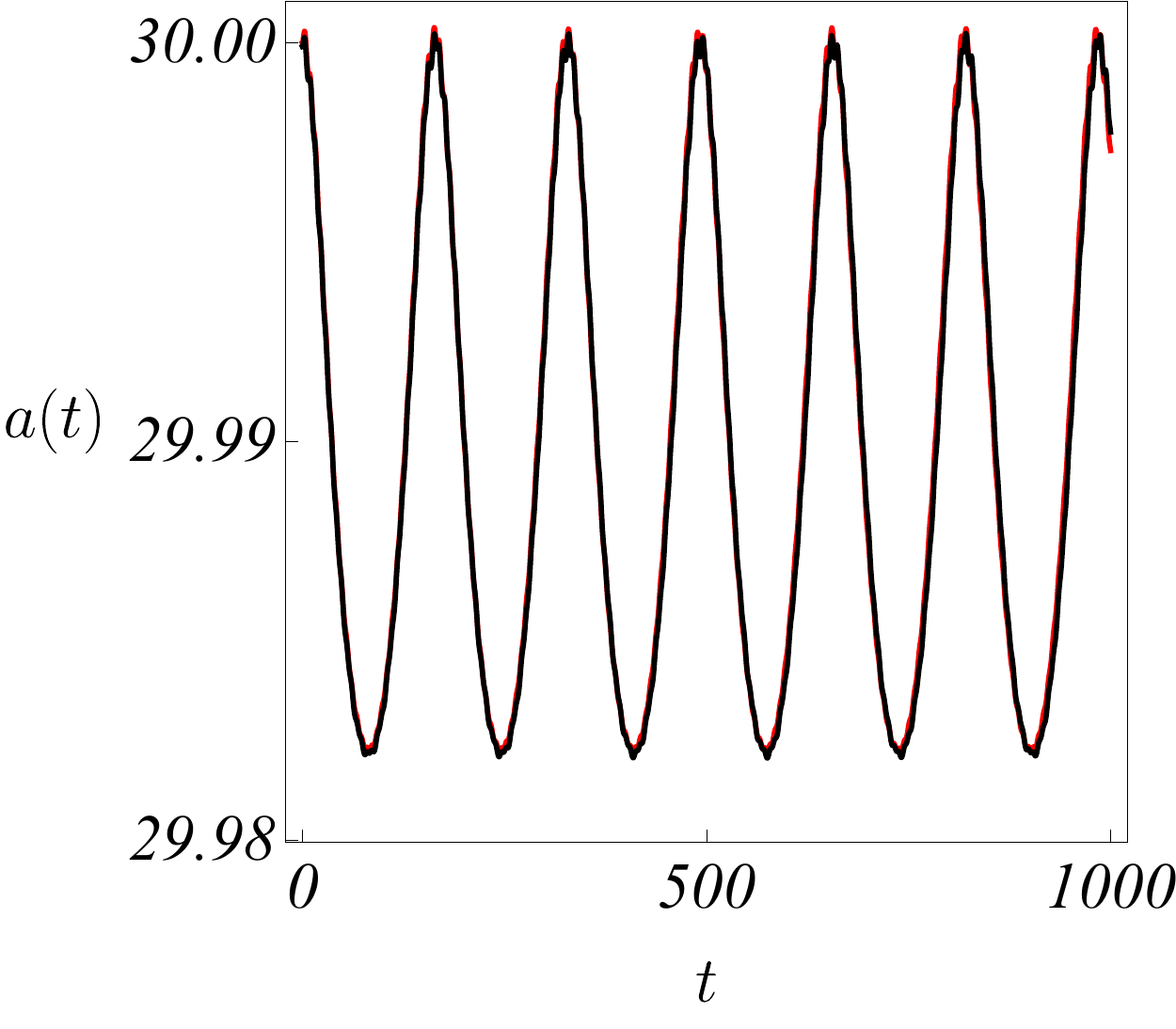}
	\includegraphics[scale=0.41]{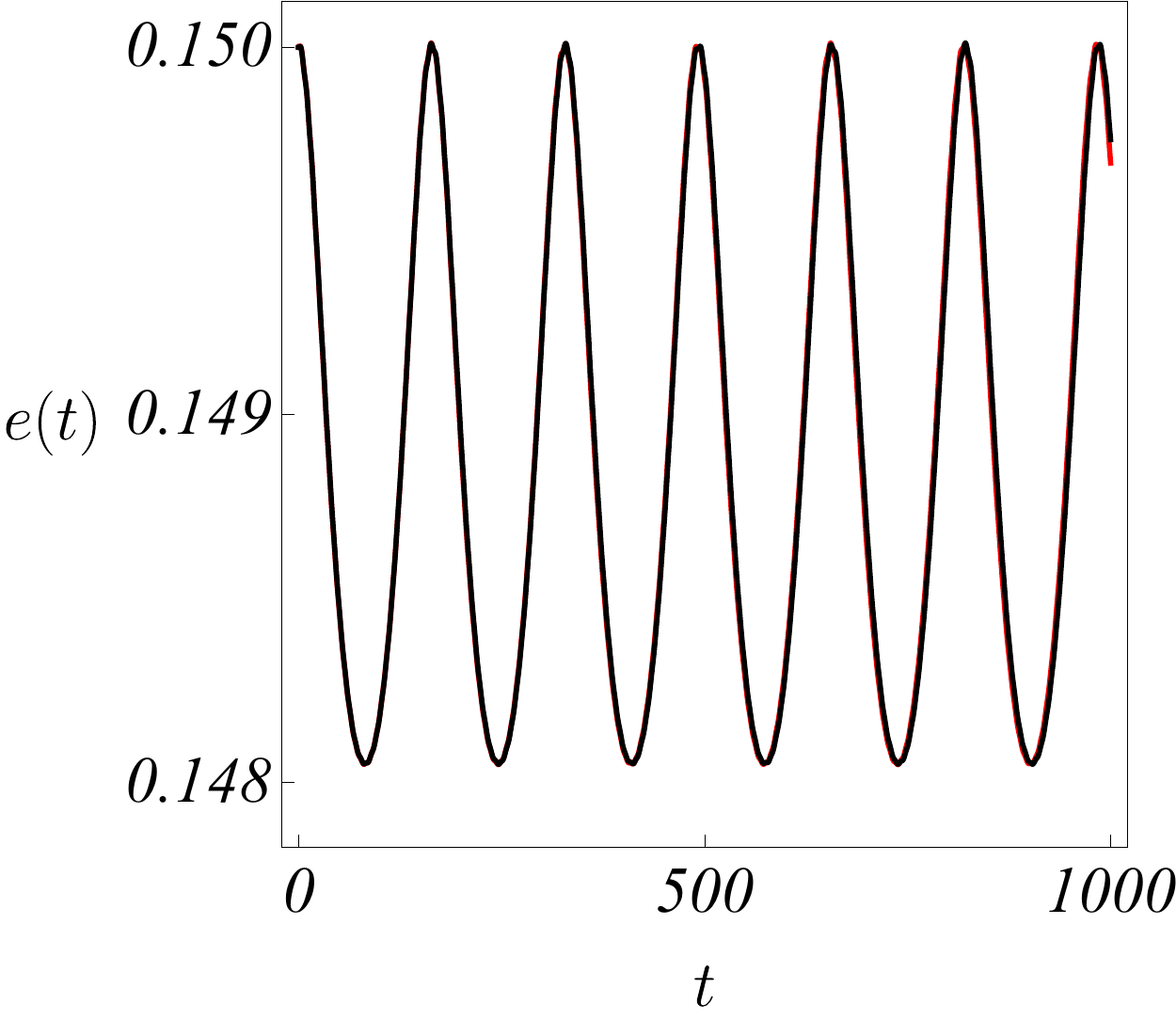}
	\includegraphics[scale=0.41]{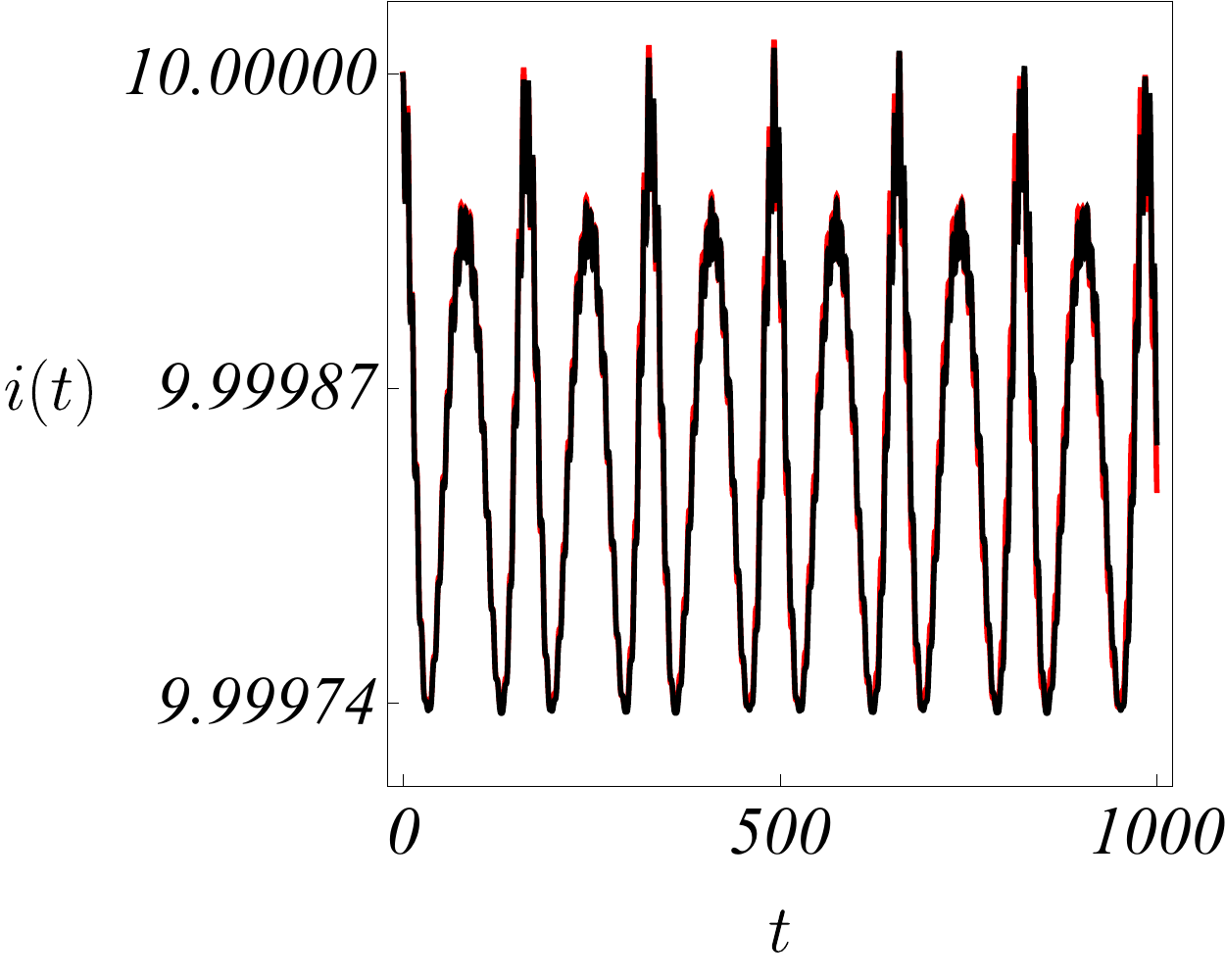}
	\caption{First and second example (ER3BP). Data: $a_*=50$AU, $e_*=0.1$ ($\nu=3$), $i(0)=10^{\circ}$, $k_{\mu}=k_{\text{mp}}=2$ (top panels); $a_*=30$AU, $e_*=0.15$ ($\nu=4$), $i(0)=10^{\circ}$, $k_{\mu}=k_{\text{mp}}=2$ (bottom panels). Black curves represent semi-analytic time variations (our method), while red curves stand for Cartesian series.}
	\label{fig:ex1}
\end{figure}
Fig. \ref{fig:ex1} shows the comparison between the Cartesian and the semi-analytical propagation of the elements in `easy' cases, where the particle departs from initial conditions $a(0)=50$AU (top left panel) or $a(0)=30$AU (bottom left panel), with a relatively low value of the eccentricity $e(0)=0.1$ or $e(0)=0.15$ respectively (middle panels) and inclination $i(0)=10^\circ$ (right panels). In these cases, the distance ratio $\norm{r_1}/\norm{R}$ is small (about $0.1$-$0.2$), a fact implying that the quadrupolar expansion ($k_\text{mp}=2$) suffices to have obtained a relative error of about $0.1\%$ in the representation of the Hamiltonian perturbation $\mathcal{H}_1$. Going to higher multipoles is straightforward, albeit with a significant computational cost as the number of terms in the Hamiltonian grows significantly. On the other hand, even with low-order truncations of the Hamiltonian we achieve to have an accurate semi-analytical representation of the $\mathcal{O}(\mu)$ short-period oscillations in all three `action-like' elements (semi-major axis, eccentricity, inclination). Most notably, keeping $a(0)=50$AU but changing the eccentricity to $e(0)=0.7$, i.e., beyond the Laplace value, yields an orbit whose pericenter is at $\norm{R_p}=15$AU, implying a distance ratio $\norm{r_1}/\norm{R}\approx 0.3$ (Fig. \ref{fig:ex2}). This time, an octupole truncation ($k_\text{mp}=3$) is required to produce an approximation of the Hamiltonian model at the level of a relative error of $0.1\%$. Still, however, as shown in Fig. \ref{fig:ex2} the semi-analytical propagation of the orbit is able to track the fully numerical one with an error which does not exceed $0.2\%$ even close to the orbit's pericentric passages. 
\begin{figure}
	\centering
	\includegraphics[scale=0.61]{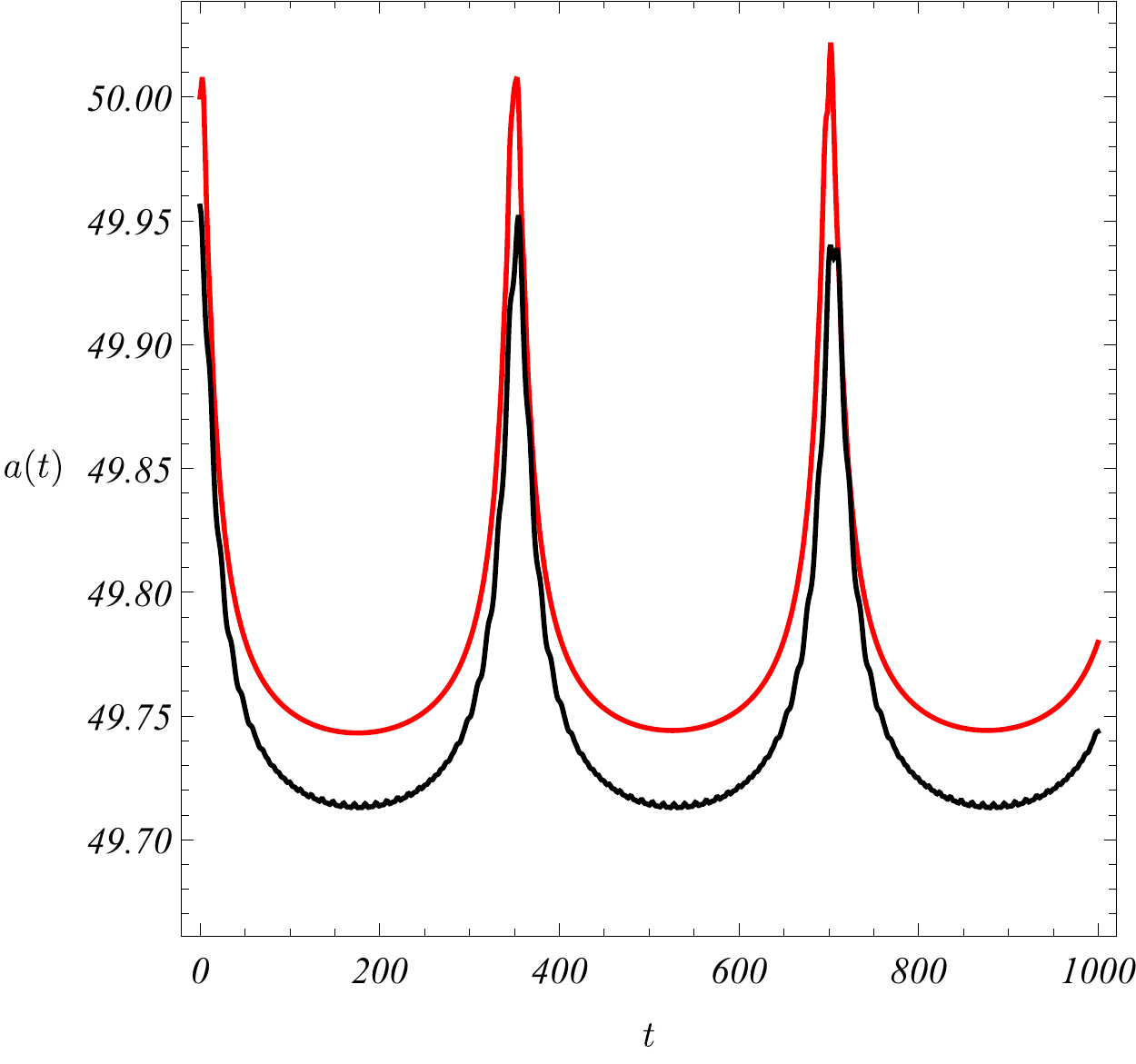}\hspace{3mm}
	\includegraphics[scale=0.61]{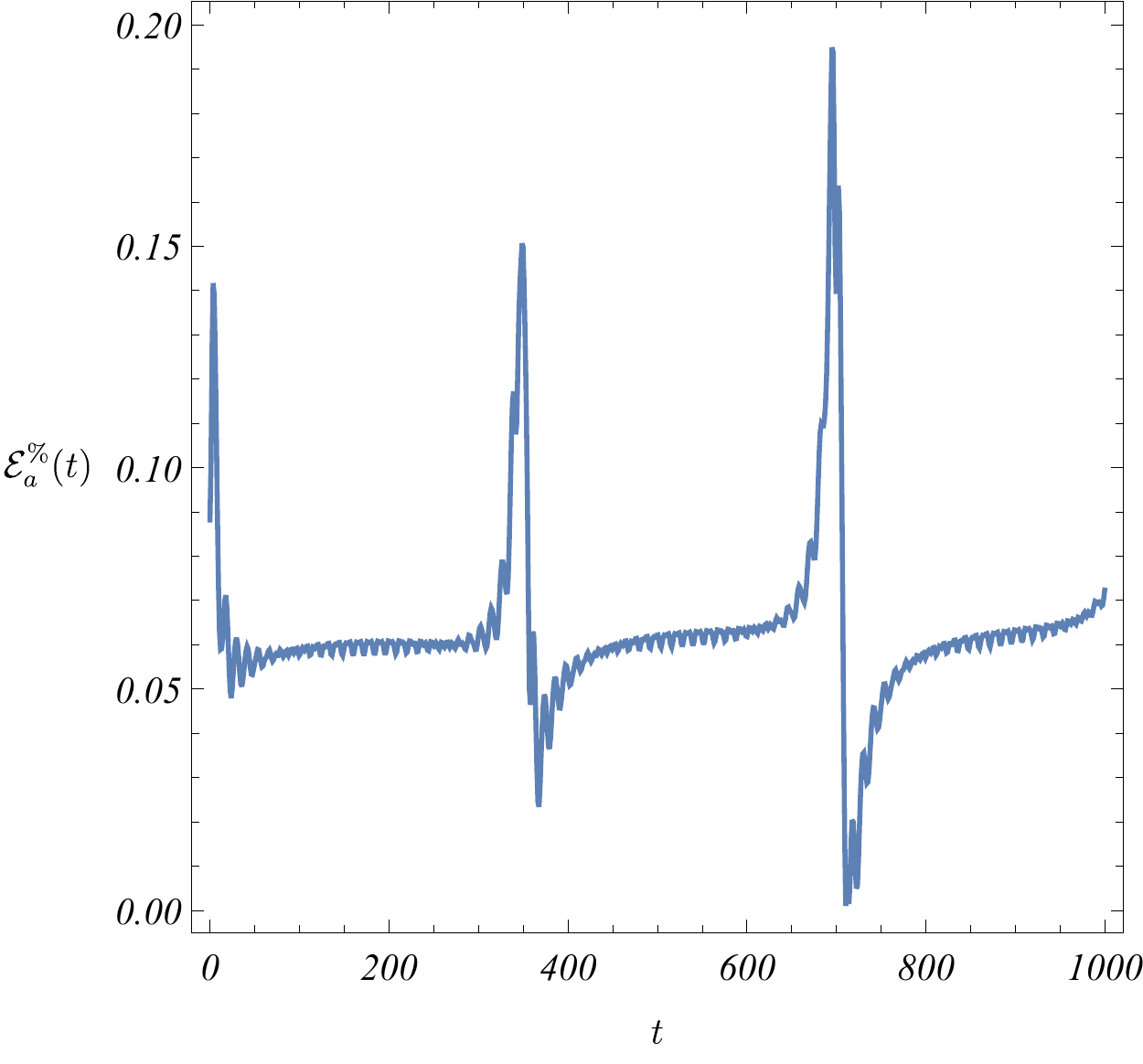}\\
	\caption{Third example (ER3BP). Data: $a_*=50$AU, $e_*=0.7$ ($\nu=20$), $i(0)=20^{\circ}$, $k_{\mu}=2$, $k_{\text{mp}}=3$. On the left, the black curve represents the semi-analytic time variation of the semi-major axis (our method) versus the one found by propagation of the Cartesian equations of motion (red). The right panel shows the evolution of the corresponding percent relative error $\mathcal{E}^{\%}_a$.}
	\label{fig:ex2}
\end{figure}

In the above examples, the maximum number of normalization steps at which the secular Hamiltonian is computed was set equal to $j_{max}=3$, $j_{max}=4$ and $j_{max}=4$ respectively, which corresponds to the best match in all cases. As discussed in the next subsection, an estimate of the \textit{minimum possible} error in the semi-analytic propagation of the trajectories requires computing first the so-called \textit{optimal} number of normalizations $j_{opt}$ (or equivalently \textit{optimal} normalization order $\nu+j_{opt}-1$) as a function of the reference values $(a_*,e_*)$ within a model given by a preset fixed multipole truncation order. Owing to the fact that the same divisors appear in the ER3BP and in the CR3BP, we verify with numerical examples that the error analysis yields essentially identical results in either case. However, the computation of the optimal normalization is easier to perform in the CR3BP, owing to the considerably smaller number of terms produced in the CAS computation of the normal form. Hence, we now turn our attention to this latter computation. 

\subsection{Numerical examples in the Sun-Jupiter planar CR3BP: order and size of the optimal remainder}
\label{subsec:SunJupcircular}
\subsubsection{Trajectory propagation: optimal remainder}
\label{subsubsec:SunJuptrajectories}
A considerable reduction of the computational cost occurs in the case of the planar and circular R3BP. This is due, in particular, to the following:
\begin{itemize}
	\item 
	the dependence on $M_1$ becomes explicit ($M_1=E_1$ in \eqref{eqn:r1}), while $a_1=\norm{r_1}$. As a consequence, $\phi_1=0$. 
	\item 
	no terms involving $(h,H)$ appear in the disturbing function, thus $\iota_c,\iota_s$ are discarded;
	\item 
	no terms requiring a book-keeping in terms of the exponent $\nu_1$ appear, hence, only $\nu$ is defined, as in \eqref{eqn:nun1};
	\item 
	$d^{\prime(j)}_{l,\lambda,p}=0$ for every $j,l,\lambda,p$ in \eqref{eqn:Rj}, \eqref{eqn:chij}, and consequently $p_1=p_2\equiv 0$ in \eqref{eqn:Znulm1j}. This is due to the fact that the expression \eqref{eqn:Rdotr1} reduces to
	\begin{equation}
	\label{eqn:Rdotr1pCR3BP}
	r_1\cdot R=\norm{r_1}\norm{R}\cos(f+g-M_1)\;,
	\end{equation}
	which always depends on the difference $g-M_1$ by D'Alembert rules. This implies that, unlike the ER3BP, the action $G$ (and the corresponding eccentricity $e$) are integrals of the secular Hamiltonian;  
	\item 
	as a consequence no lower or equal book-keeping order terms appear in any Poisson bracket of the first normalization step in the case $\nu=1$. Hence Proposition \ref{prop:adjust} is redundant.
\end{itemize}

Owing to the above, in the planar CR3BP we are able to make normal form computations in a grid of points in the plane $(a_*,e_*)$ up to a sufficiently high normalization order so that the asymptotic character of the series computed by the algorithm of Section \ref{sec:theory} can show up. To this end, we introduce an estimate of the size of the series' remainder after $j$ normalization steps via the upper norm bound 
\begin{equation}
\label{eqn:Rbound}
	\mathscr{E}^{(j)}=\sum_{l=\nu+j}^{\nu k_{\mu}}\sum_{s\in\mathbb{Z}^3}|d^{(j)}_{l,s}|\ge\norm{\mathscr{R}^{(j)}_{\nu+j}}_{\infty}\;,\quad j=1,\ldots,\nu(k_{\mu}-1)\;,
\end{equation}
where $\norm{\cdot}_{\infty}$ denotes the sup norm. Plotting $\mathscr{E}^{(j)}$ against the number of normalization steps $j$ allows then to estimate the error committed at any step (size of the remainder). Figure \ref{fig:ex3} yields an example of such computation. The relevant fact is that there is an optimal number of normalization steps ($j=j_{opt}=6$) where the estimate $\mathscr{E}^{(j)}$ of the remainder size yields a global minimum.  
\begin{figure}
	\centering
	\includegraphics[scale=0.56]{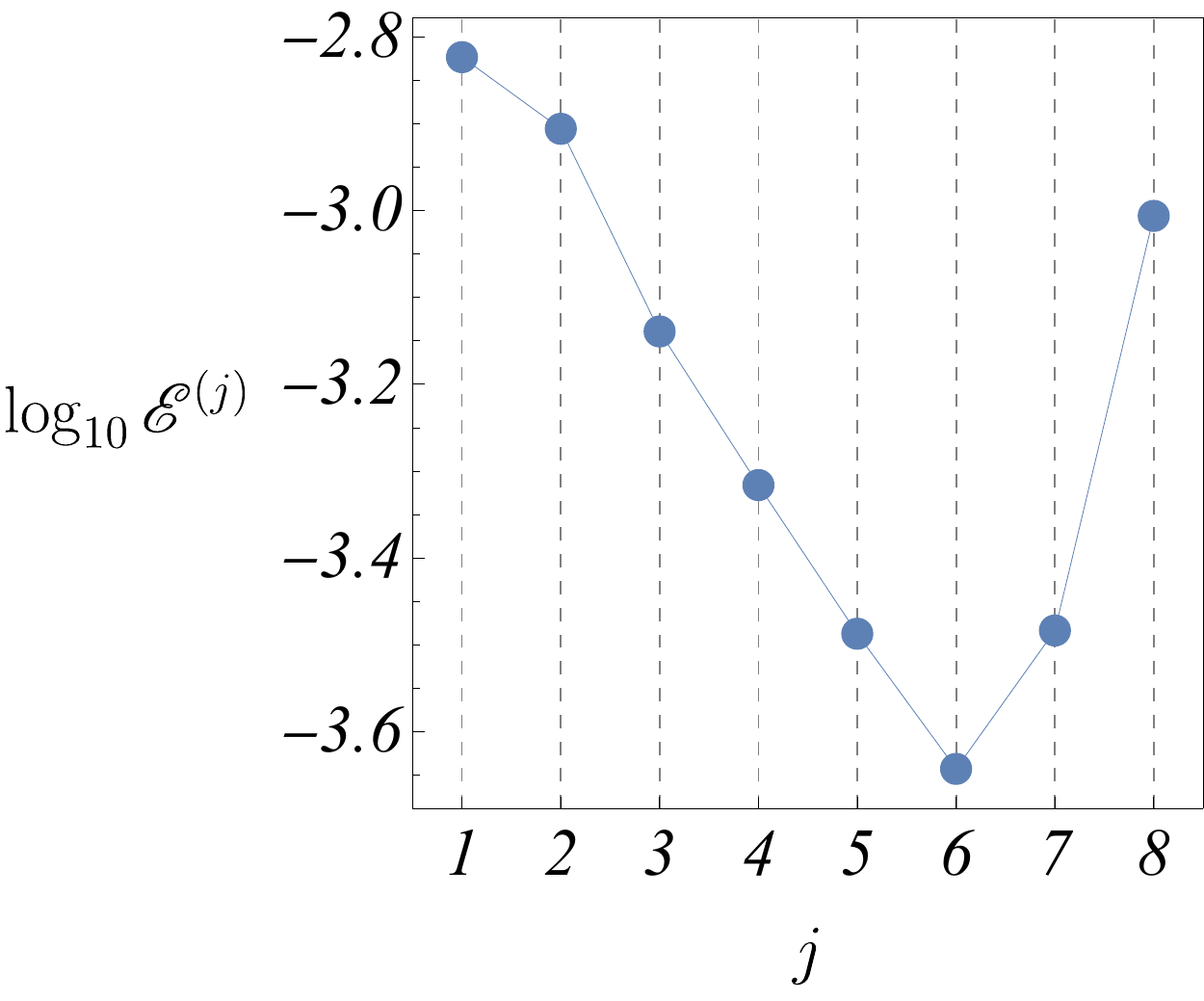}\hspace{3mm}
	\includegraphics[scale=0.52]{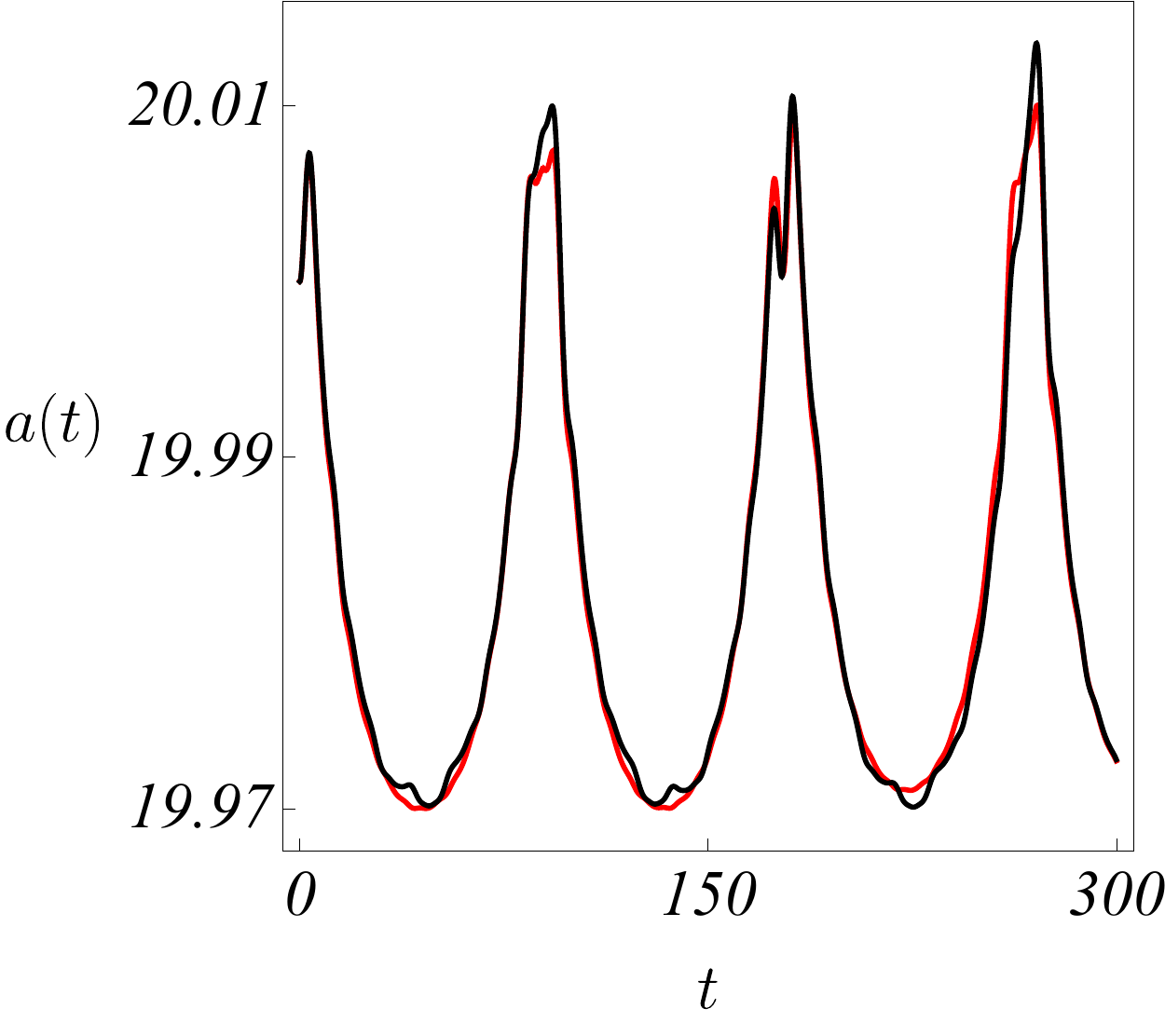}
	\\\hspace{5mm}
	\includegraphics[scale=0.52]{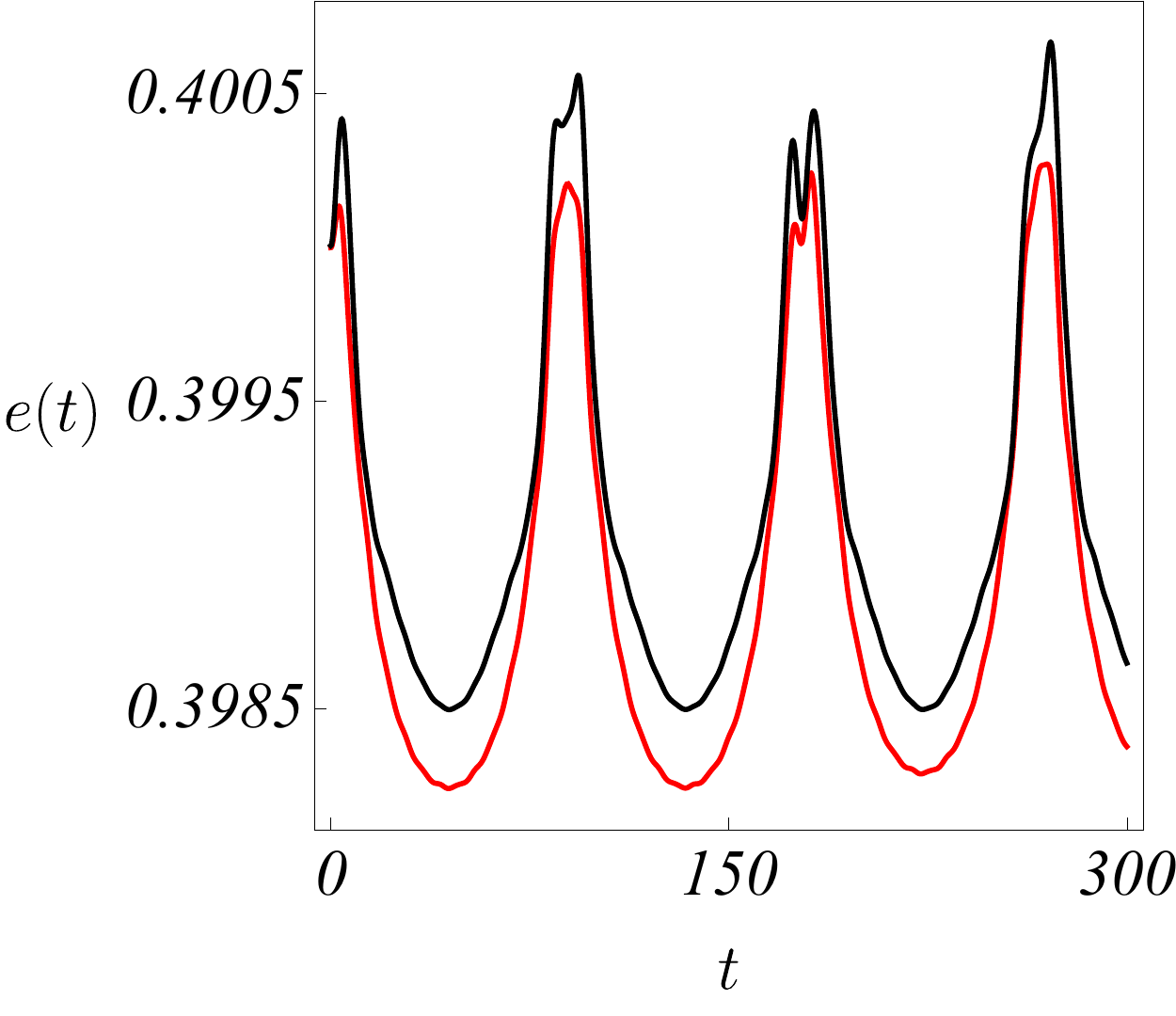}\hspace{1.5mm}
	\includegraphics[scale=0.52]{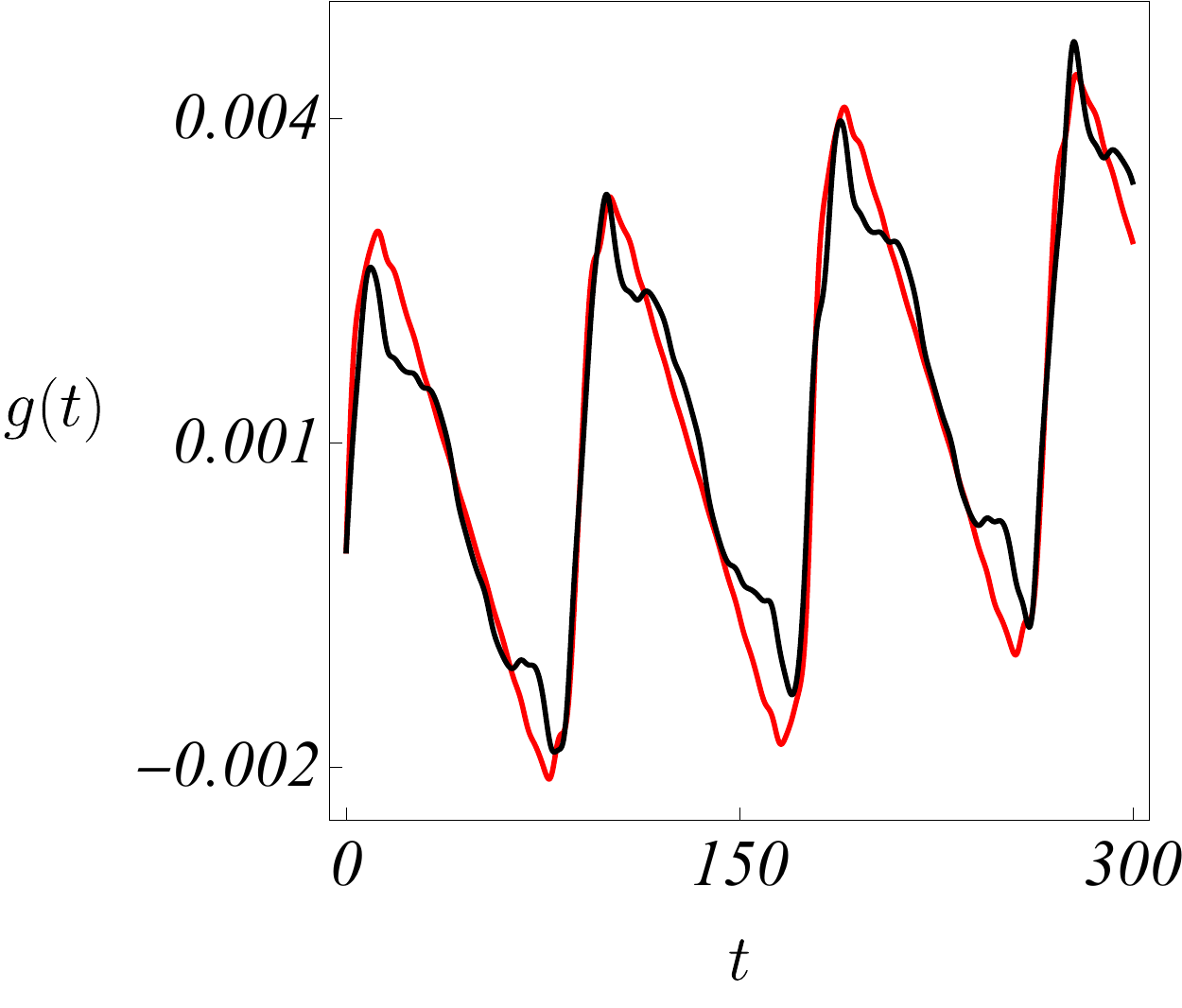}
	\caption{Fourth example (planar CR3BP). Data: $a_*=20$AU, $e_*=0.4$ ($\nu=8$), $k_{\mu}=2$, $k_{\text{mp}}=3$. The estimate $\mathscr{E}^{(j)}$ is depicted in semi-logarithmic scale on top left panel. The direct comparison of the semi-analytic (black) evolution vs. the fully numerical (red) one for the osculating elements $a(t)$, $e(t)$, $g(t)$ are shown in the top right and bottom panels respectively. The semi-analytic curves are obtained for $j=j_{opt}=6$, where $\mathscr{E}^{(j)}$ is minimum.}
	\label{fig:ex3}
\end{figure}

Although a systematic investigation of the dependence of the optimal number of normalization steps $j_{opt}$ on the parameters $(a_*,e_*)$ is beyond our present scope, Figs. \ref{fig:ex4} and \ref{fig:ex5} allow to gain some insight into the question. The most relevant remark concerns the dependence of the behavior of the curve $\mathscr{E}^{(j)}$ (versus $j$) on how close to the `hierarchical' regime the trajectory with reference values $(a_*,e_*)$ is. As a measure of the hierarchical character of an orbit we adopt either the ratio of the semi-major axes $a_1/a_*$, or of the pericentric distances $\norm{r_1}/\norm{R_p}=a_1(1-e_1)/(a_*(1-e_*))=a_1/(a_*(1-e_*))$. Fig. \ref{fig:ex4} ($a_*=30$AU, $e_*=0.5$) implies a pericentric distance ratio $\norm{r_1}/\norm{R_p}\approx 0.3$ smaller than the one of the example of Fig. \ref{fig:ex3} ($\norm{r_1}/\norm{R_p}\approx 0.4$). We observe that the optimal number of normalization steps in the former case satisfies  $j_{opt}=10$, i.e., it is larger than in the latter case. Fig. \ref{fig:ex5} shows, instead, an example of orbit far from the hierarchical limit, satisfying the estimate $\norm{r_1}/\norm{R_p}\approx 0.7$. In this case a higher order multipole expansion ($k_\text{mp}=5$) is required to obtain a precise truncated Hamiltonian model for this orbit. We note, however, that the normalization procedure performs well, producing a decreasing remainder as a function of $j$ up to the point where it is arrested, i.e. $j=6=\nu(k_{\mu}-1)$. We find numerically that this performance is deteriorated as we gradually approach the condition $\norm{r_1}/\norm{R}=1$, beyond which the multipole expansion of the Hamiltonian is no longer convergent. 
\begin{figure}
	\centering
	\includegraphics[scale=0.56]{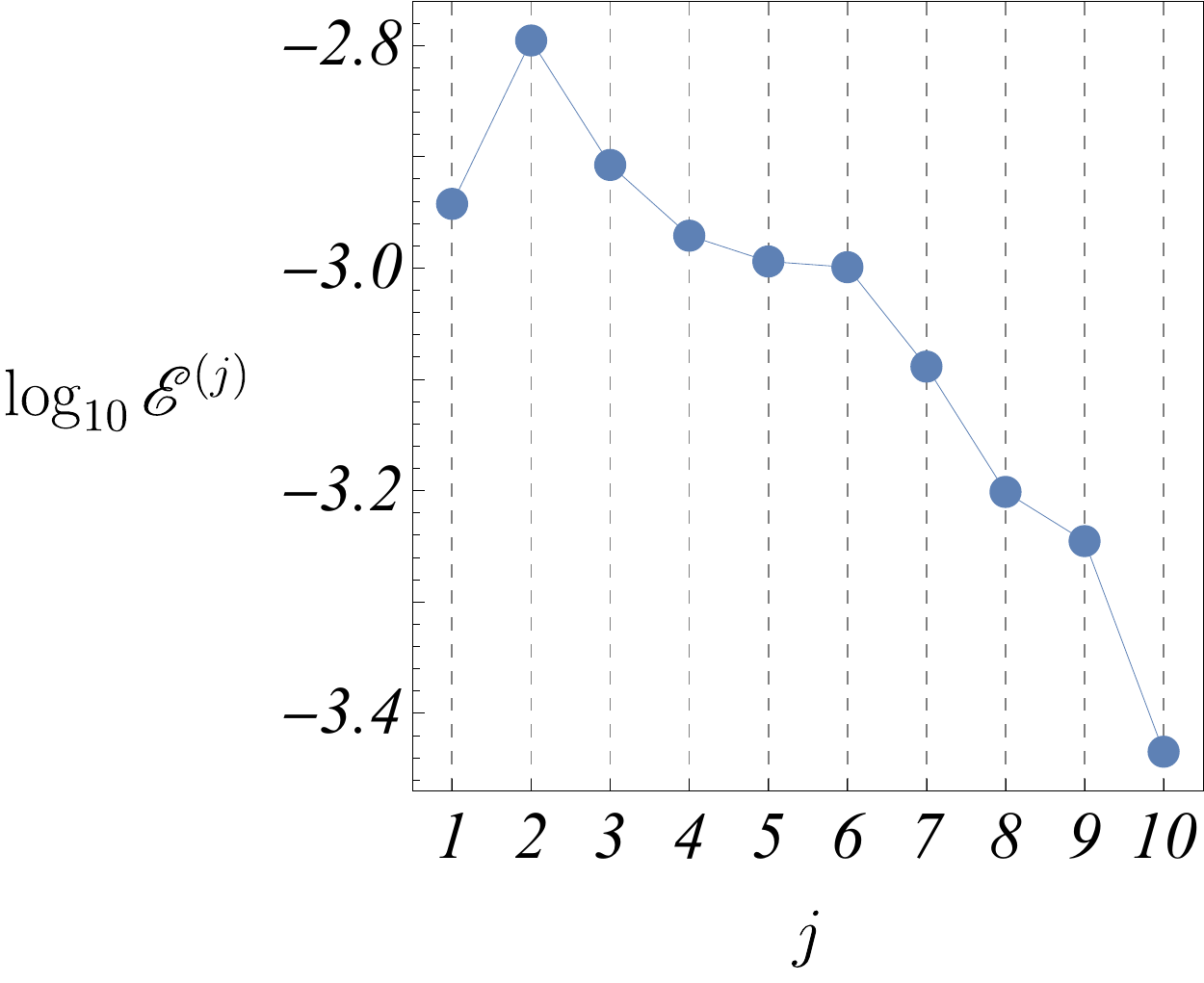}\hspace{3mm}
	\includegraphics[scale=0.52]{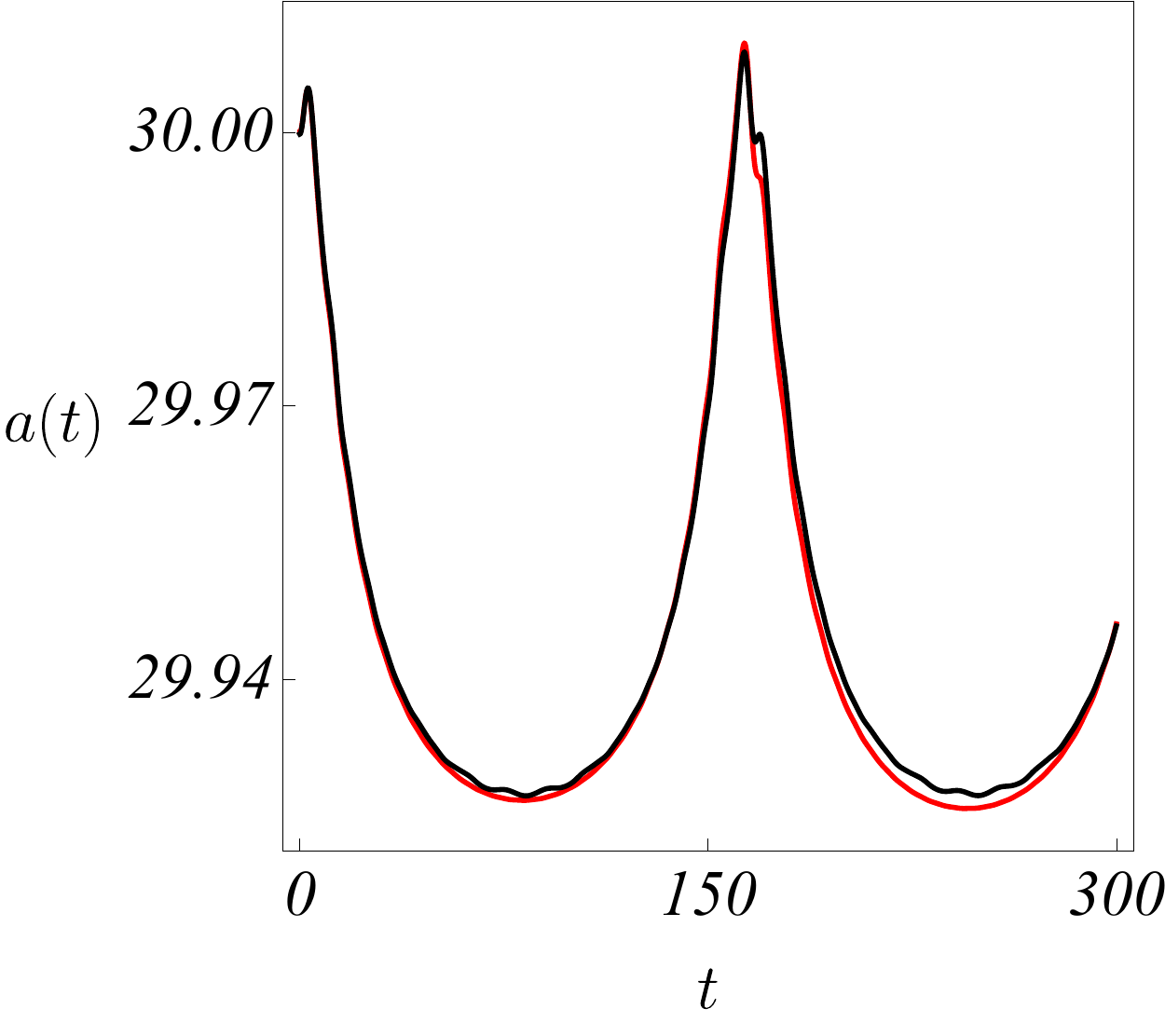}\\
	\hspace{6mm}
	\includegraphics[scale=0.52]{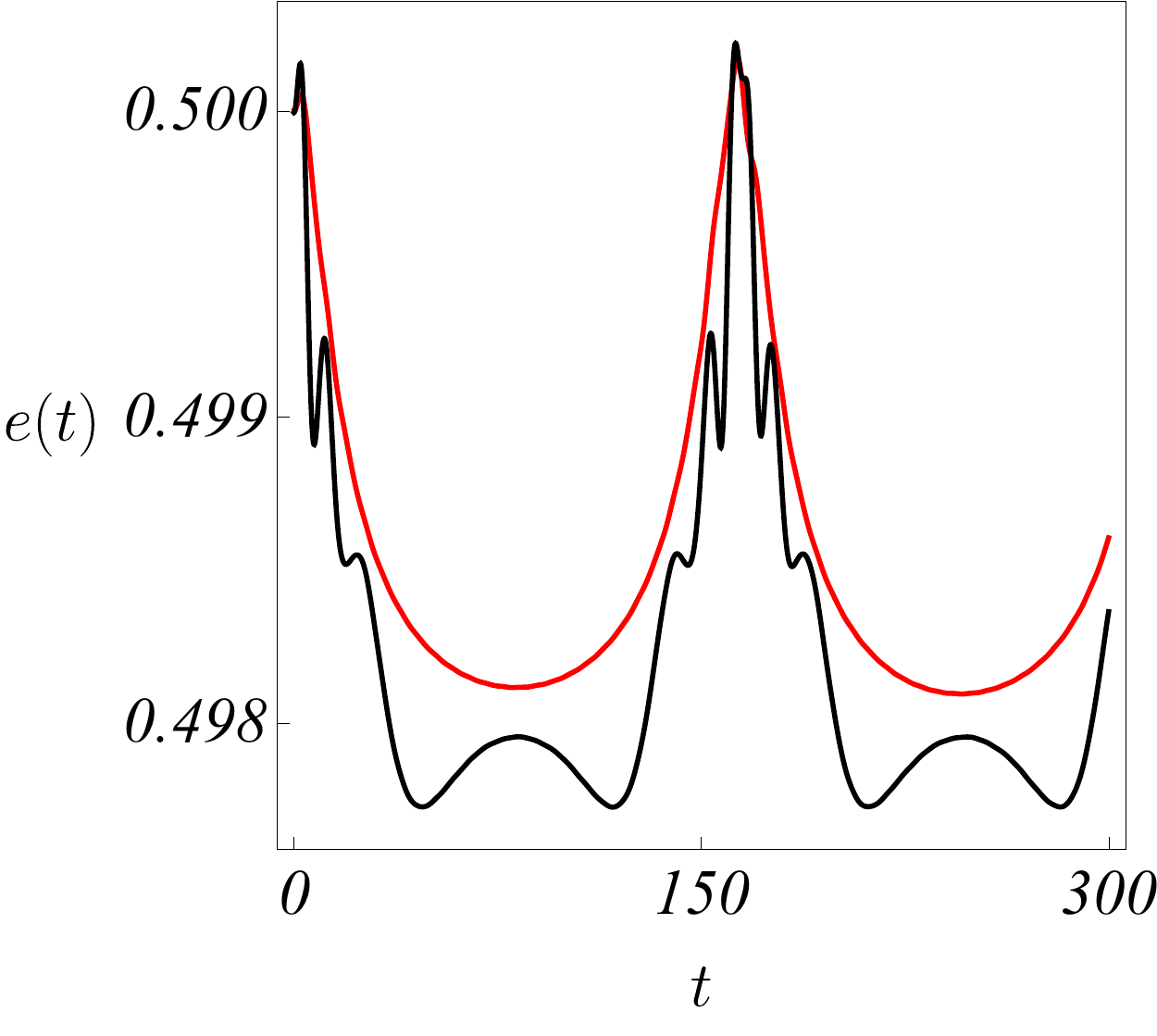}\hspace{1.5mm}
	\includegraphics[scale=0.52]{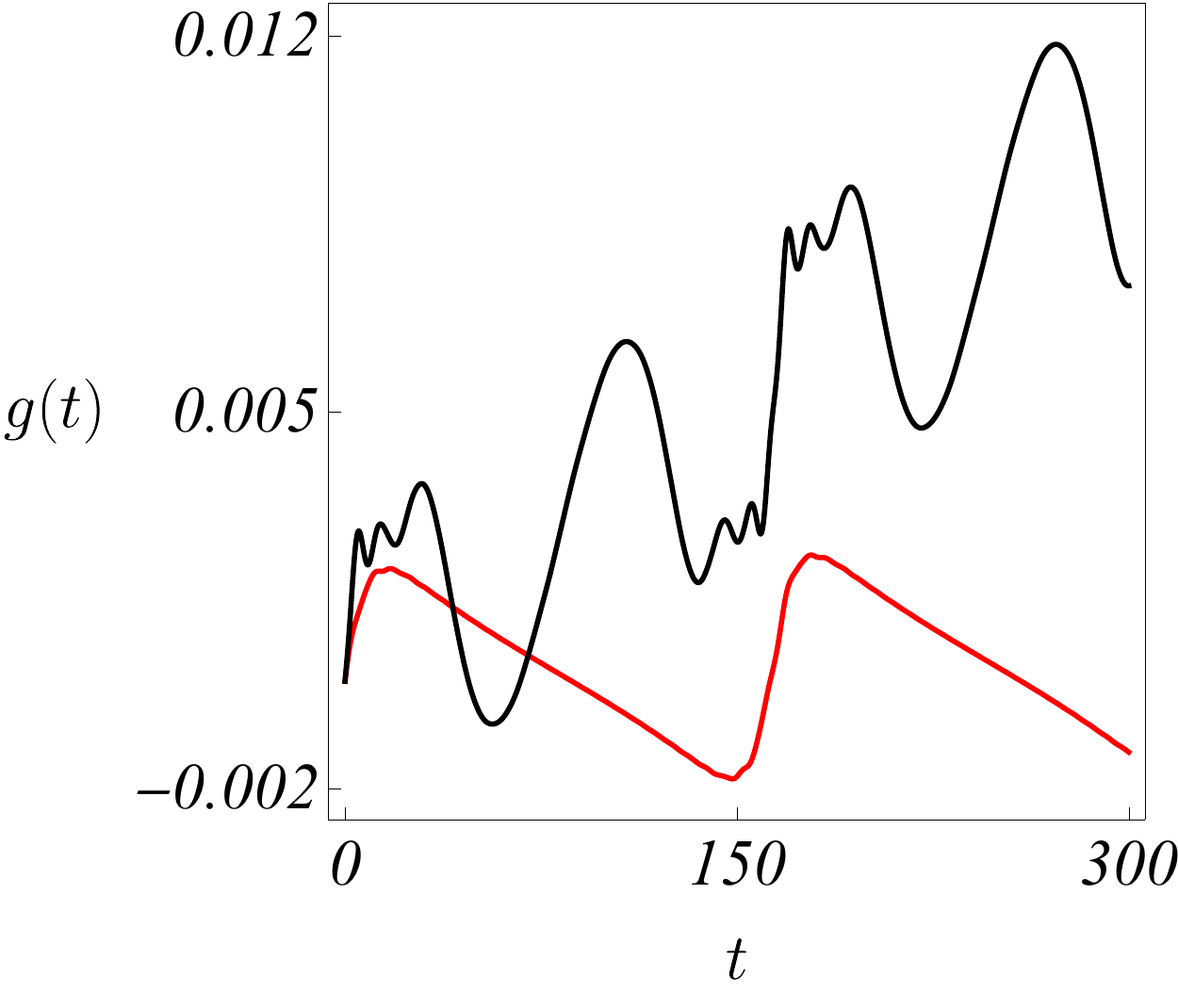}
	\caption{Fifth example (planar CR3BP): $a_*=30$AU, $e_*=0.5$ ($\nu=10$), $k_{\mu}=2$, $k_{\text{mp}}=3$. Plot types and color conventions are the same as in Fig. \ref{fig:ex3}. The semi-analytic curves are obtained for $j=j_{opt}=10$.}
	\label{fig:ex4}
\end{figure}
\begin{figure}[h]
	\centering
	\includegraphics[scale=0.56]{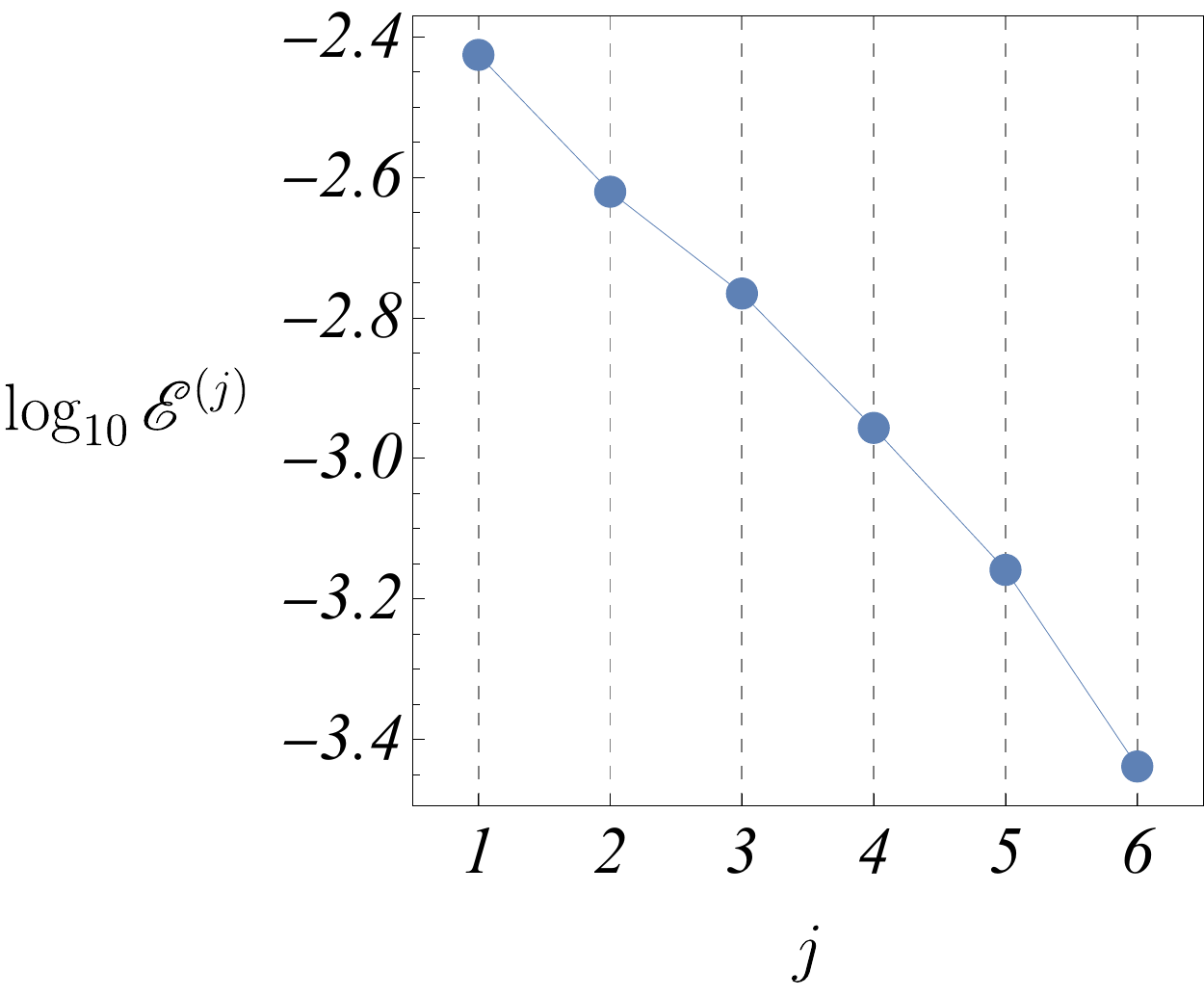}\hspace{3mm}
	\includegraphics[scale=0.52]{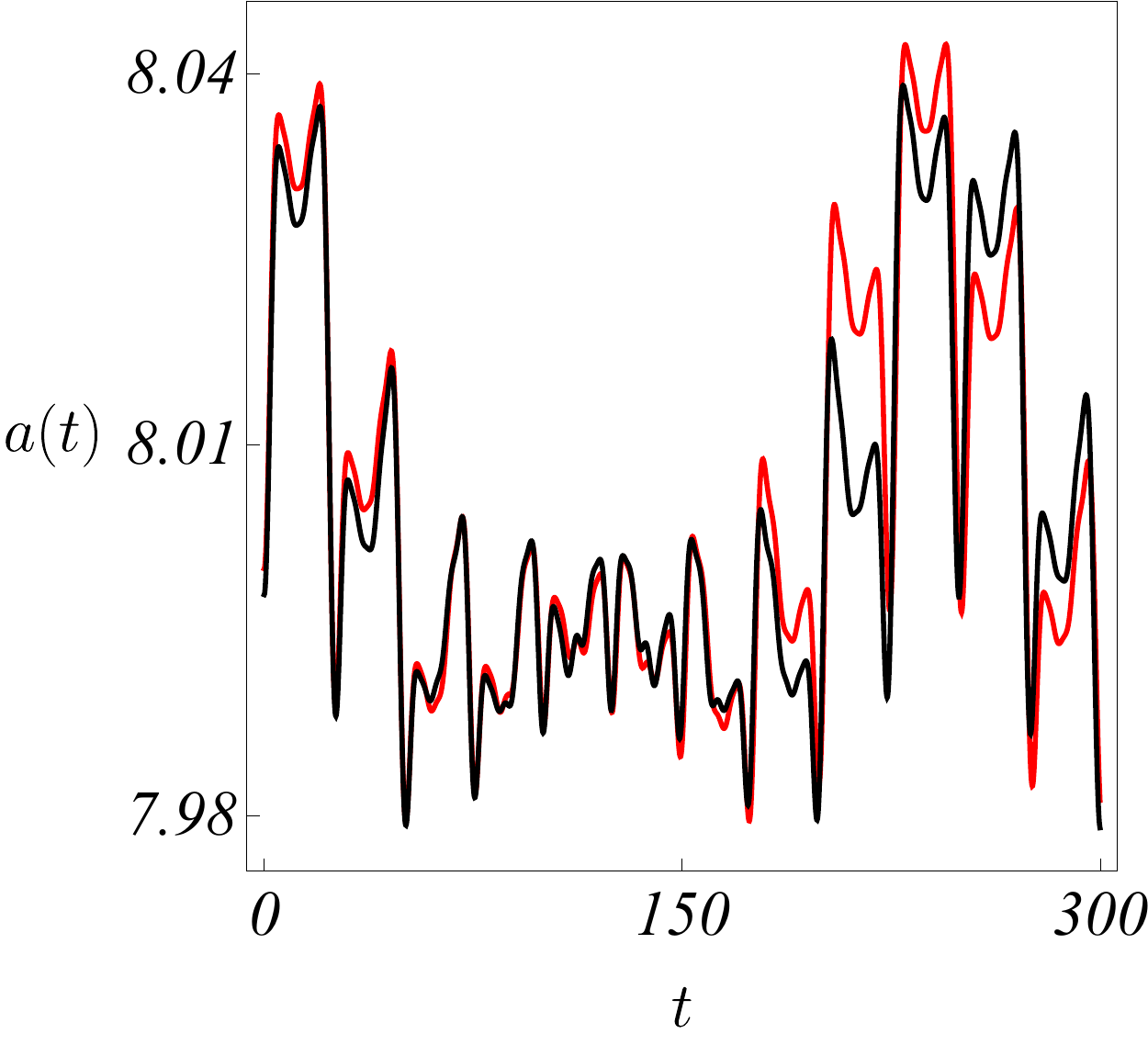}\\
	\hspace{6mm}
	\includegraphics[scale=0.52]{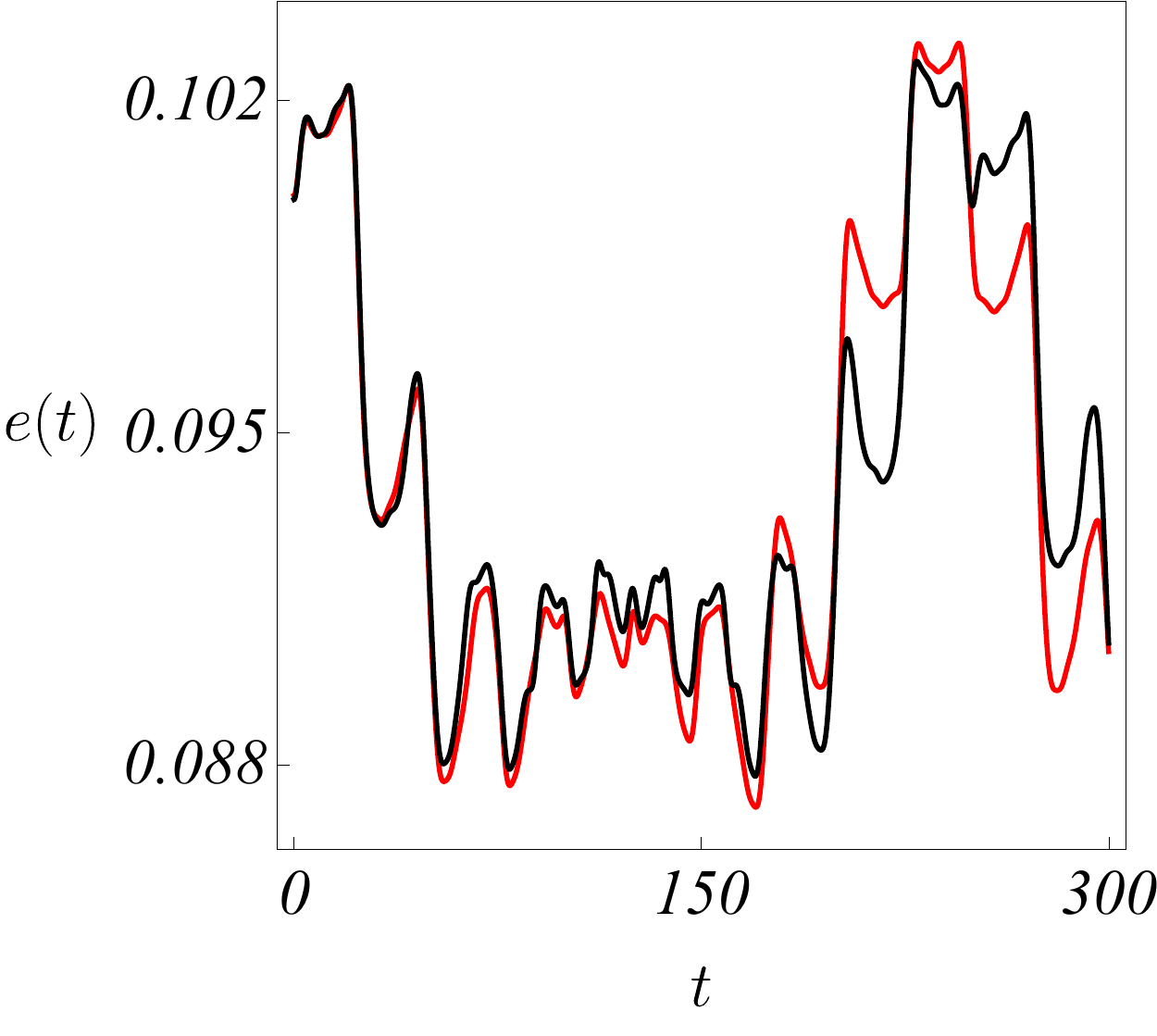}\hspace{1.5mm}
	\includegraphics[scale=0.52]{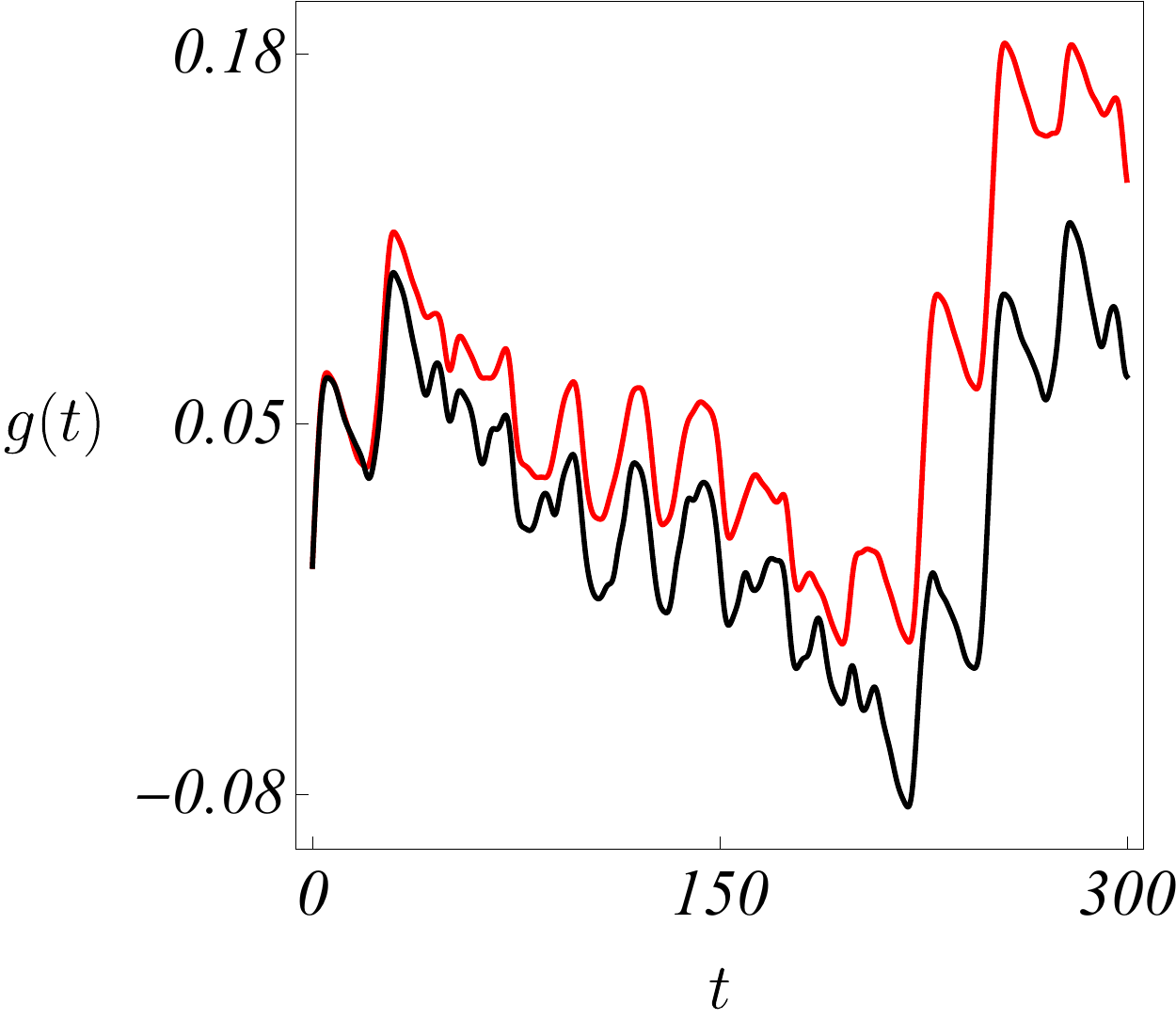}
	\caption{Sixth example (planar CR3BP): $a_*=8$, $e_*=0.1$ ($\nu=3$), $k_{\mu}=3$, $k_{\text{mp}}=5$. Plot types and color conventions are the same as in Fig. \ref{fig:ex3}. The semi-analytic curves are obtained for $j=j_{opt}=6$.}
	\label{fig:ex5}
\end{figure}

\subsubsection{Semi-analytical determination of the domain of secular motions}
\label{subsubsec:SunJupsecdomain}
The results shown in the two previous subsections refer to isolated examples of orbits treated within various multipole truncation orders as well as different choices of the number of normalization steps, searching each time to arrive at the best approximating secular model given computational restrictions. In the present subsection, we aim to investigate the behavior of the remainder in a closed-form normalization with uniform choice of all truncation orders of the problem, but performed, instead, in a fine grid ($100\times 20$) of reference values in the plane $(a_*,e_*)$. To this end, we set $k_\mu=2$ (second order in the mass parameter), and fix $k_\text{mp}=3$ (octupole approximation). The latter choice, imposed by computational restrictions, yields an initial model whose error with respect to the full Hamiltonian becomes of the order of $1\%$ only for $a_*>2a_1$. However, for reasons explained below, a computation within the framework of the octupole approximation becomes relevant to the problem addressed in the sequel also in the range $1.5a_1<a_*<2a_1$, while higher multipoles are required to address still smaller values of $a_*$. 

The result of the above computation is summarized in Fig. \ref{fig:err2d}: the left panel shows in logarithmic color scale the size of the remainder, estimated by the value of $\mathscr{E}^{(n)}(a_*,e_*)$ computed as in \eqref{eqn:Rbound}, corresponding to each point in the plane $(a_*,e_*)$, where the number of normalization steps is set as $n=\min\{\nu(k_{\mu}-1),7\}=\min\{\nu,7\}$. The maximum value $n=7$ is, again, imposed by computational restrictions, and it implies that $n$ varies with $e_*$ up to about $e_*=0.37$. 

The relevant information in Fig. \ref{fig:err2d} is provided by the black curve, which corresponds to the isocontour 
$\mathscr{E}^{(n)}(a_*,e_*)=10^{-2}$. Since in the original Hamiltonian we have the estimate $\mathscr{E}^{(0)}(a,e)\coloneqq\mathcal{H}_1^{\le k_{\mu},k_{\text{mp}}}=\mathcal{O}(10^{-2})$, the black curve provides a rough estimate of the limiting border dividing the plane $(a_*,e_*)$ in two domains: in the one below the black curve the progressive elimination of the fast angles by the iterative normalization steps leads to a secular model whose remainder decreases with the number of normalization steps $j$ at least up to $j=n$.  

A physical interpretation of the border approximated through the isocontour $\mathscr{E}^{(n)}(a_*,e_*)=10^{-2}$ can be given through a comparison with a numerical stability map obtained, e.g., as in the right panel of Fig. \ref{fig:err2d}. For each trajectory in a $300\times900$ grid in $(a,e)$, the plot shows in color scale the value of the Fast Lyapunov Indicator (FLI, see \cite{lega2016theory} for a review) obtained after integrating the variational equations of motion together with the equations of motion of the full Hamiltonian model for a time equal to $50$ periods of Jupiter. Thus, deep blue colors indicate the most regular, and light yellow the most chaotic orbits as identified by the value of the FLI. Superposed to the FLI cartography are three curves: 
\begin{enumerate}[label=(\roman*)]
\item\label{item:aps} the `perihelion crossing curve' (red) yields the locus of values satisfying the condition $a(1-e)=\norm{r_J}=a_J$ (in the circular case), that is the points where the pericenter of the test particle's orbit comes at distance equal to the radius of Jupiter's orbit;
\item\label{item:Hill} the Hill limit \cite{ramos2015resonance} (brown) is based on the relationship $C_{\text{Jac}}(a,e)=C_{\text{Jac}}(\mathscr{L}_1)$, where $C_{\text{Jac}}$ is the particle's Jacobi constant as function of the orbital elements and $C_{\text{Jac}}(\mathscr{L}_1)$ its value at the Lagrangian point $\mathscr{L}_1$;
\item\label{item:1perc} the isocontour $\mathscr{E}^{(n)}(a,e)=10^{-2}$ (black, same as in the left panel of Fig. \ref{fig:err2d}). 
\end{enumerate}
\begin{figure}
	\centering
	\includegraphics[scale=0.7]{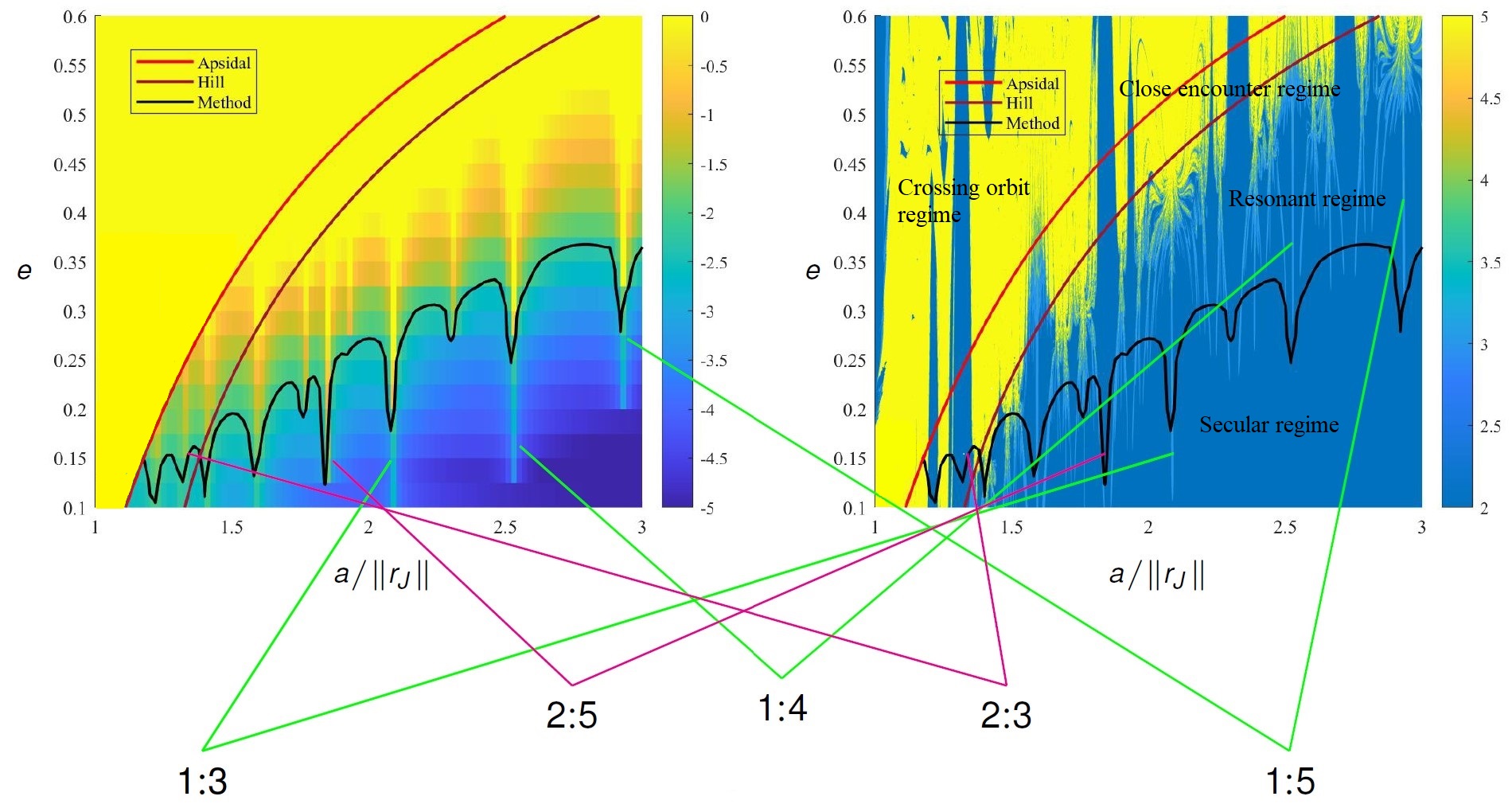}
	\caption{Left panel: computation of $\log_{10}(\mathscr{E}^{(n)})$, $n=\min\{\nu,7\}$, $k_{\text{mp}}=3$, over a $100\times20$ $(a,e)$ grid. For every $e=e_*$, $n$ different normalizations are executed and then evaluated for each $a=a_*$. Right panel: short-period FLI map over a $300\times900$ $(a,e)$ grid of initial data integrated for $50T_1$. As indicated, the three curves represent, respectively, the line of constant pericenter of the particle's trajectory equal to the radius of Jupiter's orbit $\norm{r_1}=\norm{r_J}$ (red), Hill's stability criterion (brown) and the isolevel $\mathscr{E}^{(n)}=1\%$ (black). Each region enclosed by two consecutive above curves is labeled with the corresponding regime of motion. The main mean-motion resonances are reported below the pictures.}
	\label{fig:err2d}
\end{figure}

Of the above three curves, the perihelion crossing curve is analogous, in the R3BP, of the so-called Angular Momentum Deficit criterion (AMD, \cite{laskar2017amd}) used to separate systems protected from perihelia crossings in the case of the full planetary three-body problem. As indicated by the FLI cartography data, Hill's curve gives an overall better approximation separating the domain of strong chaos (yellow) from the domain of regular or weakly-chaotic orbits (all blue nuances). This is expected, since the Hill's curve separates orbits for which Jupiter's gravitational effect becomes (at least temporarily) dominant from those for which it does not. Nevertheless, through the FLI cartography we note the presence of a large domain between the curves \ref{item:Hill} and \ref{item:1perc}, where the trajectories, while protected from close encounters, are subject to the long term effects on dynamics produced by resonant multiplets associated with the mean-motion resonances of the problem (the most important of which are marked in the figure). Note that in the octupole approximation, the Hamiltonian contains harmonics including all combinations of the fast angles of the form $\cos(s_1f + s_2(g - M_1))$, with 
$$
(s_1,s_2)=(1,3),(2,3),(3,3),(4,3),(5,3),(6,3),(7,3),   
$$
$$
(1,2),(2,2),(3,2),(4,2),(5,2),   (1,1),(2,1),(3,1),(4,1),(5,1),   
$$
$$
(1,-1),(2,-1),(3,-1),   (1,-2),(1,-3),
$$
thus including all harmonics associated with the mean-motion resonances detected in the FLI cartography of Fig. \ref{fig:err2d} for $a>1.5a_J$.  Through the closed-form normalization (Eqs.(\ref{eqn:chi1}) and (\ref{eqn:chij})) we then obtain small divisors in the series at every value of the semi-major axis $a_*$ for which one of the resonant combinations $s_1n_*-s_2n_J$, $n_J=n_1$, takes a value near zero. All these incidences lead to Arnold tongue-like spikes pointing downwards in the curve \ref{item:1perc}, marking the failure of the approximation of the orbits based on a \textit{non-resonant} normal form construction. On the other hand, we observe that, for any value of $a_*$ there is a threshold value of the eccentricity $e_{*,s}$, such that, for $e_*<e_{*,s}$ no visible effects of the harmonics associated with mean-motion resonances are visible in the FLI cartography. This implies that the secular models constructed by eliminating all harmonics involving the fast angles of the problem describe with good precision the dynamics in this domain, called, for this reason, the \textit{domain of secular motions}. In physical terms, the domain of secular motions corresponds to initial conditions for which the gravitational perturbation of Jupiter is only felt in the `Laplacian' meaning, i.e., as a mass distributed along a ring coinciding with Jupiter's orbit. The curve \ref{item:1perc} then yields the limit of this domain, which, as found by the FLI cartography, is well distinct from the limit of the Hill domain.\\
The overall situation can therefore be summarized with the identification of four regimes of motion (specified in the FLI chart):
\begin{itemize}
	\item the `crossing orbit regime' (above curve (i));
	\item the `close encounter regime' (between curves (i) and (ii));
	\item the `resonant regime' (between curves (ii) and (iii));
	\item the `secular regime' (below curve (iii)).
\end{itemize}

\section{Conclusions}
\label{sec:conc}
In summary, in the present paper we have proposed a closed-form method for the derivation of secular Hamiltonian models (normal forms) with a small (albeit finite minimum) remainder applicable to the R3BP in the case when the particle's trajectory is exterior to the trajectory of the primary perturber. Also, using this method we were led to the definition of a new heuristic limit separating the motions whose character is `secular', i.e., not affected by short-period effects, from the rest of motions in the R3BP. In particular:

\begin{enumerate}
\item 
Section \ref{sec:theory} develops the formal aspects of the method, which heavily relies on the use of a book-keeping parameter to simultaneously account for all small quantities of the problem as they appear not only in the Hamiltonian and Lie generating functions, but also in the closed-form version of all formulas involved in the Poisson algebra between the Delaunay canonical variables of the problem. A rigorous demonstration of the \textit{consistency} of the method is then given through Propositions \ref{prop:hotstep1}, \ref{prop:normj} and \ref{prop:adjust}, which also estabilish the explicit formulas for the implementation of one iterative step of the closed-form normalization algorithm. 

\item 
Section 3 gives numerical examples of the implementation and precision of the algorithm in the spatial elliptic, as well as in the planar circular R3BP, examining, also numerically, the method's convergence properties. The effect of choosing different truncation orders (in powers of the mass parameter $\mu$ or in the multipole expansion) is discussed, along with several simplifications to the normalization procedure which hold in the circular case. The essentially asymptotic character of the series is established through numerical examples, showing the existence of an optimal number of normalization steps, after which the size of the remainder becomes the minimum possible. 

\item 
A key aspect of the above presented method lies in the possibility to exploit the behavior of the size of the remainder as a function of the number of normalizing steps in order to obtain a clear separation of two well-distinct domains, as also identified by purely numerical (FLI cartography) means: one, called the \textit{domain of secular motions} corresponds to the domain where the harmonics in the Hamiltonian associated with resonant combinations of the fast angles (anomalies) of the problem produce no dynamical effect on the orbits visible at the level of the FLI cartography. From the semi-analytical point of view, this turns to be the domain where a non-resonant construction as the one proposed in section \ref{sec:theory} produces no (nearly-)resonant divisors up to the optimal normalization step. As a consequence, only the angles associated with the motions of the perihelion and of the line of nodes survive in the final normal form. We show numerically how to use the information on the size of the normal form remainder in order to determine semi-analytically the border of the domain of secular motions in the case of the Sun-Jupiter system. We finally give evidence that this border is well distinct from the border of the domains defined either by the Hill stability or by the perihelion crossing criterion. 

\end{enumerate}

\appendix
\vspace{1cm}
\noindent{\Large\bf Appendix}

\section{Computation of Poisson bracket's intermediate derivatives}
\label{sec:appendixder}
Derivatives \eqref{eqn:dfdl}--\eqref{eqn:diotasdH} are computed combining adequately definitions \eqref{eqn:Del}, the polar relationship \eqref{eqn:normR}, including its alternative expression involving the eccentric anomaly $E$
\begin{equation}
\label{eqn:normRE}
\norm{R}=a(1-e\cos E)\;,
\end{equation}
$\norm{r_1}$ via \eqref{eqn:r1} (analogous to \eqref{eqn:normRE}), Kepler's equations 
\begin{equation}
\label{eqn:keplers}
\ell=E-e\sin E\;,\quad\quad M_1=E_1-e_1\sin E_1\;,
\end{equation}
and the trigonometric equalities
\begin{equation}
\label{eqn:trigon}
\cos f=\frac{\cos E-e}{1-e\cos E}\;,\quad\quad\sin f=\frac{\eta\sin E}{1-e\cos E}\;.
\end{equation}
Eq.\eqref{eqn:dfdl} comes from \eqref{eqn:normRE} and \eqref{eqn:normR} by total differentiation with respect to $\ell$:
\begin{equation*}
\frac{\text{d}}{\text{d}\ell}\norm{R}\stackrel{\text{\eqref{eqn:normRE}}}{=}
\frac{\partial\norm{R}}{\partial E}\frac{\partial E}{\partial \ell}=\frac{ae\sin E}{1-e\cos E}\stackrel{\text{\eqref{eqn:normR}}}{=}\frac{\partial\norm{R}}{\partial f}\frac{\partial f}{\partial\ell}=\frac{a\eta^2e\sin f}{(1+e\cos f)^2}\frac{\partial f}{\partial\ell}\;,
\end{equation*}
since $a$, $e$ do not depend on $\ell$, where $\partial E/\partial\ell$ is deduced from the first of \eqref{eqn:keplers} making use of the derivative of inverse functions ($\partial\ell/\partial E\neq0$ is ensured). Thus the result by \eqref{eqn:trigon}.\\
Eqs.\eqref{eqn:dr1dE1}, \eqref{eqn:dE1dM1} are straightforwardly yielded taking respectively ordinary differentiation and the inverse derivative once again of $\text{d} M_1/\text{d} E_1\neq0$ from the second of \eqref{eqn:keplers}:
\begin{equation*}
\frac{\text{d} E_1}{\text{d}M_1}=\frac{1}{1-e_1\cos E_1}=\frac{a_1}{\norm{r_1}}\;.
\end{equation*} 
Now solving for $e$ in \eqref{eqn:Del} and partially differentiating, we immediately have Eqs.\eqref{eqn:deddeltaL} and \eqref{eqn:dedG}, from which Eqs.\eqref{eqn:detaddeltaL}, \eqref{eqn:detadG} as
\begin{equation*}
\frac{\partial\eta}{\partial\delta L}=-\frac{e}{\eta}\frac{\partial e}{\partial\delta L}=-\frac{\eta}{L}\;,\quad\quad\frac{\partial\eta}{\partial G}=-\frac{e}{\eta}\frac{\partial e}{\partial G}=\frac1L\;.
\end{equation*}
The true anomaly derivatives with respect to the actions are slightly more elaborated. Employing \eqref{eqn:trigon},
\begin{equation*}
-\sin f\frac{\partial f}{\partial\delta L}=\frac{\partial}{\partial\delta L}\cos f=\frac{\partial}{\partial e}\left(\frac{\cos E-e}{1-e\cos E}\right)\frac{\partial e}{\partial\delta L}+\frac{\partial}{\partial E}\left(\frac{\cos E-e}{1-e\cos E}\right)\frac{\partial E}{\partial\delta L}\;,
\end{equation*}
that leads upon simplifications to
\begin{equation*}
\frac{\partial f}{\partial\delta L}=\frac{\sin f}{eL}+\frac{1+e\cos f}{\eta}\frac{\partial E}{\partial\delta L}\;;
\end{equation*}
finally we explicit $\partial E/\partial\delta L$ exploiting 
the corresponding Kepler equation \eqref{eqn:keplers} and the inter-independence $\ell,\delta L$ by conjugacy:
\begin{equation*}
0=\frac{\text{d}}{\text{d}\delta L}(E-e\sin E)=\frac{\partial E}{\partial\delta L}-\frac{\partial e}{\partial\delta L}\sin E-e\cos E\frac{\partial E}{\partial\delta L}\;\implies\;\frac{\partial E}{\partial\delta L}=\frac{\eta\sin f}{eL}\;,
\end{equation*} 
thereby Eq.\eqref{eqn:dfddeltaL}.\\
The relation for $\partial f/\partial G$ is achieved precisely in the same manner, so one finds out
\begin{equation*}
\frac{\partial f}{\partial G}=-\frac{\sin f}{\eta eL}+\frac{1+e\cos f}{\eta}\frac{\partial E}{\partial G}\;,\quad\quad\frac{\partial E}{\partial G}=-\frac{\sin f}{eL}\;,
\end{equation*}
that is Eq.\eqref{eqn:dfdG}.\\
Finally, derivatives \eqref{eqn:diotacdG}, \eqref{eqn:diotacdH} involving $\iota_c=\cos i$ easily follow again by partial differentiation in \eqref{eqn:Del} with respect to $G$ and $H$ respectively; while for those containing $\iota_s=\sin i$ we can rely, for example, to the identity $\sin^2 i+\cos^2 i=1$:
\begin{equation*}
0=2\sin i\frac{\partial\iota_s}{\partial G}+2\cos i\frac{\partial\iota_c}{\partial G}\;
\end{equation*}
and consequently Eq.\eqref{eqn:diotasdG} provided $\sin i\neq 0$, as well as Eq.\eqref{eqn:diotasdH} repeating the same argument with the variable $H$.

\section{Example of normalization for a $\boldsymbol{\mu^2}$ quadrupolar expansion}
\label{sec:appendixex}
Consider the following toy model Hamiltonian with $k_{\mu}=k_{\text{mp}}=\nu=2$, $\nu_1=1$, according to conventions introduced in \S\ref{subsubsec:prelim}:
\begin{equation*}
\mathscr{H}^{(0)}=\mathscr{Z}_0+\mathscr{R}_{2,2}^{(0)}+\mathscr{R}_{2,3}^{(0)}+\mathscr{R}_{2,4}^{(0)}\;,
\end{equation*}
where
\begin{multline*}
\mathscr{R}^{(0)}_{2,2}=\sigma ^2 \Bigg(-\frac{3 a_1^3 \mathcal{G}^4 \mu  m_0^4 \iota _c^2 \cos \left(2 \left(E_1-f-g-h\right)\right)}{16 L_*^6 \norm{r_1}}\\
-\frac{3 a_1^3 \mathcal{G}^4 \mu  m_0^4 \iota _c^2 \cos \left(2
	\left(E_1+f+g-h\right)\right)}{16 L_*^6 \norm{r_1}}-\frac{3 a_1^3 \mathcal{G}^4 \mu  m_0^4 \iota _c \cos \left(2 \left(E_1-f-g-h\right)\right)}{8 L_*^6 \norm{r_1}}\\
+\frac{3 a_1^3 \mathcal{G}^4 \mu  m_0^4 \iota _c
	\cos \left(2 \left(E_1+f+g-h\right)\right)}{8 L_*^6 \norm{r_1}}+\frac{3 a_1^3 \mathcal{G}^4 \mu  m_0^4 \iota _c^2 \cos \left(2 \left(E_1-h\right)\right)}{8 L_*^6 \norm{r_1}}\\
+\frac{3 a_1^3 \mathcal{G}^4 \mu  m_0^4
	\iota _c^2 \cos (2 (f+g))}{8 L_*^6 \norm{r_1}}-\frac{3 a_1^3 \mathcal{G}^4 \mu  m_0^4 \iota _c^2}{8 L_*^6 \norm{r_1}}
-\frac{3 a_1^3 \mathcal{G}^4 \mu  m_0^4 \cos \left(2 \left(E_1-f-g-h\right)\right)}{16 L_*^6 \norm{r_1}}\\
-\frac{3
	a_1^3 \mathcal{G}^4 \mu  m_0^4 \cos \left(2 \left(E_1+f+g-h\right)\right)}{16 L_*^6 \norm{r_1}}-\frac{3 a_1^3 \mathcal{G}^4 \mu  m_0^4 \cos \left(2 \left(E_1-h\right)\right)}{8 L_*^6 \norm{r_1}}\\
-\frac{3 a_1^3 \mathcal{G}^4
	\mu  m_0^4 \cos (2 (f+g))}{8 L_*^6 \norm{r_1}}+\frac{a_1^3 \mathcal{G}^4 \mu  m_0^4}{8 L_*^6 \norm{r_1}}
-\frac{3 a_1 \delta L^2 \mathcal{G}^2 m_0^2}{2 L_*^4 \norm{r_1}}
-\frac{a_1 \mathcal{G}^2 \mu  m_0^2}{L_*^2 \norm{r_1}}\Bigg)\;.
\end{multline*}
The first step $j=1$ of the method aims precisely at normalizing $\mathscr{R}_{2,2}^{(0)}$ via \eqref{eqn:firsthom} solved by
\begin{multline*}
\chi_{2}^{(1)}=\sigma ^3\Bigg(\frac{3 \mathcal{G}^4 \mu  a_1^2 \iota _c^2 \phi _1 n_*^2 m_0^4}{8 n_1 L_*^6 \left(n_1^2-n_*^2\right)}-\frac{\mathcal{G}^4 \mu  a_1^2 \phi _1 n_*^2 m_0^4}{8 n_1 L_*^6 \left(n_1^2-n_*^2\right)}\\
-\frac{3 \mathcal{G}^4 \mu  a_1^2 n_1 \iota _c^2 \phi _1 m_0^4}{8 L_*^6 \left(n_1^2-n_*^2\right)}+\frac{\mathcal{G}^4 \mu  a_1^2 n_1 \phi _1 m_0^4}{8 L_*^6 \left(n_1^2-n_*^2\right)}+\frac{\mathcal{G}^2 \mu  \phi _1 n_*^2 m_0^2}{n_1 L_*^2 \left(n_1^2-n_*^2\right)}\\
+\frac{3 \mathcal{G}^2 \delta L^2 \phi _1 n_*^2 m_0^2}{2 n_1 L_*^4 \left(n_1^2-n_*^2\right)}-\frac{\mathcal{G}^2 \mu  n_1 \phi _1 m_0^2}{L_*^2 \left(n_1^2-n_*^2\right)}-\frac{3 \mathcal{G}^2 \delta L^2 n_1 \phi _1 m_0^2}{2 L_*^4 \left(n_1^2-n_*^2\right)}\Bigg)\\
+ \sigma ^2\Bigg(-\frac{3 \mathcal{G}^4 \mu  \sin \left(2 \left(E_1-h\right)\right) a_1^2 \iota _c^2 n_*^2 m_0^4}{16 n_1 L_*^6 \left(n_1^2-n_*^2\right)}+\frac{3 \mathcal{G}^4 \mu  \sin \left(2 \left(E_1-h\right)\right) a_1^2 n_*^2 m_0^4}{16 n_1 L_*^6 \left(n_1^2-n_*^2\right)}\\
+\frac{3 \mathcal{G}^4 \mu  \sin (2 (f+g)) a_1^2 n_* m_0^4}{16 L_*^6 \left(n_1^2-n_*^2\right)}-\frac{3 \mathcal{G}^4 \mu  \sin \left(2 \left(-f-g-h+E_1\right)\right) a_1^2 n_* m_0^4}{32 L_*^6 \left(n_1^2-n_*^2\right)}\\
+\frac{3 \mathcal{G}^4 \mu  \sin \left(2 \left(f+g-h+E_1\right)\right) a_1^2 n_* m_0^4}{32 L_*^6 \left(n_1^2-n_*^2\right)}-\frac{3 \mathcal{G}^4 \mu  \sin (2 (f+g)) a_1^2 \iota _c^2 n_* m_0^4}{16 L_*^6 \left(n_1^2-n_*^2\right)}\\
-\frac{3 \mathcal{G}^4 \mu  \sin \left(2 \left(-f-g-h+E_1\right)\right) a_1^2 \iota _c^2 n_* m_0^4}{32 L_*^6 \left(n_1^2-n_*^2\right)}+\frac{3 \mathcal{G}^4 \mu  \sin \left(2 \left(f+g-h+E_1\right)\right) a_1^2 \iota _c^2 n_* m_0^4}{32 L_*^6 \left(n_1^2-n_*^2\right)}\\
-\frac{3 \mathcal{G}^4 \mu  \sin \left(2 \left(-f-g-h+E_1\right)\right) a_1^2 \iota _c n_* m_0^4}{16 L_*^6 \left(n_1^2-n_*^2\right)}-\frac{3 \mathcal{G}^4 \mu  \sin \left(2 \left(f+g-h+E_1\right)\right) a_1^2 \iota _c n_* m_0^4}{16 L_*^6 \left(n_1^2-n_*^2\right)}\\
+\frac{3 \mathcal{G}^4 \mu  \sin \left(2 \left(E_1-h\right)\right) a_1^2 n_1 \iota _c^2 m_0^4}{16 L_*^6 \left(n_1^2-n_*^2\right)}-\frac{3 \mathcal{G}^4 \mu  \sin \left(2 \left(-f-g-h+E_1\right)\right) a_1^2 n_1 \iota _c^2 m_0^4}{32 L_*^6 \left(n_1^2-n_*^2\right)}\\
-\frac{3 \mathcal{G}^4 \mu  \sin \left(2 \left(f+g-h+E_1\right)\right) a_1^2 n_1 \iota _c^2 m_0^4}{32 L_*^6 \left(n_1^2-n_*^2\right)}-\frac{3 \mathcal{G}^4 \mu  \sin \left(2 \left(E_1-h\right)\right) a_1^2 n_1 m_0^4}{16 L_*^6 \left(n_1^2-n_*^2\right)}\\
-\frac{3 \mathcal{G}^4 \mu  \sin \left(2 \left(-f-g-h+E_1\right)\right) a_1^2 n_1 m_0^4}{32 L_*^6 \left(n_1^2-n_*^2\right)}-\frac{3 \mathcal{G}^4 \mu  \sin \left(2 \left(f+g-h+E_1\right)\right) a_1^2 n_1 m_0^4}{32 L_*^6 \left(n_1^2-n_*^2\right)}\\
-\frac{3 \mathcal{G}^4 \mu  \sin \left(2 \left(-f-g-h+E_1\right)\right) a_1^2 n_1 \iota _c m_0^4}{16 L_*^6 \left(n_1^2-n_*^2\right)}+\frac{3 \mathcal{G}^4 \mu  \sin \left(2 \left(f+g-h+E_1\right)\right) a_1^2 n_1 \iota _c m_0^4}{16 L_*^6 \left(n_1^2-n_*^2\right)}\\
-\frac{3 \mathcal{G}^4 \mu  \sin (2 (f+g)) a_1^2 n_1^2 m_0^4}{16 L_*^6 n_*\left(n_1^2-n_*^2\right)}+\frac{3 \mathcal{G}^4 \mu  \sin (2 (f+g)) a_1^2 n_1^2 \iota _c^2 m_0^4}{16 L_*^6 n_*\left(n_1^2-n_*^2\right)}\Bigg)\;,
\end{multline*}
so that the new truncated Hamiltonian becomes
\begin{equation*}
\mathscr{H}^{(1)}=\mathscr{Z}_0+\mathscr{Z}^{(1)}_2+\mathscr{R}^{(1)}_{3,3}+\mathscr{R}^{(1)}_{3,4}\;,
\end{equation*}
with
\begin{equation*}
\mathscr{Z}_{2}^{(1)}=\sigma ^2 \left(-\frac{3 a_1^2 \mathcal{G}^4 \mu  m_0^4 \iota _c^2}{8 L_*^6}+\frac{a_1^2 \mathcal{G}^4 \mu  m_0^4}{8 L_*^6}-\frac{3 \delta L^2 \mathcal{G}^2 m_0^2}{2 L_*^4}-\frac{\mathcal{G}^2 \mu  m_0^2}{L_*^2}\right)
\end{equation*}
and
\begin{multline*}
\mathscr{R}^{(1)}_{3,3}=\sigma ^3 \Bigg(-\frac{3 e \mathcal{G}^4 \mu  \cos \left(f+2 g+2 h-2 E_1\right) a_1^3 \iota _c^2 n_* m_0^4}{8 \eta ^3 \norm{r_1} L_*^6 \left(2 n_1-2 n_*\right)}\\
-\frac{3 e \mathcal{G}^4 \mu  \cos \left(3 f+2 g+2 h-2 E_1\right)
	a_1^3 \iota _c^2 n_* m_0^4}{8 \eta ^3 \norm{r_1} L_*^6 \left(2 n_1-2 n_*\right)}+\frac{3 \mathcal{G}^4 \mu  \cos \left(2 f+2 g+2 h-3 E_1\right) a_1^3 e_1 \iota _c^2 n_* m_0^4}{16 \norm{r_1} L_*^6 \left(2 n_1-2
	n_*\right)}\\
+\frac{3 \mathcal{G}^4 \mu  \cos \left(2 f+2 g+2 h-E_1\right) a_1^3 e_1 \iota _c^2 n_* m_0^4}{16 \norm{r_1} L_*^6 \left(2 n_1-2 n_*\right)}-\frac{3 e \mathcal{G}^4 \mu  \cos \left(f+2 g+2 h-2 E_1\right) a_1^3
	\iota _c n_* m_0^4}{4 \eta ^3 \norm{r_1} L_*^6 \left(2 n_1-2 n_*\right)}\\
-\frac{3 e \mathcal{G}^4 \mu  \cos \left(3 f+2 g+2 h-2 E_1\right) a_1^3 \iota _c n_* m_0^4}{4 \eta ^3 \norm{r_1} L_*^6 \left(2 n_1-2
	n_*\right)}+\frac{3 \mathcal{G}^4 \mu  \cos \left(2 f+2 g+2 h-3 E_1\right) a_1^3 e_1 \iota _c n_* m_0^4}{8 \norm{r_1} L_*^6 \left(2 n_1-2 n_*\right)}\\
+\frac{3 \mathcal{G}^4 \mu  \cos \left(2 f+2 g+2 h-E_1\right) a_1^3 e_1
	\iota _c n_* m_0^4}{8 \norm{r_1} L_*^6 \left(2 n_1-2 n_*\right)}-\frac{3 e \mathcal{G}^4 \mu  \cos \left(f+2 g+2 h-2 E_1\right) a_1^3 n_* m_0^4}{8 \eta ^3 \norm{r_1} L_*^6 \left(2 n_1-2 n_*\right)}\\
-\frac{3 e \mathcal{G}^4 \mu  \cos
	\left(3 f+2 g+2 h-2 E_1\right) a_1^3 n_* m_0^4}{8 \eta ^3 \norm{r_1} L_*^6 \left(2 n_1-2 n_*\right)}+\frac{3 \mathcal{G}^4 \mu  \cos \left(2 f+2 g+2 h-3 E_1\right) a_1^3 e_1 n_* m_0^4}{16 \norm{r_1} L_*^6 \left(2 n_1-2
	n_*\right)}\\
+\frac{3 \mathcal{G}^4 \mu  \cos \left(2 f+2 g+2 h-E_1\right) a_1^3 e_1 n_* m_0^4}{16 \norm{r_1} L_*^6 \left(2 n_1-2 n_*\right)}+\frac{3 e \mathcal{G}^4 \mu  \cos \left(f+2 g-2 h+2 E_1\right) a_1^3 \iota _c^2 n_*
	m_0^4}{8 \eta ^3 \norm{r_1} L_*^6 \left(2 n_1+2 n_*\right)}\\
+\frac{3 e \mathcal{G}^4 \mu  \cos \left(3 f+2 g-2 h+2 E_1\right) a_1^3 \iota _c^2 n_* m_0^4}{8 \eta ^3 \norm{r_1} L_*^6 \left(2 n_1+2 n_*\right)}-\frac{3 \mathcal{G}^4 \mu
	\cos \left(2 f+2 g-2 h+E_1\right) a_1^3 e_1 \iota _c^2 n_* m_0^4}{16 \norm{r_1} L_*^6 \left(2 n_1+2 n_*\right)}\\
-\frac{3 \mathcal{G}^4 \mu  \cos \left(2 f+2 g-2 h+3 E_1\right) a_1^3 e_1 \iota _c^2 n_* m_0^4}{16 \norm{r_1}
	L_*^6 \left(2 n_1+2 n_*\right)}-\frac{3 e \mathcal{G}^4 \mu  \cos \left(f+2 g-2 h+2 E_1\right) a_1^3 \iota _c n_* m_0^4}{4 \eta ^3 \norm{r_1} L_*^6 \left(2 n_1+2 n_*\right)}\\
-\frac{3 e \mathcal{G}^4 \mu  \cos \left(3 f+2 g-2
	h+2 E_1\right) a_1^3 \iota _c n_* m_0^4}{4 \eta ^3 \norm{r_1} L_*^6 \left(2 n_1+2 n_*\right)}+\frac{3 \mathcal{G}^4 \mu  \cos \left(2 f+2 g-2 h+E_1\right) a_1^3 e_1 \iota _c n_* m_0^4}{8 \norm{r_1} L_*^6 \left(2 n_1+2
	n_*\right)}\\
+\frac{3 \mathcal{G}^4 \mu  \cos \left(2 f+2 g-2 h+3 E_1\right) a_1^3 e_1 \iota _c n_* m_0^4}{8 \norm{r_1} L_*^6 \left(2 n_1+2 n_*\right)}+\frac{3 e \mathcal{G}^4 \mu  \cos \left(f+2 g-2 h+2 E_1\right) a_1^3 n_*
	m_0^4}{8 \eta ^3 \norm{r_1} L_*^6 \left(2 n_1+2 n_*\right)}\\
+\frac{3 e \mathcal{G}^4 \mu  \cos \left(3 f+2 g-2 h+2 E_1\right) a_1^3 n_* m_0^4}{8 \eta ^3 \norm{r_1} L_*^6 \left(2 n_1+2 n_*\right)}-\frac{3 \mathcal{G}^4 \mu  \cos
	\left(2 f+2 g-2 h+E_1\right) a_1^3 e_1 n_* m_0^4}{16 \norm{r_1} L_*^6 \left(2 n_1+2 n_*\right)}\\
-\frac{3 \mathcal{G}^4 \mu  \cos \left(2 f+2 g-2 h+3 E_1\right) a_1^3 e_1 n_* m_0^4}{16 \norm{r_1} L_*^6 \left(2 n_1+2
	n_*\right)}-\frac{9 e \mathcal{G}^4 \mu  \cos (f) a_1^3 \iota _c^2 m_0^4}{8 \norm{r_1} L_*^6}\\
+\frac{9 e \mathcal{G}^4 \mu  \cos (f+2 g) a_1^3 \iota _c^2 m_0^4}{16 \norm{r_1} L_*^6}-\frac{3 e \mathcal{G}^4 \mu  \cos (f+2 g) a_1^3 \iota _c^2
	m_0^4}{8 \eta ^3 \norm{r_1} L_*^6}\\
+\frac{9 e \mathcal{G}^4 \mu  \cos (3 f+2 g) a_1^3 \iota _c^2 m_0^4}{16 \norm{r_1} L_*^6}-\frac{3 e \mathcal{G}^4 \mu  \cos (3 f+2 g) a_1^3 \iota _c^2 m_0^4}{8 \eta ^3 \norm{r_1} L_*^6}\\
+\frac{9 e \mathcal{G}^4 \mu 
	\cos \left(f+2 h-2 E_1\right) a_1^3 \iota _c^2 m_0^4}{16 \norm{r_1} L_*^6}-\frac{9 e \mathcal{G}^4 \mu  \cos \left(f+2 g+2 h-2 E_1\right) a_1^3 \iota _c^2 m_0^4}{32 \norm{r_1} L_*^6}\\
-\frac{9 e \mathcal{G}^4 \mu  \cos \left(3 f+2
	g+2 h-2 E_1\right) a_1^3 \iota _c^2 m_0^4}{32 \norm{r_1} L_*^6}+\frac{9 e \mathcal{G}^4 \mu  \cos \left(f-2 h+2 E_1\right) a_1^3 \iota _c^2 m_0^4}{16 \norm{r_1} L_*^6}\\
-\frac{9 e \mathcal{G}^4 \mu  \cos \left(f+2 g-2 h+2 E_1\right)
	a_1^3 \iota _c^2 m_0^4}{32 \norm{r_1} L_*^6}-\frac{9 e \mathcal{G}^4 \mu  \cos \left(3 f+2 g-2 h+2 E_1\right) a_1^3 \iota _c^2 m_0^4}{32 \norm{r_1} L_*^6}\\
\allowdisplaybreaks
-\frac{3 \mathcal{G}^4 \mu  \cos \left(2 h-3 E_1\right) a_1^3 e_1 \iota _c^2
	m_0^4}{16 \norm{r_1} L_*^6}+\frac{3 \mathcal{G}^4 \mu  \cos \left(2 f+2 g+2 h-3 E_1\right) a_1^3 e_1 \iota _c^2 m_0^4}{32 \norm{r_1} L_*^6}\\
-\frac{3 \mathcal{G}^4 \mu  \cos \left(2 f+2 g-E_1\right) a_1^3 e_1 \iota _c^2 m_0^4}{8 \norm{r_1}
	L_*^6}-\frac{15 \mathcal{G}^4 \mu  \cos \left(2 h-E_1\right) a_1^3 e_1 \iota _c^2 m_0^4}{16 \norm{r_1} L_*^6}\\
+\frac{15 \mathcal{G}^4 \mu  \cos \left(2 f+2 g+2 h-E_1\right) a_1^3 e_1 \iota _c^2 m_0^4}{32 \norm{r_1} L_*^6}+\frac{9
	\mathcal{G}^4 \mu  \cos \left(E_1\right) a_1^3 e_1 \iota _c^2 m_0^4}{8 \norm{r_1} L_*^6}\\
-\frac{3 \mathcal{G}^4 \mu  \cos \left(2 f+2 g+E_1\right) a_1^3 e_1 \iota _c^2 m_0^4}{8 \norm{r_1} L_*^6}+\frac{15 \mathcal{G}^4 \mu  \cos \left(2 f+2
	g-2 h+E_1\right) a_1^3 e_1 \iota _c^2 m_0^4}{32 \norm{r_1} L_*^6}\\
+\frac{3 \mathcal{G}^4 \mu  \cos \left(2 f+2 g-2 h+3 E_1\right) a_1^3 e_1 \iota _c^2 m_0^4}{32 \norm{r_1} L_*^6}-\frac{9 e \mathcal{G}^4 \mu  \cos \left(f+2 g+2 h-2
	E_1\right) a_1^3 \iota _c m_0^4}{16 \norm{r_1} L_*^6}\\
-\frac{9 e \mathcal{G}^4 \mu  \cos \left(3 f+2 g+2 h-2 E_1\right) a_1^3 \iota _c m_0^4}{16 \norm{r_1} L_*^6}+\frac{9 e \mathcal{G}^4 \mu  \cos \left(f+2 g-2 h+2 E_1\right) a_1^3
	\iota _c m_0^4}{16 \norm{r_1} L_*^6}\\
+\frac{9 e \mathcal{G}^4 \mu  \cos \left(3 f+2 g-2 h+2 E_1\right) a_1^3 \iota _c m_0^4}{16 \norm{r_1} L_*^6}+\frac{3 \mathcal{G}^4 \mu  \cos \left(2 f+2 g+2 h-3 E_1\right) a_1^3 e_1 \iota _c
	m_0^4}{16 \norm{r_1} L_*^6}\\
+\frac{15 \mathcal{G}^4 \mu  \cos \left(2 f+2 g+2 h-E_1\right) a_1^3 e_1 \iota _c m_0^4}{16 \norm{r_1} L_*^6}-\frac{15 \mathcal{G}^4 \mu  \cos \left(2 f+2 g-2 h+E_1\right) a_1^3 e_1 \iota _c m_0^4}{16 \norm{r_1}
	L_*^6}\\
-\frac{3 \mathcal{G}^4 \mu  \cos \left(2 f+2 g-2 h+3 E_1\right) a_1^3 e_1 \iota _c m_0^4}{16 \norm{r_1} L_*^6}+\frac{3 e \mathcal{G}^4 \mu  \cos (f) a_1^3 m_0^4}{8 \norm{r_1} L_*^6}\\
-\frac{9 e \mathcal{G}^4 \mu  \cos (f+2 g) a_1^3
	m_0^4}{16 \norm{r_1} L_*^6}+\frac{3 e \mathcal{G}^4 \mu  \cos (f+2 g) a_1^3 m_0^4}{8 \eta ^3 \norm{r_1} L_*^6}\\
-\frac{9 e \mathcal{G}^4 \mu  \cos (3 f+2 g) a_1^3 m_0^4}{16 \norm{r_1} L_*^6}+\frac{3 e \mathcal{G}^4 \mu  \cos (3 f+2 g) a_1^3 m_0^4}{8
	\eta ^3 \norm{r_1} L_*^6}\\
-\frac{9 e \mathcal{G}^4 \mu  \cos \left(f+2 h-2 E_1\right) a_1^3 m_0^4}{16 \norm{r_1} L_*^6}-\frac{9 e \mathcal{G}^4 \mu  \cos \left(f+2 g+2 h-2 E_1\right) a_1^3 m_0^4}{32 \norm{r_1} L_*^6}\\
-\frac{9 e \mathcal{G}^4 \mu 
	\cos \left(3 f+2 g+2 h-2 E_1\right) a_1^3 m_0^4}{32 \norm{r_1} L_*^6}-\frac{9 e \mathcal{G}^4 \mu  \cos \left(f-2 h+2 E_1\right) a_1^3 m_0^4}{16 \norm{r_1} L_*^6}\\
-\frac{9 e \mathcal{G}^4 \mu  \cos \left(f+2 g-2 h+2 E_1\right) a_1^3
	m_0^4}{32 \norm{r_1} L_*^6}-\frac{9 e \mathcal{G}^4 \mu  \cos \left(3 f+2 g-2 h+2 E_1\right) a_1^3 m_0^4}{32 \norm{r_1} L_*^6}\\
+\frac{3 \mathcal{G}^4 \mu  \cos \left(2 h-3 E_1\right) a_1^3 e_1 m_0^4}{16 \norm{r_1} L_*^6}+\frac{3 \mathcal{G}^4 \mu 
	\cos \left(2 f+2 g+2 h-3 E_1\right) a_1^3 e_1 m_0^4}{32 \norm{r_1} L_*^6}\\
+\frac{3 \mathcal{G}^4 \mu  \cos \left(2 f+2 g-E_1\right) a_1^3 e_1 m_0^4}{8 \norm{r_1} L_*^6}+\frac{15 \mathcal{G}^4 \mu  \cos \left(2 h-E_1\right) a_1^3 e_1
	m_0^4}{16 \norm{r_1} L_*^6}\\
+\frac{15 \mathcal{G}^4 \mu  \cos \left(2 f+2 g+2 h-E_1\right) a_1^3 e_1 m_0^4}{32 \norm{r_1} L_*^6}-\frac{3 \mathcal{G}^4 \mu  \cos \left(E_1\right) a_1^3 e_1 m_0^4}{8 \norm{r_1} L_*^6}\\
+\frac{3 \mathcal{G}^4 \mu  \cos
	\left(2 f+2 g+E_1\right) a_1^3 e_1 m_0^4}{8 \norm{r_1} L_*^6}+\frac{15 \mathcal{G}^4 \mu  \cos \left(2 f+2 g-2 h+E_1\right) a_1^3 e_1 m_0^4}{32 \norm{r_1} L_*^6}\\
+\frac{3 \mathcal{G}^4 \mu  \cos \left(2 f+2 g-2 h+3 E_1\right) a_1^3
	e_1 m_0^4}{32 \norm{r_1} L_*^6}-\frac{e \mathcal{G}^2 \mu  \cos (f) a_1 m_0^2}{\norm{r_1} L_*^2}\\
+\frac{\mathcal{G}^2 \mu  \cos \left(E_1\right) a_1 e_1 m_0^2}{\norm{r_1} L_*^2}+\frac{3 \mathcal{G}^2 \delta L^2 \cos \left(E_1\right) a_1 e_1
	m_0^2}{2 \norm{r_1} L_*^4}\Bigg)\;.
\end{multline*}
Next, we move on with the second and last iteration $j=2$ targeted to $\mathscr{R}^{(1)}_{3,3}$:
\begin{equation*}
\mathscr{H}^{(2)}=\mathscr{Z}_0+\mathscr{Z}^{(1)}_{2}+\mathscr{Z}^{(2)}_{3}+\mathscr{R}_{4,4}^{(2)}\;,
\end{equation*}
in which $\chi_{3}^{(2)}$ is omitted for brevity and
\begin{equation*}
\mathscr{Z}^{(2)}_{3}=0
\end{equation*}
as expected, being $\mathscr{R}^{(1)}_{3,3}$ solely made up of harmonics containing fast angles.\\

\noindent\textbf{Acknowledgements.} C.E. was partially supported by the MIUR-PRIN 20178CJA2B New Frontiers of Celestial Mechanics: Theory and Applications.\\

\printbibliography[heading=bibintoc]

\end{document}